\documentclass[10pt]{article}
\usepackage{amsmath,amsfonts,amssymb,amsthm,hyperref,enumerate,multicol,bm}
\usepackage{indentfirst,lettrine}
\usepackage{caption}
\usepackage{relsize}
\usepackage{eqnarray}
\usepackage{cancel}
\usepackage{array}
\usepackage{longtable}
\usepackage{cite}
\usepackage{tablefootnote}

\usepackage{titlesec}
\titlespacing*{\section}{0pt}{*1.5}{*1.5}
\titlespacing*{\subsection}{0pt}{*1.2}{*1.2}

\usepackage{enumitem}
\setlist{itemsep=4pt, topsep=2pt, parsep=0pt, partopsep=0pt}
\makeatletter

\usepackage{etoolbox}
\makeatletter
\patchcmd{\@thm}{\thm@preskip}{0.1\baselineskip}{}{}
\patchcmd{\@thm}{\thm@postskip}{0.1\baselineskip}{}{}
\makeatother

\AtBeginEnvironment{tabular}{\small}

\DeclareFontFamily{U}{mathc}{}
\DeclareFontShape{U}{mathc}{m}{it}
{<->s*[1.03] mathc10}{}

\usepackage{comment}
\DeclareMathOperator{\supp}{supp}
\usepackage{mathtools}
\usepackage{graphicx}
\usepackage{ulem}
\usepackage{subfigure}
\usepackage{xcolor}
\usepackage{color, soul}
\usepackage{mathrsfs}
\usepackage{float}
\usepackage{longtable}
\DeclareMathAlphabet{\mathpzc}{OT1}{pzc}{m}{it}
\DeclarePairedDelimiter\ceil{\lceil}{\rceil}

\allowdisplaybreaks
\numberwithin{equation}{section}
\numberwithin{equation}{section}

\usepackage{upgreek}
\newtheorem{thm}{Theorem}[section]
\newcommand{\etal}{\textit{et al.}}

\newtheorem{cor}{Corollary}[section]
\newtheorem{lemma}{Lemma}[section]
\newtheorem{example}{Example}[section]
\newtheorem{remark}{Remark}[section]
\newcommand{\ie}{\textit{i.e.,}}

\begin{document}

\title{Linear codes over a mixed-alphabet ring and their Gray images with applications to projective and locally repairable codes}

\author{\normalsize Leijo Jose{\footnote{Email address:~\urlstyle{same}\href{mailto:leijoj@iiitd.ac.in}{leijoj@iiitd.ac.in}}} ,~~Lavanya G.{\footnote{Email address:~\urlstyle{same}\href{mailto:lavanyag@iiitd.ac.in}{lavanyag@iiitd.ac.in}}}~ and ~Anuradha Sharma{\footnote{Corresponding Author, Email address:~\urlstyle{same}\href{mailto:anuradha@iiitd.ac.in}{anuradha@iiitd.ac.in}} } \\
 {\normalsize Department of  Mathematics, IIIT-Delhi}\\{\normalsize New Delhi 110020, India}}
\date{}
\maketitle
\date{}
	\maketitle
\begin{abstract}\label{Abstract}
Let $m \geq 2$ be an integer, and let $\mathbb{F}_q$ be the finite field of prime power order $q.$  Let $\mathcal{R}=\frac{\mathbb{F}_q[u]}{\langle u^2 \rangle}\times \mathbb{F}_q$ be the mixed-alphabet ring, where $\frac{\mathbb{F}_q[u]}{\langle u^2 \rangle}$ is the quasi-Galois ring with maximal ideal $\langle u\rangle$ of nilpotency index $2$ and residue field $\mathbb{F}_q.$  
In this paper, we construct four infinite families of linear codes over the ring $\frac{\mathbb{F}_q[u]}{\langle u^2 \rangle}$ whose defining sets  are certain non-empty subsets of $\mathcal{R}^m$  associated with three simplicial complexes of $\mathbb{F}_q^m,$ each possessing a single maximal element. We explicitly determine the parameters and Lee weight distributions of these codes. We also study their Gray images and identify several infinite families of few-weight codes over $\mathbb{F}_q,$  as well as an infinite family of minimal, near-Griesmer, and distance-optimal codes over $\mathbb{F}_q.$ We also observe that their Gray images are self-orthogonal codes for $q=2$ or $3.$   Furthermore, for any subset $\mathcal{D} \subseteq \mathcal{R}^m,$  we determine a spanning matrix of a linear code over $\frac{\mathbb{F}_q[u]}{\langle u^2 \rangle}$ with  defining set $\mathcal{D},$ as well as that of its Gray image. Leveraging this result, we provide two constructions of infinite families of projective few-weight codes over $\mathbb{F}_q$ with new parameters, and observe that these codes are self-orthogonal for $q=2$ or $3.$ 
Additionally, we study the duals of these projective codes and explicitly determine  their parameters. As a result, we obtain two infinite families of binary distance-optimal projective codes. Apart from this, we construct an infinite family of quaternary projective $3$-weight codes whose non-zero Hamming weights sum to $\frac{9}{4}$ times the code length, which give rise to strongly walk-regular graphs.
 As an application of our newly constructed minimal codes over $\mathbb{F}_q,$ we examine the  minimal access structures of Massey’s secret sharing schemes based on their duals and determine the number of dictatorial participants in these schemes. Finally, we investigate the locality properties of our newly constructed projective codes and show that these codes have locality either $2$ or $3.$ As a consequence, we obtain four infinite families of $q$-ary locally repairable codes (LRCs) with locality $2,$  and two infinite families of binary LRCs with locality $3.$
\end{abstract}

\section{Introduction}

A fundamental objective in coding theory is the construction of linear codes that are optimal with respect to established bounds, along with the precise determination of their Hamming weight distributions. Notably, \textit{distance-optimal codes} maximize error detection and correction capabilities. The \textit{Hamming weight distribution} (or equivalently, the \textit{Hamming weight enumerator}) of a code provides crucial information about structure of the code and plays a key role in assessing its error performance. A prevalent approach to achieve this objective involves studying codes over rings equipped with suitable metrics and examining their Gray images.

Hammons \etal~\cite{Hammons1994} were the first to study linear codes over the ring $\mathbb{Z}_4$ of integers modulo $4,$ with respect to the Lee weight function on $\mathbb{Z}_4,$ and showed that many important binary non-linear codes can be viewed as the Gray images of such codes. This seminal contribution catalyzed extensive research on codes over finite commutative chain rings (see \cite{Greferath1999,Jose2024} and references therein). Subsequently, Rif\'a  and Pujol \cite{Rifa1997} generalized the study of linear codes over  rings to the mixed-alphabet ring $\mathbb{Z}_4 \times \mathbb{Z}_2,$ interpreting these as abelian translation-invariant propelinear codes. Further, Aydogdu \etal~\cite{Siap2016} explored cyclic and constacyclic codes over the mixed-alphabet ring $\frac{\mathbb{F}_2[u]}{\langle u^2 \rangle} \times \mathbb{F}_2,$ deriving binary codes with good parameters as Gray images of cyclic codes over this ring. In a subsequent study, Aydogdu \etal~\cite{Siap2017} examined linear and cyclic codes over the mixed-alphabet ring $\frac{\mathbb{F}_2[u]}{\langle u^3 \rangle} \times \mathbb{F}_2,$ yielding numerous optimal binary linear codes as Gray images of cyclic codes over the specified ring. Since these developments, codes over mixed-alphabets of finite commutative chain rings have attracted considerable attention \cite{Bajalan2023,Borges2018,Jose2025}.

\textit{Projective codes with few weights} and \textit{self-orthogonal codes} constitute two  significant classes of linear codes. Projective codes with few weights have applications in the construction of strongly regular graphs and association schemes \cite{Cheng2022,Calderbank1984} and exhibit strong connections with combinatorial designs \cite{Olmez2018}. Meanwhile, self-orthogonal codes play a fundamental role in the construction of pure quantum stabilizer codes and are closely related to the theory of unimodular lattices and modular forms \cite{Ling2010,Rains1998}.
Another important family of codes is that of \textit{locally repairable codes} (LRCs). These are erasure-correcting codes used in distributed storage systems, designed to recover the information stored on a failed node by accessing only a small number of other nodes, unlike classical MDS codes \cite{Gopalan2012, Luo2022}.

In a parallel direction, \textit{minimal codes} have garnered substantial interest due to their applications in Cryptography  \cite{Yuan2005,Ding2003}. More specifically, in \textit{Massey’s secret sharing schemes} based on the duals of minimal codes, the minimal access structures can be fully characterized in terms of the parameters and supports of codewords in the original minimal codes \cite{Massey1993}. Ashikhmin and Barg \cite{Ashikhmin1998} established a sufficient condition under which a linear code is minimal. Subsequently, Chang and Hyun \cite{Chang2018} constructed the first infinite family of binary minimal codes that violate this sufficient condition by choosing the defining set as $\Delta \setminus \{\mathbf{0}\},$ where $\Delta$ is a simplicial complex of $\mathbb{F}_2^m$ with two maximal elements. This breakthrough has inspired a proliferation of research focused on constructing linear codes with novel parameters via simplicial complexes (see \cite{Cheng2022,Hu2022} and references therein).

Building upon these foundations, Wu \etal~\cite{Wu2020} investigated linear codes over the ring $\frac{\mathbb{F}_2[u]}{\langle u^2 \rangle}$ with defining sets of the forms $\Delta_{\mathpzc{P}} + u\Delta_{\mathpzc{Q}}^c$ and $\Delta_{\mathpzc{P}}^c + u\Delta_{\mathpzc{Q}}^c,$ where $\Delta_{\mathpzc{P}}$ and $\Delta_{\mathpzc{Q}}$ denote the simplicial complexes of $\mathbb{F}_2^m$ with supports $\mathpzc{P}$ and $\mathpzc{Q},$ respectively. They determined the Lee weight distributions of these codes. They also studied their Gray images and obtained an infinite family of binary codes attaining the Griesmer bound as well as an infinite family of binary distance-optimal codes. Subsequently, Li and Shi \cite{Li2021} examined a linear code over the ring $\frac{\mathbb{F}_2[u]}{\langle u^3 \rangle}$ with defining set $\Delta_{\mathpzc{P}} + u\Delta_{\mathpzc{Q}}^c + u^2\Delta_{\mathpzc{R}}$ and determined its Lee weight distribution, where $\Delta_{\mathpzc{P}},$ $\Delta_{\mathpzc{Q}},$ and $\Delta_{\mathpzc{R}}$ are simplicial complexes of $\mathbb{F}_2^m$ with supports $\mathpzc{P},$ $\mathpzc{Q},$ and $\mathpzc{R},$ respectively. They also obtained an infinite family of binary minimal and distance-optimal codes via the Gray images of these codes.

Several additional studies have focused on constructing linear codes over rings — not exclusively finite commutative chain rings — with defining sets specified via simplicial complexes and on determining their Lee weight distributions (see \cite{Wu2024S,Shi2022,Wu2024} and references therein). In a recent work, Mondal and Lee \cite{Mondal2024} constructed linear codes over the ring $\frac{\mathbb{F}_2[u]}{\langle u^2 \rangle}$ whose defining sets are subsets of $\big( \frac{\mathbb{F}_2[u]}{\langle u^2 \rangle} \times \mathbb{F}_2 \big)^m$ of the forms $(\Delta_{\mathpzc{P}} + u\Delta_{\mathpzc{Q}}) \times \Delta_{\mathpzc{R}},$ $(\Delta_{\mathpzc{P}} + u\Delta_{\mathpzc{Q}}^c) \times \Delta_{\mathpzc{R}},$ and  $(\Delta_{\mathpzc{P}}^c + u\Delta_{\mathpzc{Q}}) \times \Delta_{\mathpzc{R}} ,$ where $\Delta_{\mathpzc{P}},$ $\Delta_{\mathpzc{Q}},$ and $\Delta_{\mathpzc{R}}$ are simplicial complexes of $\mathbb{F}_2^m$ with supports $\mathpzc{P},$ $\mathpzc{Q}$ and $\mathpzc{R},$ respectively. They determined the parameters and Lee weight distributions of these codes. They also studied their Gray images and obtained two infinite families of binary distance-optimal codes. Moreover, they derived a sufficient condition for these codes to be minimal and observed that these codes are always  self-orthogonal. Additionally, by employing the Pless  power moment identities \cite[Sec. 7.2]{HuffF}, they obtained an infinite family of binary projective $3$-weight codes with non-zero Hamming weights summing to $\frac{3}{2}$ times the code length,  which give rise to strongly $\ell$-walk-regular graphs with new parameters for all odd integers $\ell \geq 3.$

Motivated by the aforementioned developments, in this paper, we will consider a mixed-alphabet ring $\mathcal{R} = \frac{\mathbb{F}_q[u]}{\langle u^2 \rangle}\times \mathbb{F}_q,$  where $\mathbb{F}_q$ is the finite field of prime power order $q$  and $\frac{\mathbb{F}_q[u]}{\langle u^2 \rangle}$ is the quasi-Galois ring with maximal ideal $\langle u \rangle$ of nilpotency index $2$ and residue field $\mathbb{F}_q.$ Here, we will construct four infinite families of linear codes over the ring $\frac{\mathbb{F}_q[u]}{\langle u^2 \rangle}$ whose defining sets are subsets of $ \mathcal{R}^m$ of the forms: \begin{eqnarray}\label{S12}\mathcal{S}_1= (\Delta_\mathpzc{A}+u\Delta_\mathpzc{B})\times (\mathbb{F}_{q}^m \setminus \Delta_\mathpzc{C}), & &\mathcal{S}_2= \big((\mathbb{F}_q^m \setminus \Delta_\mathpzc{A})+u\Delta_\mathpzc{B}\big)\times\Delta_\mathpzc{C},\\ \label{S34}\mathcal{S}_3= \big(\Delta_\mathpzc{A}+u(\mathbb{F}_q^m \setminus \Delta_\mathpzc{B})\big)\times\Delta_\mathpzc{C} &\text{ and } & \mathcal{S}_4= \big((\Delta_\mathpzc{A}\setminus \{\mathbf{0}\})+u\Delta_\mathpzc{B}\big)\times (\mathbb{F}_q^m \setminus \Delta_\mathpzc{C}),\end{eqnarray} 
where $m \geq 2$ is an integer and $\Delta_\mathpzc{A},$ $\Delta_\mathpzc{B},$ and $\Delta_\mathpzc{C}$ are simplicial complexes of $\mathbb{F}_q^m$ with support $\mathpzc{A},$ $\mathpzc{B},$ and $\mathpzc{C},$ respectively. We will explicitly determine the parameters and Lee weight distributions of these codes by extending the techniques employed in \cite[Sec. III]{Mondal2024}. 
We will also study their Gray images and identify several infinite families of few-weight codes as well as an infinite family of minimal, near-Griesmer and distance-optimal codes over $\mathbb{F}_q.$ We will show that the codes belonging to these families are self-orthogonal for  $q = 2$ or $ 3.$
Furthermore, for an arbitrary defining set $\mathcal{D} \subseteq \mathcal{R}^m,$ we determine a spanning matrix of the linear code over $\frac{\mathbb{F}_q[u]}{\langle u^2 \rangle}$ with defining set $\mathcal{D},$ and use it  to construct a spanning matrix of its Gray image over $\mathbb{F}_q.$ With the help of this result, we will construct two infinite families of projective few-weight codes over $\mathbb{F}_q$  by directly analyzing the spanning matrices of the Gray images of  linear codes over $\frac{\mathbb{F}_q[u]}{\langle u^2 \rangle}$ with  defining sets $\mathcal{S}_2$ and $\mathcal{S}_4,$ while Mondal and Lee \cite[Rem. 6]{Mondal2024} employed the Pless power moment identities \cite[Sec. 7.2]{HuffF} to obtain a family of binary projective codes.   We will then study the duals of these projective codes and explicitly determine their parameters. As a consequence, we will identify two infinite families of binary distance-optimal codes.  Additionally, we will construct an infinite family of quaternary projective $3$-weight codes, with non-zero Hamming weights summing to $\frac{9}{4}$ times the code length. This  addresses an open question posed by Shi and Sol\'{e} \cite[Sec. 6]{Shi2019}  concerning the construction of new projective $3$-weight codes over $\mathbb{F}_q,$  with non-zero Hamming weights   summing to $\frac{3  (q-1)}{q}$ times the code length, in the special case $q=4.$ 
We will examine the minimal access structures of Massey’s secret-sharing schemes based on the duals of our newly constructed minimal codes over $\mathbb{F}_q$ and obtain the number of dictatorial participants in these schemes. Finally, we will investigate the locality properties of our newly constructed projective codes, and demonstrate that these codes have locality either $2$ or $3.$ This yields four infinite families of $q$-ary locally repairable codes (LRCs) with locality $2$ and two infinite families of binary LRCs with locality $3.$

We also compare the parameters of the linear codes over $\mathbb{F}_q$ obtained in this work with those listed in Table II of Hu \etal\ \cite[Sec.~5]{Hu2024} and Tables 1 and 2 of Lavanya and Sharma \cite{Lavanya2025}, which provide comparisons of the parameters of linear codes constructed using various defining sets. Furthermore, we compare the parameters of our codes with those reported in the survey article by Wu \etal\ \cite{Wu2024S}. In Table \ref{TableC1}, we compare the parameters of our codes with those of existing codes having similar parameters, as constructed in \cite{Chen2025,Mondal2024,Wu2024,Xia02025,Wu2026,Hu2026}. These comparisons show that our codes are new, except in certain special cases, as discussed in Remarks \ref{comp1} -- \ref{comp3.2} and \ref{NewRemark}. Furthermore, this paper addresses Open Problem 3 proposed in the aforementioned survey by  Wu \etal\ \cite[p. 14]{Wu2024S}, which calls for the construction of more optimal codes over different finite rings and the determination of their weight distributions with respect to various metrics. Our work not only addresses this open problem, but also provides several constructions of distance-optimal, few-weight, minimal, near-Griesmer, self-orthogonal, and projective codes over finite fields, together with an explicit determination of their Hamming weight distributions.

\vspace{-1mm}
\fontsize{7pt}{6pt}\selectfont
\begin{longtable}{|>{\centering\arraybackslash}p{8mm}|c|>{\centering\arraybackslash}p{12mm}|>{\centering\arraybackslash}p{29mm}|c|c|}
 \hline
$q$ & $\mathrm{n}$ & $\mathrm{k}$ & $\mathrm{d}$ &   Remarks & Reference 
\\\hline \rule{0pt}{1em}
$2$ & $2^{2m+1}-2$ & $2(m+1)$ & $4^m-2^{m+1}$ & $m\geq 1$ & \cite[Cor. 5]{Wu2024}
\\\hline\rule{0pt}{1em}
$2$ & $2^{m+1}-2$ & $2(m+1)$ & $2^{m-1}-2$ & $m\geq 1$ & \cite[Cor. 6]{Wu2024} 
\\\hline\rule{0pt}{1em}
prime   & $(p^{|A|}-p^{|A'|})(p^{|B|} -p^{|B'|})$ & $|A|+|B|$ & $(p-1)(p^{|A|+|B|-1} $ & $1 \leq n < m,$  & \cite[Th. 1]{Wu2026} \\
$p$ & & & $-p^{|A|+|B'|-1}$ & $A'\subsetneq A \subseteq [n],$  &\\
& & & $-p^{|A'|+|B|-1} )$ & and $B'\subsetneq B \subseteq [m-n]$ &
\\\hline\rule{0pt}{1em}
prime & $(p^n-p^{|A|}+p^{|A'|})(p^{m-n}-p^{|B|}+p^{|B'|})$ & $m$ & Not Determined & $1\leq n< m,$ $A'\subseteq A\subseteq [n],$ & \cite[Th. 2]{Wu2026} \\
$p$ & & & & and $B'\subseteq B\subseteq [m-n]$ & 
\\\hline\rule{0pt}{1em}
odd prime & $\big(\frac{p-1}{2}\big)^{m-n}(p^{|A|}-p^{|A'|})$ & $m-n+|A|$ & Not Determined & $1\leq n< m$ and $A'\subsetneq A\subseteq [n]$ & \cite[Th. 3]{Wu2026} \\
$p$ & & & &  & 
\\\hline\rule{0pt}{1em}
odd prime & $\big(\frac{p-1}{2}\big)^{m-n}(p^n-p^{|A|}+p^{|A'|})$ & $m$ & Not Determined & $1\leq n< m$ and $A'\subseteq A\subseteq [n]$ & \cite[Th. 4]{Wu2026} \\
$p$& & & &  & 
\\\hline\rule{0pt}{1em}
any & $q^m(q^m - q^{|B|})$ & $2m$ & $2(q-1)q^{m-1}(q^m - q^{|B|})$ & $B\subsetneq [m]$ & \cite[Th. 3.6]{Xia02025} 
\\\hline\rule{0pt}{1em}
any & $q^m(q^m-q^{|B_1|}-q^{|B_2|}+q^{|B_1\cap B_2|})$ & $2m$ & $(q-1)q^{m-1}(2q^m$ & $B_1, B_2 \subseteq [m]$ with & \cite[Th. 3.8]{Xia02025} \\
&&& $-2q^{|B_1|}-2q^{|B_2|}$ & $1\leq |B_1|\leq |B_2|< m$ & \\
&&& $+q^{|B_1\cap B_2|})$ & &
\\\hline\rule{0pt}{1em}
any & $2q^{|A|}(q^{|B|}-q^{|B'|})$ & $|A|+|B|$ & $2(q-1)q^{|A|-1}(q^{|B|}$ & $A\subseteq B\subseteq [m],$ $B'\subsetneq B,$  & \cite[Th. 5]{Chen2025} \\
&&&$-q^{|B'|})$&  and $|A|+|B'|>0$&
\\\hline\rule{0pt}{1em}
any & $2(q^m-q^{|A|})(q^m-q^{|B'|})$ & $2m$ & $2(q-1) (q^{2m-1}$ & $A, B'\subseteq [m],$ $ |A\cup B'|<m,$  & \cite[Th. 6]{Chen2025} \\
& & & $-q^{m+|A|-1}$ & and $q^{m-|A|}>2$&\\
& & & $-q^{m+|B'|-1})$ & &
\\\hline\rule{0pt}{1em}
any & $2q^{|A|}(q^m-q^{|B|}+ q^{|B'|})$ & $m+|A|$ & $2(q-1)q^{|A|-1}(q^m$ & $A, B\subseteq [m],$ $B'\subseteq B,$  & \cite[Th. 8]{Chen2025} \\
&&&$-q^{|B|})$& and $|A\cup B'|\neq|A\cup B|$ &
\\\hline\rule{0pt}{1em}
any & $(q^m
-q^{|A|} + q^{|A'|})q^{|B|}$ & $m+|B|$ & $(q-1)(q^m
-q^{|A|})q^{|B|-1}$ & $A' \subsetneq A \subsetneq [m]$ and $B \subseteq [e]$ &\cite[Th. 1]{Hu2026}
\\\hline\rule{0pt}{1em}
any & $(q^m- q^{|A|} + q^{|A'|})(q^e- q^{|B|})$ & $m+e$ & $(q-1)(q^{m-1}$& $A' \subsetneq A \subsetneq [m],$ $B \subsetneq [e],$ and &\cite[Th. 2]{Hu2026}\\
&&&$-q^{|A|-1})(q^e-q^{|B|})$& $ |A'| + e \geq |A| + |B|$&
\\\hline\rule{0pt}{1em}
any & $(q^m- q^{|A|} + q^{|A'|})(q^e-q^{|B|})$ & $m+e$ & $(q-1) q^e (q^{m-1}$& $A' \subsetneq A \subsetneq [m],$ $B \subsetneq [e],$ and &\cite[Th. 2]{Hu2026}\\
&&& $-q^{|A|-1}+q^{|A'|-1})$& $ |A'| + e < |A| + |B|$&\\
&&& $-(q-1)q^{m+|B|-1}$& &
\\\hline\rule{0pt}{1em}
any & $(q^m-q^{|A_1|}-q^{|A_2|}+2q^{|A_1 \cap A_2 |})q^{|B|}$ & $m + |B|$ & $(q-1)q^{|B|}(q^{m-1}$& $B\subseteq[e],$ $A_1,A_2\subseteq[m],$  & \cite[Th. 4]{Hu2026}\\
&&& $-q^{|A_1 |-1}- q^{|A_2 |-1})$ & $|A_1|\leq|A_2|,$ and &\\
&&&& $q^m > q^{|A_1|}+q^{|A_2|}$ &
\\\hline\rule{0pt}{1em}
any & $2q^{|\mathpzc{A}|+|\mathpzc{B}|}(q^m-q^{|\mathpzc{C}|})$ & $m+|\mathpzc{A}|$ & $(q-1)q^{|\mathpzc{A}|+|\mathpzc{B}|-1}(q^m$ & $\mathpzc{C}\subsetneq[m],$ $\mathpzc{A}, \mathpzc{B}\subseteq [m],$   & Th. \ref{ThN1}\\
&& $+|\mathpzc{A}\cup\mathpzc{B}|$ & $-q^{|\mathpzc{C}|})$ & and $\mathpzc{A}\not\subseteq \mathpzc{B}$&
\\\hline\rule{0pt}{1em}
any & $2q^{|\mathpzc{B}|+|\mathpzc{C}|}(q^m-q^{|\mathpzc{A}|})$ & $2m+|\mathpzc{C}|$ & $(q-1)(q^m $ & $\mathpzc{C}\subseteq [m]$ and $\mathpzc{A}, \mathpzc{B}\subsetneq[m]$ with & Th. \ref{ThN2}\\
& & & $-\epsilon q^{|\mathpzc{A}|})q^{|\mathpzc{B}|+|\mathpzc{C}|-1}$ & $\epsilon=0$ if $\mathpzc{A}\subseteq \mathpzc{B},$ while &\\
&&&& $\epsilon=1$ if $\mathpzc{A}\not\subseteq \mathpzc{B}$&
\\\hline\rule{0pt}{1em}
any & $2q^{|\mathpzc{A}|+|\mathpzc{C}|}(q^m-q^{|\mathpzc{B}|})$ & $m+|\mathpzc{A}|$ & $2 (q-1) (q^m  $ & $\mathpzc{A}, \mathpzc{C}\subseteq [m],$ $\mathpzc{B}\subsetneq[m],$ &Th. \ref{ThN3}\\
&& $+|\mathpzc{C}|$ & $-q^{|\mathpzc{B}|}) q^{|\mathpzc{A}|+|\mathpzc{C}|-1} $ & and $\mathpzc{A}\not\subseteq \mathpzc{B}$&
\\\hline\rule{0pt}{1em}
any & $2(q^m - q^{|\mathpzc{C}|})(q^{|\mathpzc{A}|} - 1) q^{|\mathpzc{B}|}$ & $m + |\mathpzc{A}|$ & $2(q-1) q^{|\mathpzc{B}|-1}\big((q^m $ & $\mathpzc{C}\subsetneq[m],$ $\mathpzc{B}\subseteq[m],$  $\mathpzc{A}\subseteq \mathpzc{B},$ & Th. \ref{ThN4}(a) \\
&& $+|\mathpzc{B}|$ & $- q^{|\mathpzc{C}|})(q^{|\mathpzc{A}|}-1)-q^{|\mathpzc{C}|} \big)$ & and $|\mathpzc{A}|\geq 2$ &
\\\hline\rule{0pt}{1em}
any & $2(q^m-q^{|\mathpzc{C}|})(q^{|\mathpzc{A}|}-1)q^{|\mathpzc{B}|}$ & $m+|\mathpzc{A}|$ & $(q-1)(q^m $ & $\mathpzc{C}\subsetneq[m],$ $\mathpzc{A},\mathpzc{B}\subseteq[m],$ &Th. \ref{ThN4}(b)\\
&& $+|\mathpzc{A}\cup \mathpzc{B}|$ & $-q^{|\mathpzc{C}|})q^{|\mathpzc{A}|+|\mathpzc{B}|-1}$ & $|\mathpzc{A}|\geq 2,$ and $\mathpzc{A}\not\subseteq\mathpzc{B}$&
\\\hline\rule{0pt}{1em}
any & $2q^{|\mathpzc{B}|+|\mathpzc{C}|}(q^m-q^{|\mathpzc{A}|})/(q-1)$& $2m+|\mathpzc{C}|$ & $(q^m-\rho q^{|\mathpzc{A}|})q^{|\mathpzc{B}|+|\mathpzc{C}|-1}$ & $\mathpzc{B} \subseteq \mathpzc{A}\subsetneq[m]$ with & Th. \ref{Th7}\\
&&&& $\rho=0$ if $\mathpzc{B}= \mathpzc{A},$ while & \\
&&&& $\rho=1$ if $\mathpzc{B}\subsetneq \mathpzc{A}$&
\\\hline\rule{0pt}{1em}
any & $2(q^m-q^{|\mathpzc{C}|})(q^{|\mathpzc{A}|}-1)q^{|\mathpzc{B}|}/(q-1) $ & $m + 2|\mathpzc{A}|$ & $(q^{m}-q^{|\mathpzc{C}|})q^{|\mathpzc{A}|+|\mathpzc{B}|-1}$ & $\mathpzc{A}, \mathpzc{B}\subseteq[m],$ $\mathpzc{A}\cap\mathpzc{B}=\emptyset,$ & Th. \ref{Th8}\\
&& $+|\mathpzc{B}|$ & & $|\mathpzc{A}|\geq 2,$ and $\mathpzc{C} \subsetneq [m]$ &\\
\hline
 
\caption{Parameters of linear $[\mathrm{n},\mathrm{k},\mathrm{d}]$-codes over $\mathbb{F}_q$ constructed using simplicial complexes} 
\label{TableC1}
\end{longtable}
\normalsize

In a subsequent work \cite{Josehull}, we show that the Gray images of the linear codes  with defining sets $\mathcal{S}_1,$ $\mathcal{S}_2,$ and $\mathcal{S}_3$ are Galois self-orthogonal for every automorphism of $\mathbb{F}_q$ over $\mathbb{F}_p,$ where $q$ is a power of the prime $p.$  Leveraging this Galois self-orthogonality, we construct several  families of entanglement-assisted quantum error-correcting codes (EAQECCs). In addition, we identify three classes of EAQECCs  that achieve the Griesmer-type bound on the lengths of EAQECCs constructed from linear codes over finite fields. Apart from this, we derive three infinite families of intersecting codes over $\mathbb{F}_q$ and explicitly determine their trellis complexities.

The remainder of this paper is organized as follows: In Section \ref{Prelim}, we present some preliminaries needed to derive our main results. In Section \ref{Basiclem},   
we establish two key lemmas needed to derive our main results  (Lemmas \ref{Lem3} and \ref{Lem4}).  In Section \ref{sec4}, we first obtain a spanning matrix of a linear code over $\frac{\mathbb{F}_q[u]}{\langle u^2 \rangle}$ with an arbitrary defining set $\mathcal{D} \subseteq \mathcal{R}^m$ (Theorem \ref{Th5}), and subsequently use it to construct a spanning matrix for its  Gray image over $\mathbb{F}_q$ (Theorem \ref{Th6}). In Section \ref{Newcodes}, we construct four new infinite families of linear codes over $\frac{\mathbb{F}_q[u]}{\langle u^2 \rangle}$ with  defining sets $\mathcal{S}_1,$ $\mathcal{S}_2,$ $\mathcal{S}_3,$ and $\mathcal{S}_4$ as defined in \eqref{S12} and \eqref{S34}, and 
 explicitly determine their parameters and Lee weight distributions   (Theorems \ref{Th1} - \ref{Th4}). In Section \ref{GI}, we study the Gray images of the codes with defining sets $\mathcal{S}_1,$ $\mathcal{S}_2,$ $\mathcal{S}_3,$ and $\mathcal{S}_4,$ and obtain several infinite families of few-weight codes, binary and ternary self-orthogonal codes, as well as an infinite family of minimal, near-Griesmer, and distance-optimal codes over $\mathbb{F}_q$  (Theorem \ref{ThN1} - \ref{ThN4}). We also present examples to illustrate these results (Examples \ref{Ex4.1} – \ref{Ex4.4}).
In Section \ref{Sec5}, 
we construct two families of projective few-weight codes over $\mathbb{F}_q$ with new parameters, and observe that these codes are self-orthogonal for $q = 2$ or $3$  (Theorems \ref{Th7} and \ref{Th8}, and Remark \ref{NewRemark}). Furthermore, we study the duals of these projective codes and obtain two infinite families of binary distance-optimal codes (Theorems \ref{Th9} and \ref{Th10}). Additionally, we construct an infinite family of quaternary projective $3$-weight codes  with new parameters and non-zero Hamming weights   summing to $ \frac{9}{4}$ times the code length. This  addresses an open question posed by   Shi and Sol{\'e}  \cite[Sec. 6]{Shi2019} pertaining to the construction of projective $3$-weight codes over $\mathbb{F}_q,$  with non-zero Hamming weights   summing to $\frac{3  (q-1)}{q}$ times the code length,  in the particular case $q=4$ (Corollary \ref{SHI}). In Section \ref{add}, we explore two additional applications of the results derived in Sections \ref{GI} and \ref{Sec5}. In Section \ref{Sec7},  we study the minimal access structures of Massey’s secret sharing schemes based on the duals of  minimal  codes constructed in Theorem \ref{ThN3}, and determine the number of dictatorial participants in these schemes (Theorem \ref{TA3}). In Section \ref{Sec8}, we  study the locality properties of the projective codes studied in Theorems \ref{Th7} and \ref{Th8}, and show that these codes have locality either $2$ or $3$ (Theorems \ref{Thm5.12} and \ref{Thm5.13}). This gives rise to four infinite families of $q$-ary LRCs with locality $2,$ and two infinite families of binary LRCs with locality $3.$  In Section \ref{Conclusion}, we conclude with a brief summary and outline possible directions for future work. In the appendix, we  provide an elementary proof of the result that the coset graph of a linear code over $\mathbb{F}_q$ is connected. We also construct an infinite family of strongly $\ell$-walk-regular graphs for all odd integers $\ell \geq 3,$ using the quaternary projective $3$-weight codes constructed in Corollary \ref{SHI} (Theorem \ref{Thm5.8}).  Notably, the parameters of these strongly $\ell$-walk-regular graphs match with those obtained in Corollary 1 of Mondal and Lee\cite{Mondal2024}, upon substituting $k = 2m$ in their result.

\section{Some preliminaries}\label{Prelim}
In this section, we will first present some fundamental definitions and results related to linear codes over finite fields. We will then define a specific mixed-alphabet ring, constructed from a quasi-Galois ring with maximal ideal of nilpotency index two and its residue field, along with  the Euclidean bilinear form over this ring. Subsequently, we will recall the definition of a Gray map on a quasi-Galois ring with  maximal ideal of nilpotency index two and  investigate the Gray images of linear codes defined over this quasi-Galois ring.

\subsection{Linear codes over finite fields}\label{Prelim1}
Throughout this paper, let $q$ be a power of a prime $p,$ and let $\mathbb{F}_q$ denote the finite field of order $q.$ Let $\mathrm{n}$ be a positive integer, and let $[\mathrm{n}]$ denote  the set $\{1, 2, \ldots, \mathrm{n}\}.$ Let $\mathbb{F}_q^\mathrm{n}$ denote the $\mathrm{n}$-dimensional vector space consisting of all $\mathrm{n}$-tuples over $\mathbb{F}_q.$
For a word $v \in \mathbb{F}_q^{\mathrm{n}},$ let $(v)_i$ denote the $i$-th coordinate of $v$ for all $i \in [\mathrm{n}].$ For a vector  $v \in \mathbb{F}_q^{\mathrm{n}}$ and a non-empty subset $A$ of $[\mathrm{n}],$ let $(v)_A$ denote the vector of length $|A|$ obtained by deleting the coordinates of $v$ indexed by $A^c:=[\mathrm{n}] \setminus A.$
The support of $v,$ denoted by $\supp(v),$ is defined as the set consisting of all its non-zero coordinate positions, \textit{i.e.,}
$\supp(v) = \{i \in [\mathrm{n}] : (v)_i \neq 0\}.$ 
Further, the Hamming weight of $v,$ denoted by $wt_H(v),$ is defined as $wt_H(v) = |\supp(v)|,$ where $|\cdot|$ denotes the cardinality function. Clearly, for a vector  $v \in \mathbb{F}_q^{\mathrm{n}}$ and a non-empty subset $A$ of $[\mathrm{n}],$ we have $\supp(v) \cap A =\emptyset$ if and only if $(v)_A=\mathbf{0},$ or equivalently, $wt_H((v)_A)=0.$

A linear code $\mathcal{C}$ of length $\mathrm{n}$ and dimension $\mathrm{k}$ over $\mathbb{F}_q$ is defined as a $\mathrm{k}$-dimensional subspace of $\mathbb{F}_q^\mathrm{n}.$ We refer to elements of $\mathcal{C}$ as codewords. A spanning matrix $G$ of the code $\mathcal{C}$ is a matrix over $\mathbb{F}_q$ whose rows span $\mathcal{C}$ as a vector space over $\mathbb{F}_q.$ The set of all rows of $G$ is called a spanning set of $\mathcal{C}.$  The Hamming distance of the code $\mathcal{C},$ denoted by $d_H(\mathcal{C}),$ is given by
\begin{equation*}
    d_H(\mathcal{C}) = \min\{wt_H(c) : c \in \mathcal{C} \text{ and }\ c \neq \mathbf{0} \}.
\end{equation*} 
Now, let us define $A_i = \left| \{c \in \mathcal{C} : wt_H(c) = i \} \right|$ for all $i \in \{0\} \cup [\mathrm{n}].$ The sequence $A_0 = 1, A_1, A_2, \ldots, A_\mathrm{n}$ is called the Hamming weight distribution of the code $\mathcal{C}$ and the polynomial $W_\mathcal{C}(Z)=1+A_1Z+A_2Z^2+\cdots+A_\mathrm{n}Z^{\mathrm{n}}$ is called the Hamming weight enumerator of the code $\mathcal{C}.$
Furthermore, if $t = \left|\{i \in [\mathrm{n}] : A_i \neq 0\} \right|,$ then the code $\mathcal{C}$ is called a $t$-weight code. A few-weight code is defined as a $t$-weight code with a small value of $t.$

Henceforth, we will refer to a linear code of length $\mathrm{n},$ dimension $\mathrm{k},$ and Hamming distance $\mathrm{d}$ over $\mathbb{F}_q$ as a linear $[\mathrm{n},\mathrm{k},\mathrm{d}]$-code over $\mathbb{F}_q.$
A linear $[\mathrm{n},\mathrm{k},\mathrm{d}]$-code over $\mathbb{F}_q$ is said to be (i) distance-optimal if there does not exist a linear $[\mathrm{n},\mathrm{k},\mathrm{d}+1]$-code over $\mathbb{F}_q,$ and (ii)  almost distance-optimal if there exists a distance-optimal linear $[\mathrm{n},\mathrm{k},\mathrm{d}+1]$-code over $\mathbb{F}_q.$
A well-known bound for linear codes over finite fields is the Griesmer bound, which gives a lower bound on the length of a code for a given dimension, Hamming distance, and alphabet size. The Griesmer bound for a linear $[\mathrm{n},\mathrm{k},\mathrm{d}]$-code over $\mathbb{F}_q$ (see \cite[Th. 2.7.4]{HuffF}) is given by
\begin{equation}\label{GB}
\sum\limits_{i=0}^{\mathrm{k}-1} \left\lceil \frac{\mathrm{d}}{q^i} \right\rceil \leq \mathrm{n},
\end{equation}
where $\lceil \cdot \rceil$ denotes the ceiling function. A linear $[\mathrm{n},\mathrm{k},\mathrm{d}]$-code over $\mathbb{F}_q$ is said to be a Griesmer code if its parameters satisfies
$\sum\limits_{i=0}^{\mathrm{k}-1} \left\lceil \frac{\mathrm{d}}{q^i} \right\rceil = \mathrm{n},$ while it is said to be a near-Griesmer code if $\sum\limits_{i=0}^{\mathrm{k}-1} \left\lceil \frac{\mathrm{d}}{q^i} \right\rceil = \mathrm{n} - 1.$ A near-Griesmer  $[\mathrm{n},\mathrm{k},\mathrm{d}]$-code over $\mathbb{F}_q$ is distance-optimal if $\mathrm{k} > 1$ and $q$ divides $\mathrm{d}$ \cite[Lem. 2]{Hu2022}. Another well-known bound for linear codes over finite fields is the Sphere-packing bound. For a linear $[\mathrm{n}, \mathrm{k}, \mathrm{d}]$-code over $\mathbb{F}_q,$ the Sphere-packing bound (see \cite[Th. 11.1.4]{HuffF}) is given by
\begin{equation}\label{SPB}
  \sum\limits_{i=0}^{\left\lfloor \frac{\mathrm{\mathrm{d}-1}}{2} \right\rfloor}\binom{\mathrm{n}}{i}(q-1)^i \leq q^{\mathrm{n}-\mathrm{k}},
\end{equation}where $\binom{\cdot}{\cdot}$ denotes the binomial coefficient and $\lfloor \cdot \rfloor$ denotes the floor function. 

Minimal codes form another important class of linear codes.
A linear code $\mathcal{C}$ of length $\mathrm{n}$ over $\mathbb{F}_q$ is said to be minimal if all its codewords are minimal, \textit{i.e.,} if $\supp(c') \subseteq \supp(c)$ for any two codewords $c, c' \in \mathcal{C},$  then we must have $c' = \alpha c$ for some $\alpha \in \mathbb{F}_q.$ Below, we state Lemma 2.1(3) of Ashikhmin and Barg \cite{Ashikhmin1998}, which provides a sufficient condition under which a linear code over $\mathbb{F}_q$ is minimal.

\begin{lemma}\label{Lem1}\cite[Lem. 2.1(3)]{Ashikhmin1998}
Let $\mathcal{C}$ be a linear code over $\mathbb{F}_q,$ and let $w_0$ and $w_\infty$ denote the minimum and maximum among the Hamming weights of  non-zero codewords of  $\mathcal{C},$ respectively. If  $\frac{w_0}{w_\infty}>\frac{q-1}{q},$ then the code $\mathcal{C}$ is minimal. 
\end{lemma}

The dual  of a linear code $\mathcal{C}$ of length $\mathrm{n}$ over $\mathbb{F}_q,$ denoted by $\mathcal{C}^\perp,$ is defined as
\begin{equation*}
  \mathcal{C}^\perp = \{v \in \mathbb{F}_q^\mathrm{n} : v \cdot c = 0 \text{ for all } c \in \mathcal{C} \},  
\end{equation*}
where the map $\cdot: \mathbb{F}_q^{\mathrm{n}} \times  \mathbb{F}_q^{\mathrm{n}} \rightarrow \mathbb{F}_q$ is given by \begin{equation}\label{INP1}v.w=v_1w_1+v_2w_2+\cdots+v_{\mathrm{n}}w_{\mathrm{n}}\end{equation} 
for all $v=(v_1,v_2, \ldots,v_{\mathrm{n}}), w=(w_1,w_2,\ldots,w_{\mathrm{n}}) \in \mathbb{F}_q^{\mathrm{n}}.$ The map $\cdot$ is called the Euclidean bilinear form on $\mathbb{F}_q^\mathrm{n}.$
Note that $\mathcal{C}^\perp$ is a linear code of length $\mathrm{n}$ and dimension $\mathrm{n} - \dim(\mathcal{C})$ over $\mathbb{F}_q,$ where $\dim(\cdot)$ denotes the dimension of a code (see \cite[Th. 7.3]{Hill1986}). The code $\mathcal{C}$ is said to be self-orthogonal if it satisfies $\mathcal{C}\subseteq \mathcal{C}^\perp.$ 

Projective codes constitute another important class of linear codes.  A linear code $\mathcal{C}$ over $\mathbb{F}_q$ is said to be projective if $d(\mathcal{C}^\perp) \geq 3.$   We next state the following well-known result.

\begin{lemma}\label{Lem2}\cite[Th. 8.4]{Hill1986}
Let $\mathcal{C}$ be a linear code over $\mathbb{F}_q$ with a spanning matrix $G.$ The following hold.
\begin{itemize}
    \item [(a)] Any $\mathrm{d}-1$ columns of $G$ are linearly independent over $\mathbb{F}_q$ if and only if $d(\mathcal{C}^\perp) \geq \mathrm{d}.$
    \item[(b)] There are $\mathrm{d}$ linearly dependent columns of $G$  over $\mathbb{F}_q$ if and only if $d(\mathcal{C}^\perp) \leq \mathrm{d}.$ 
\end{itemize}
Consequently, we have $d(\mathcal{C}^\perp) = \mathrm{d}$ if and only if any $\mathrm{d}-1$ columns of $G$ are linearly independent over $\mathbb{F}_q,$ and there are $\mathrm{d}$ linearly dependent columns of $G$ over $\mathbb{F}_q.$
\end{lemma}

 Further, for a non-empty subset $A$  of $[\mathrm{n}],$ the set
\begin{equation*}
    \Delta_A = \{ v \in \mathbb{F}_q^\mathrm{n} : \supp(v) \subseteq A \}
\end{equation*}
is called a simplicial complex of $\mathbb{F}_q^\mathrm{n}$ with support $A$ (see \cite{Chang2018}).
Note that $\Delta_A$ is an $\mathbb{F}_q$-linear subspace of  $\mathbb{F}_q^{\mathrm{n}}$ with dimension $|A|.$  Furthermore,  we define $\Delta_A^c=\mathbb{F}_q^{\mathrm{n}}\setminus \Delta_A$ and $\Delta_A^\ast=\Delta_A\setminus \{\mathbf{0}\}.$ Note that 
\begin{equation}\label{Scard}
    |\Delta_A|=q^{|A|},~~|\Delta_A^c|=q^{\mathrm{n}}-q^{|A|}, ~~\text{ and }~~|\Delta_A^\ast|=q^{|A|}-1.
\end{equation}
\subsection{A mixed-alphabet ring of a quasi-Galois ring and its residue field, and the associated Euclidean bilinear form}\label{Prelim2}
Here, we first recall that the quotient ring $\frac{\mathbb{F}_q[u]}{\langle u^2 \rangle},$ also known as a quasi-Galois ring,  is a finite commutative chain ring with maximal ideal $\langle u \rangle$ of nilpotency index $2$ and residue field $\mathbb{F}_q.$   For a positive integer $\mathrm{n},$ let $\big(\frac{\mathbb{F}_q[u]}{\langle u^2 \rangle}\big)^\mathrm{n}$ denote the set of all $\mathrm{n}$-tuples over $\frac{\mathbb{F}_q[u]}{\langle u^2 \rangle}.$ One can easily see that $\big(\frac{\mathbb{F}_q[u]}{\langle u^2 \rangle}\big)^\mathrm{n}=\{d+ue: d,e \in \mathbb{F}_q^\mathrm{n}\}.$ Note that $\big(\frac{\mathbb{F}_q[u]}{\langle u^2 \rangle}\big)^\mathrm{n}$ is a free module of rank $\mathrm{n}$ over $\frac{\mathbb{F}_q[u]}{\langle u^2 \rangle}.$
For our convenience, we will denote the Euclidean bilinear form on $\big(\frac{\mathbb{F}_q[u]}{\langle u^2 \rangle}\big)^\mathrm{n}$ by $\cdot$ itself. That is, the Euclidean bilinear form on $\big(\frac{\mathbb{F}_q[u]}{\langle u^2 \rangle}\big)^\mathrm{n}$  is a map $\cdot: \big(\frac{\mathbb{F}_q[u]}{\langle u^2 \rangle}\big)^\mathrm{n} \times \big(\frac{\mathbb{F}_q[u]}{\langle u^2 \rangle}\big)^\mathrm{n} \rightarrow \big(\frac{\mathbb{F}_q[u]}{\langle u^2 \rangle}\big) ,$ defined as \begin{equation}\label{INP2}r\cdot s=r_1s_1+r_2s_2+\cdots+r_{\mathrm{n}}s_{\mathrm{n}}\end{equation} for all $r=(r_1,r_2,\ldots,r_{\mathrm{n}}),$ $s=(s_1,s_2,\ldots, s_{\mathrm{n}}) \in \big(\frac{\mathbb{F}_q[u]}{\langle u^2 \rangle}\big)^\mathrm{n}.$

Now, we define a mixed-alphabet ring $\mathcal{R}$ as follows (see \cite[p. 6522]{Borges2018}):
\begin{equation*}
    \mathcal{R}:=\frac{\mathbb{F}_q[u]}{\langle u^2 \rangle}\times \mathbb{F}_q=\{(x+uy,z):x,y,z \in \mathbb{F}_q\}.
\end{equation*} Note that the ring $\mathcal{R}$ can also be viewed as a module over $\frac{\mathbb{F}_q[u]}{\langle u^2 \rangle}.$ The set of all $\mathrm{n}$-tuples over $\mathcal{R},$ denoted by
$\mathcal{R}^\mathrm{n},$ is given by $\mathcal{R}^\mathrm{n}=\{(d+ue,f) : d,e,f \in \mathbb{F}_q^\mathrm{n}\}.$ Further, the set $\mathcal{R}^\mathrm{n}$ can be naturally viewed as a module over  $\frac{\mathbb{F}_q[u]}{\langle u^2 \rangle}.$  Now, the Euclidean bilinear form on $\mathcal{R}^\mathrm{n}$ (see \cite[p. 6523]{Borges2018}) is a map $\langle \cdot , \cdot \rangle : \mathcal{R}^\mathrm{n} \times \mathcal{R}^\mathrm{n} \rightarrow \frac{\mathbb{F}_q[u]}{\langle u^2 \rangle},$ defined as 
\begin{equation}\label{IP}
\langle (d_1+ue_1,f_1),(d_2+ue_2,f_2) \rangle = (d_1+ue_1) \cdot (d_2+ue_2) + uf_1 \cdot f_2
\end{equation} for all $(d_1+ue_1,f_1),(d_2+ue_2,f_2)\in \mathcal{R}^\mathrm{n}$ with $d_1,d_2,e_1,e_2,f_1,f_2 \in \mathbb{F}_q^{\mathrm{n}},$ where $\cdot$ denotes the Euclidean bilinear form on  $\mathbb{F}_q^\mathrm{n}$ and $\big(\frac{\mathbb{F}_q[u]}{\langle u^2 \rangle}\big)^\mathrm{n},$ as defined by equations \eqref{INP1} and \eqref{INP2}, respectively.
One can easily see that the Euclidean bilinear form $\langle \cdot , \cdot \rangle$ on $\mathcal{R}^\mathrm{n}$ is a  non-degenerate and symmetric bilinear form. 
\subsection{Gray images of linear codes over the quasi-Galois ring $\frac{\mathbb{F}_q[u]}{\langle u^2 \rangle}$}\label{Prelim3}

A Gray map on the quasi-Galois ring $\frac{\mathbb{F}_q[u]}{\langle u^2 \rangle}$  (see \cite[p. 2522]{Greferath1999}) is a map $\Phi : \frac{\mathbb{F}_q[u]}{\langle u^2 \rangle} \rightarrow \mathbb{F}_q^2,$ defined as $\Phi(x + uy) = (y, x + y)$ for all $x + uy \in \frac{\mathbb{F}_q[u]}{\langle u^2 \rangle}$ with $x,y \in \mathbb{F}_q.$
The map $\Phi$ can be naturally extended component-wise to a map from $ \big(\frac{\mathbb{F}_q[u]}{\langle u^2 \rangle}\big)^\mathrm{n} $ onto $ \mathbb{F}_q^{2\mathrm{n}}$ as
\begin{equation*}
d + ue \mapsto (e, d + e) \quad \text{for all } d + ue \in \bigg(\frac{\mathbb{F}_q[u]}{\langle u^2 \rangle}\bigg)^\mathrm{n} \text{ ~with } d,e \in \mathbb{F}_q^{\mathrm{n}},
\end{equation*}
which we shall denote by  $\Phi$ itself for our convenience. Further,  the Lee weight of a word $d + ue \in \left(\frac{\mathbb{F}_q[u]}{\langle u^2 \rangle}\right)^\mathrm{n}$ with $d,e \in \mathbb{F}_q^{\mathrm{n}},$  denoted by $wt_L(d + ue),$ is defined as the Hamming weight of its Gray image, \textit{i.e.,}
\begin{equation}\label{GIe}
  wt_L(d + ue) := wt_H(\Phi(d + ue)) = wt_H(e) + wt_H(d + e).  
\end{equation}
Thus, the map $\Phi$ is an $\mathbb{F}_q$-linear isometry from $\big(\big(\frac{\mathbb{F}_q[u]}{\langle u^2 \rangle}\big)^\mathrm{n}, wt_L(\cdot)\big)$ onto $\big(\mathbb{F}_q^{2\mathrm{n}}, wt_H(\cdot)\big).$

Now, a linear code $\mathscr{C}$ of length $\mathrm{n}$ over $\frac{\mathbb{F}_q[u]}{\langle u^2 \rangle}$ is defined as a submodule of $\big(\frac{\mathbb{F}_q[u]}{\langle u^2 \rangle}\big)^\mathrm{n}$ over $\frac{\mathbb{F}_q[u]}{\langle u^2 \rangle}.$  
 A spanning matrix $\mathscr{G}$ of the code $\mathscr{C}$ is a matrix over $\frac{\mathbb{F}_q[u]}{\langle u^2 \rangle}$ whose rows generate $\mathscr{C}$ as a module over $\frac{\mathbb{F}_q[u]}{\langle u^2 \rangle}.$  The set of all rows of $\mathscr{G}$ is called a spanning set of $\mathscr{C}.$ We will refer to the cardinality of the code $\mathscr{C}$ as its size.
The Lee distance of $\mathscr{C},$ denoted by $d_L(\mathscr{C}),$ is given by
$d_L(\mathscr{C}) = \min\{wt_L(c) : c \in \mathscr{C} \text{ and }\ c \neq \mathbf{0}\}.$
Further, let us define $\mathsf{A}_i = \left| \{c \in \mathscr{C} : wt_L(c) = i \} \right|$ for all $i \in \{0\} \cup [2\mathrm{n}],$ where $[2\mathrm{n}]=\{1,2,\ldots, 2\mathrm{n}\}.$ The sequence $\mathsf{A}_0 = 1, \mathsf{A}_1, \mathsf{A}_2, \ldots, \mathsf{A}_{2\mathrm{n}}$ is called the Lee weight distribution of the code $\mathscr{C}.$ Moreover, if $t = \left|\{i \in [2\mathrm{n}] : \mathsf{A}_i \neq 0\} \right|,$ then the code $\mathscr{C}$ is called a $t$-weight code. A linear code $\mathscr{C}$ of length $\mathrm{n},$ size $\mathrm{K},$ and Lee distance $\mathrm{D}$ over $\frac{\mathbb{F}_q[u]}{\langle u^2 \rangle}$ is referred to as a linear code over $\frac{\mathbb{F}_q[u]}{\langle u^2 \rangle}$ with parameters $(\mathrm{n}, \mathrm{K}, \mathrm{D}),$ or simply a linear $(\mathrm{n}, \mathrm{K}, \mathrm{D})$-code over $\frac{\mathbb{F}_q[u]}{\langle u^2 \rangle}.$ We further observe the following:
\begin{remark}\label{RRk}
Since the map $\Phi$ is an $\mathbb{F}_q$-linear isomorphism, the Gray image $\Phi(\mathscr{C}) = \{ \Phi(c) : c \in \mathscr{C} \}$ of a linear code $\mathscr{C}$ of length $\mathrm{n}$ over $\frac{\mathbb{F}_q[u]}{\langle u^2 \rangle}$ is a linear code of length $2\mathrm{n}$ over $\mathbb{F}_q$ with  $ d_H(\Phi(\mathscr{C}))=d_L(\mathscr{C})$ and $|\Phi(\mathscr{C})|= |\mathscr{C}| .$
Additionally, the Hamming weight distribution of $\Phi(\mathscr{C})$ coincides with the Lee weight distribution of $\mathscr{C}.$ \end{remark}

From this point on, we will use the same notations as introduced in Section \ref{Prelim}. In the following section, we will  establish some key lemmas needed to derive our main results.
\section{Some basic lemmas}\label{Basiclem}
In this section, we will prove two lemmas needed to establish our main results. First of all, let $m \geq 2$ be an integer, and let $\mathpzc{P}$ and $\mathpzc{Q}$ be non-empty subsets of $[m].$ Let us define the following subsets of $\mathbb{F}_q^m$:
\begin{eqnarray}
\mathpzc{X}_{\mathpzc{P}} &=& \{v \in \mathbb{F}_q^m : \supp(v) \cap \mathpzc{P} = \emptyset  \}, \label{XP}\\
\mathpzc{Y}_{\mathpzc{P},\mathpzc{Q}} &=& \{v \in \mathbb{F}_q^m : \supp(v) \cap \mathpzc{P} = \emptyset \text{ and }\supp(v) \cap \mathpzc{Q} \neq \emptyset \}, \text{ and }\label{YPQ}\\
\mathpzc{Z}_{\mathpzc{P},\mathpzc{Q}} &=& \{v \in \mathbb{F}_q^m : \supp(v) \cap \mathpzc{P} \neq \emptyset \text{ and }\supp(v) \cap \mathpzc{Q} \neq \emptyset \}.\label{ZPQ}
\end{eqnarray}
Additionally, let us define $\mathpzc{X}_{\mathpzc{P}}^c=\mathbb{F}_q^m\setminus \mathpzc{X}_{\mathpzc{P}}=\{v \in \mathbb{F}_q^m : \supp(v) \cap \mathpzc{P} \neq \emptyset\}.$ 
   In the following lemma, we determine the cardinalities of the sets $\mathpzc{X}_{\mathpzc{P}},$ $\mathpzc{X}_{\mathpzc{P}}^c,$ $\mathpzc{Y}_{\mathpzc{P},\mathpzc{Q}},$ and $\mathpzc{Z}_{\mathpzc{P},\mathpzc{Q}}.$ 
   \begin{lemma}\label{Lem3} For non-empty subsets $\mathpzc{P}$ and $\mathpzc{Q}$ of $[m],$ we have \begin{eqnarray*}|\mathpzc{X}_{\mathpzc{P}}| = q^{m-|\mathpzc{P}|}, && |\mathpzc{X}_{\mathpzc{P}}^c| = q^{m}-q^{m-|\mathpzc{P}|}, \\
      |\mathpzc{Y}_{\mathpzc{P},\mathpzc{Q}}| = q^{m-|\mathpzc{P}|} - q^{m-|\mathpzc{P} \cup \mathpzc{Q}|}, &\text{ and } & |\mathpzc{Z}_{\mathpzc{P},\mathpzc{Q}}| = q^{m}-q^{m-|\mathpzc{P}|}-q^{m-|\mathpzc{Q}|}+q^{m-|\mathpzc{P} \cup \mathpzc{Q}|}.\end{eqnarray*}
   \end{lemma}
   \begin{proof}
       Its proof is an easy exercise.
   \end{proof}
   Next, let us define the following subsets of $\mathbb{F}_q^{2m}:$
   \begin{eqnarray}\label{M}
  \mathpzc{M}_{\mathpzc{P},\mathpzc{Q}} &=& \{ (e,f) \in (\mathbb{F}_q^{m})^2 : \supp(f) \cap \mathpzc{P} \neq \emptyset,~\supp(e) \cap \mathpzc{P} \neq \emptyset, \text{ and } \supp(e) \cap \mathpzc{Q} = \emptyset \},\\
     \label{N}
  \mathpzc{N}_{\mathpzc{P},\mathpzc{Q}} &=& \{ (e,f) \in (\mathbb{F}_q^{m})^2 : \supp(f) \cap \mathpzc{P} \neq \emptyset,~\supp(e) \cap \mathpzc{P} \neq \emptyset, \text{ and } \supp(e) \cap \mathpzc{Q} \neq \emptyset \}, \\
       \label{M1}
       \widehat{\mathpzc{M}}_{\mathpzc{P},\mathpzc{Q}} &=& \{(e,f) \in \mathpzc{M}_{\mathpzc{P},\mathpzc{Q}} : \supp(e+f) \cap \mathpzc{P} = \emptyset \}, \text{ and }\\
       \label{N1}
       \widehat{\mathpzc{N}}_{\mathpzc{P},\mathpzc{Q}} &=& \{(e,f) \in \mathpzc{N}_{\mathpzc{P},\mathpzc{Q}} : \supp(e+f) \cap \mathpzc{P} = \emptyset \}.
   \end{eqnarray} Further, let us define the sets $\widetilde{\mathpzc{M}}_{\mathpzc{P},\mathpzc{Q}}=\mathpzc{M}_{\mathpzc{P},\mathpzc{Q}}\setminus \widehat{\mathpzc{M}}_{\mathpzc{P},\mathpzc{Q}}$ and $\widetilde{\mathpzc{N}}_{\mathpzc{P},\mathpzc{Q}}=\mathpzc{N}_{\mathpzc{P},\mathpzc{Q}}\setminus \widehat{\mathpzc{N}}_{\mathpzc{P},\mathpzc{Q}}.$
In the following lemma, we  determine the cardinalities of the sets $\mathpzc{M}_{\mathpzc{P},\mathpzc{Q}},$ $\widehat{\mathpzc{M}}_{\mathpzc{P},\mathpzc{Q}},$ $\widetilde{\mathpzc{M}}_{\mathpzc{P},\mathpzc{Q}},$ $\mathpzc{N}_{\mathpzc{P},\mathpzc{Q}},$ $\widehat{\mathpzc{N}}_{\mathpzc{P},\mathpzc{Q}},$ and $\widetilde{\mathpzc{N}}_{\mathpzc{P},\mathpzc{Q}}.$
   \begin{lemma}\label{Lem4}   
   For non-empty subsets $\mathpzc{P}$ and $\mathpzc{Q}$ of $[m],$ we  have \begin{eqnarray*} |\mathpzc{M}_{\mathpzc{P},\mathpzc{Q}}| &=& (q^m-q^{m-|\mathpzc{P}|})(q^{m-|\mathpzc{Q}|}-q^{m-|\mathpzc{P} \cup \mathpzc{Q}|}),\\|\widehat{\mathpzc{M}}_{\mathpzc{P},\mathpzc{Q}}|& = &  (q^{|\mathpzc{P}|-|\mathpzc{P} \cap \mathpzc{Q}|}-1)q^{2m-|\mathpzc{P}|-|\mathpzc{P} \cup \mathpzc{Q}|}, \\ |\widetilde{\mathpzc{M}}_{\mathpzc{P},\mathpzc{Q}}| &= & \big((q^{|\mathpzc{P}|}-1)(q^{|\mathpzc{P} \cup \mathpzc{Q}|-|\mathpzc{Q}|}-1)-(q^{|\mathpzc{P}|-|\mathpzc{P} \cap \mathpzc{Q}|}-1)\big)q^{2m-|\mathpzc{P}|-|\mathpzc{P} \cup \mathpzc{Q}|},\\ 
         |\mathpzc{N}_{\mathpzc{P},\mathpzc{Q}}| &=& (q^m-q^{m-|\mathpzc{P}|})(q^m - q^{m-|\mathpzc{P}|} - q^{m-|\mathpzc{Q}|} + q^{m-|\mathpzc{P} \cup \mathpzc{Q}|}),\\ |\widehat{\mathpzc{N}}_{\mathpzc{P},\mathpzc{Q}}| &= & q^{m-|\mathpzc{P}|}(q^m-q^{m-|\mathpzc{P}|}-q^{m-|\mathpzc{Q}|}+q^{m-|\mathpzc{P} \cup \mathpzc{Q}|}), \text{ and}\\  |\widetilde{\mathpzc{N}}_{\mathpzc{P},\mathpzc{Q}}| &=& (q^m-2q^{m-|\mathpzc{P}|})(q^m - q^{m-|\mathpzc{P}|} - q^{m-|\mathpzc{Q}|} + q^{m-|\mathpzc{P} \cup \mathpzc{Q}|}).
       \end{eqnarray*}
   \end{lemma}
\begin{proof}
To prove the result,  we first note, by \eqref{M}, that $(e,f) \in \mathpzc{M}_{\mathpzc{P},\mathpzc{Q}}$ if and only if 
$$\supp(f) \cap \mathpzc{P} \neq \emptyset,~ \supp(e) \cap \mathpzc{P} \neq \emptyset, \text{ and } \supp(e) \cap \mathpzc{Q} = \emptyset.$$
We see, by \eqref{XP} and \eqref{YPQ}, that $f$ and $e$ have $|\mathpzc{X}_{\mathpzc{P}}^c|$ and  $|\mathpzc{Y}_{\mathpzc{Q},\mathpzc{P}}|$ choices, respectively. This, by Lemma \ref{Lem3}, gives $$|\mathpzc{M}_{\mathpzc{P},\mathpzc{Q}}|=|\mathpzc{X}_{\mathpzc{P}}^c||\mathpzc{Y}_{\mathpzc{Q},\mathpzc{P}}| = (q^m-q^{m-|\mathpzc{P}|})(q^{m-|\mathpzc{Q}|}-q^{m-|\mathpzc{P} \cup \mathpzc{Q}|}).$$ 

Now, to determine the cardinality of the set $\widehat{\mathpzc{M}}_{\mathpzc{P},\mathpzc{Q}},$ 
we first observe, by \eqref{M} and \eqref{M1},  that  $(e,f)\in \widehat{\mathpzc{M}}_{\mathpzc{P},\mathpzc{Q}}$ if and only if 
\begin{equation}\label{A1}
   (f)_{\mathpzc{P}\cap\mathpzc{Q}}=\mathbf{0}, ~
 (e)_{\mathpzc{Q}}=\mathbf{0}, \text{~ and ~} (e)_{\mathpzc{P}\setminus (\mathpzc{P}\cap \mathpzc{Q})}=-(f)_{\mathpzc{P}\setminus (\mathpzc{P}\cap \mathpzc{Q})} \neq \mathbf{0} .
\end{equation}
One can easily see that $[m]=(\mathpzc{P}\setminus (\mathpzc{P}\cap \mathpzc{Q}))\cup (\mathpzc{P}\cap \mathpzc{Q})\cup ([m]\setminus \mathpzc{P})$ (a disjoint union). We will first count the vectors $f\in \mathbb{F}_q^m$ satisfying $(e,f)\in \widehat{\mathpzc{M}}_{\mathpzc{P},\mathpzc{Q}}$ for some $e\in \mathbb{F}_q^m.$ To do this, we see, 
 by \eqref{A1}, that we need to count the vectors $f\in \mathbb{F}_q^m$ satisfying 
 \begin{equation}\label{A2}
   (f)_{\mathpzc{P}\cap\mathpzc{Q}}=\mathbf{0} \text{ ~ and ~ }  (f)_{\mathpzc{P}\setminus (\mathpzc{P}\cap \mathpzc{Q})}\neq \mathbf{0}.
 \end{equation}
 One can easily see that such a vector $f\in \mathbb{F}_q^m$ has precisely $(q^{|\mathpzc{P}|-|\mathpzc{P}\cap\mathpzc{Q}|}-1)q^{m-|\mathpzc{P}|}$ distinct choices.  Further, for a given choice of $ f \in \mathbb{F}_q^m$ satisfying \eqref{A2}, we need to count the vectors $e \in \mathbb{F}_q^m$ such that $(e,f) \in\widehat{\mathpzc{M}}_{\mathpzc{P},\mathpzc{Q}}.$ For this, we see, by \eqref{A1}, that the desired vector $e \in \mathbb{F}_q^m$ must satisfy $
 (e)_{\mathpzc{Q}}=\mathbf{0} \text{~ and ~} (e)_{\mathpzc{P}\setminus (\mathpzc{P}\cap \mathpzc{Q})}=-(f)_{\mathpzc{P}\setminus (\mathpzc{P}\cap \mathpzc{Q})}.$
Since $[m] = \mathpzc{Q} \cup (\mathpzc{P}\setminus (\mathpzc{P}\cap \mathpzc{Q})) \cup ([m] \setminus (\mathpzc{P} \cup \mathpzc{Q}))$ (a disjoint union), we  see that such a vector $e \in \mathbb{F}_q^m$ has precisely $q^{m - |\mathpzc{P} \cup \mathpzc{Q}|}$ distinct choices. From this, it follows that $|\widehat{\mathpzc{M}}_{\mathpzc{P},\mathpzc{Q}}| = (q^{|\mathpzc{P}|-|\mathpzc{P} \cap \mathpzc{Q}|}-1)q^{2m-|\mathpzc{P}|-|\mathpzc{P} \cup \mathpzc{Q}|}.$ 

Furthermore, as $\widetilde{\mathpzc{M}}_{\mathpzc{P},\mathpzc{Q}} = \mathpzc{M}_{\mathpzc{P},\mathpzc{Q}} \setminus \widehat{\mathpzc{M}}_{\mathpzc{P},\mathpzc{Q}},$ we get $$|\widetilde{\mathpzc{M}}_{\mathpzc{P},\mathpzc{Q}}| = |\mathpzc{M}_{\mathpzc{P},\mathpzc{Q}}|-|\widehat{\mathpzc{M}}_{\mathpzc{P},\mathpzc{Q}}|=\big((q^{|\mathpzc{P}|}-1)(q^{|\mathpzc{P} \cup \mathpzc{Q}|-|\mathpzc{Q}|}-1)-(q^{|\mathpzc{P}|-|\mathpzc{P} \cap \mathpzc{Q}|}-1)\big)q^{2m-|\mathpzc{P}|-|\mathpzc{P} \cup \mathpzc{Q}|}.$$

We next observe, by
\eqref{M} and \eqref{N}, that $\mathpzc{M}_{\mathpzc{P},\mathpzc{Q}}\cap \mathpzc{N}_{\mathpzc{P},\mathpzc{Q}}=\emptyset$ and \begin{eqnarray*}\mathpzc{M}_{\mathpzc{P},\mathpzc{Q}} \cup \mathpzc{N}_{\mathpzc{P},\mathpzc{Q}} &= &\{ (e,f) \in (\mathbb{F}_q^{m})^2 : \supp(f) \cap \mathpzc{P} \neq \emptyset \text{ and } \supp(e) \cap \mathpzc{P} \neq \emptyset \}\\&=&\{ (e,f) \in (\mathbb{F}_q^{m})^2 : (f)_{\mathpzc{P}} \neq \mathbf{0} \text{ and } (e)_{\mathpzc{P}}\neq \mathbf{0}\}.\end{eqnarray*} By Lemma \ref{Lem3}, one can easily see that $|\mathpzc{M}_{\mathpzc{P},\mathpzc{Q}} \cup \mathpzc{N}_{\mathpzc{P},\mathpzc{Q}}|=(q^m-q^{m-|\mathpzc{P}|})^2.$
From this, we get $$|\mathpzc{N}_{\mathpzc{P},\mathpzc{Q}}| = |\mathpzc{M}_{\mathpzc{P},\mathpzc{Q}} \cup \mathpzc{N}_{\mathpzc{P},\mathpzc{Q}}| - |\mathpzc{M}_{\mathpzc{P},\mathpzc{Q}}|=(q^m-q^{m-|\mathpzc{P}|})(q^m - q^{m-|\mathpzc{P}|} - q^{m-|\mathpzc{Q}|} + q^{m-|\mathpzc{P} \cup \mathpzc{Q}|}).$$

Finally, to determine $|\widehat{\mathpzc{N}}_{\mathpzc{P},\mathpzc{Q}}|,$ we note, by \eqref{M1} and \eqref{N1}, that $\widehat{\mathpzc{M}}_{\mathpzc{P},\mathpzc{Q}}\cap \widehat{\mathpzc{N}}_{\mathpzc{P},\mathpzc{Q}}=\emptyset$ and \begin{eqnarray*}\widehat{\mathpzc{M}}_{\mathpzc{P},\mathpzc{Q}} \cup \widehat{\mathpzc{N}}_{\mathpzc{P},\mathpzc{Q}}& =& \{ (e,f) \in (\mathbb{F}_q^m)^2 : \supp(f) \cap \mathpzc{P} \neq \emptyset,~ \supp(e) \cap \mathpzc{P} \neq \emptyset, \text{ and } \supp(e+f) \cap \mathpzc{P} = \emptyset  \}\\&=&\{ (e,f) \in (\mathbb{F}_q^m)^2 : (e)_{\mathpzc{P}}=-(f)_{\mathpzc{P}} \neq \mathbf{0} \}.\end{eqnarray*} Here, working as above, we obtain $|\widehat{\mathpzc{M}}_{\mathpzc{P},\mathpzc{Q}} \cup \widehat{\mathpzc{N}}_{\mathpzc{P},\mathpzc{Q}}| = (q^{|\mathpzc{P}|}-1)q^{2m-2|\mathpzc{P}|}.$ From this,  one can easily see that $$|\widehat{\mathpzc{N}}_{\mathpzc{P},\mathpzc{Q}}| = |\widehat{\mathpzc{M}}_{\mathpzc{P},\mathpzc{Q}} \cup \widehat{\mathpzc{N}}_{\mathpzc{P},\mathpzc{Q}}| - |\widehat{\mathpzc{M}}_{\mathpzc{P},\mathpzc{Q}}|=q^{m-|\mathpzc{P}|}(q^m-q^{m-|\mathpzc{P}|}-q^{m-|\mathpzc{Q}|}+q^{m-|\mathpzc{P} \cup \mathpzc{Q}|}).$$ Moreover, since $\widetilde{\mathpzc{N}}_{\mathpzc{P},\mathpzc{Q}} = \mathpzc{N}_{\mathpzc{P},\mathpzc{Q}} \setminus \widehat{\mathpzc{N}}_{\mathpzc{P},\mathpzc{Q}},$ we get $$|\widetilde{\mathpzc{N}}_{\mathpzc{P},\mathpzc{Q}}|=|\mathpzc{N}_{\mathpzc{P},\mathpzc{Q}}|-|\widehat{\mathpzc{N}}_{\mathpzc{P},\mathpzc{Q}}| = (q^m-2q^{m-|\mathpzc{P}|})(q^m - q^{m-|\mathpzc{P}|} - q^{m-|\mathpzc{Q}|} + q^{m-|\mathpzc{P} \cup \mathpzc{Q}|}).$$ This proves the lemma.
\end{proof}

\section{Spanning matrices of linear codes over $\frac{\mathbb{F}_q[u]}{\langle u^2 \rangle}$ with defining sets in $\mathcal{R}^m$ and their Gray images}\label{sec4}
Throughout this paper, let $m \geq 2$ be a fixed  integer. We will represent elements of $\mathcal{R}^{m}$ as $(d+ue,f),$ where $d,e,f \in \mathbb{F}_q^{m}.$ For a non-empty subset $\mathcal{D}$ of $\mathcal{R}^m,$  let us define 
\begin{equation}\label{Code}
\mathscr{C}_{\mathcal{D}} = \{ c_r:= \big( \langle r,s \rangle \big)_{s \in \mathcal{D}} : r \in \mathcal{R}^m \},
\end{equation}
where the map $\langle \cdot, \cdot \rangle$ is as defined by \eqref{IP}. Note that the code $\mathscr{C}_{\mathcal{D}}$ is a linear code of length $|\mathcal{D}|$ over $\frac{\mathbb{F}_q[u]}{\langle u^2 \rangle}.$ Furthermore, the code $\mathscr{C}_{\mathcal{D}}$ is unique up to permutation equivalence and is called the linear code over  $\frac{\mathbb{F}_q[u]}{\langle u^2 \rangle}$ with defining set $\mathcal{D}$ \cite{Mondal2024}.  In this section,  we will determine a spanning matrix of the code $\mathscr{C}_{\mathcal{D}}$ over $\frac{\mathbb{F}_q[u]}{\langle u^2 \rangle}$ with defining set $\mathcal{D} \subseteq \mathcal{R}^m,$ and subsequently, using this matrix, we will derive a spanning matrix of its Gray image $\Phi(\mathscr{C}_\mathcal{D})$ over $\mathbb{F}_q.$ To this end,  we first note that the finite field $\mathbb{F}_q$ can be embedded into the quasi-Galois ring $\frac{\mathbb{F}_q[u]}{\langle u^2 \rangle}.$ Let us define a map $\lambda: \mathcal{R}^m \rightarrow \big(\frac{\mathbb{F}_q[u]}{\langle u^2 \rangle} \big)^{2m}$ as \begin{equation}\label{lambda}\lambda(w_1+uw_2,w_3) = (w_1+uw_2,uw_3)\text{ ~for all ~}(w_1+uw_2,w_3) \in \mathcal{R}^m.\end{equation}
We further recall, by   \eqref{IP} and \eqref{Code}, that  the code $\mathscr{C}_{\mathcal{D}}$ can be expressed as
\begin{equation}\label{Code2}
\mathscr{C}_{\mathcal{D}} = \{ c_r= \big( (d+ue)\cdot(w_1+uw_2) +uf\cdot w_3 \big)_{(w_1+uw_2,w_3) \in \mathcal{D}} : r=(d+ue,f) \in \mathcal{R}^m \}.
\end{equation}
Throughout this paper, let $M_{n_1\times n_2}(Y)$ denote the set of all $n_1 \times n_2$ matrices over the set $Y,$ where $n_1,n_2$ are positive integers. In the following theorem, we determine a spanning matrix of the linear code $\mathscr{C}_{\mathcal{D}}$ over $\frac{\mathbb{F}_q[u]}{\langle u^2 \rangle}.$
\begin{thm}\label{Th5}
For a non-empty subset $\mathcal{D} $  of $\mathcal{R}^m,$ the linear code $\mathscr{C}_{\mathcal{D}}$ over $\frac{\mathbb{F}_q[u]}{\langle u^2 \rangle}$ with defining set $\mathcal{D}$ has a spanning matrix $\mathcal{G} \in M_{2m \times |\mathcal{D}|}\big(\frac{\mathbb{F}_q[u]}{\langle u^2 \rangle}\big)$  whose columns are $\lambda(s)^T \in M_{2m\times 1}\big(\frac{\mathbb{F}_q[u]}{\langle u^2 \rangle}\big)$, where $s$ runs over the  elements of  $\mathcal{D}.$ (Recall that the code $\mathscr{C}_{\mathcal{D}}$ is defined uniquely up to permutation equivalence.) 
\end{thm}
\begin{proof}
To prove the result, let  $\mathscr{D}$ be a linear code  over $\frac{\mathbb{F}_q[u]}{\langle u^2 \rangle}$ with a spanning matrix $\mathcal{G},$ \textit{i.e.,}
\begin{equation}\label{Di}
    \mathscr{D} = \bigg\{ z \mathcal{G} : z \in \bigg(\frac{\mathbb{F}_q[u]}{\langle u^2 \rangle}\bigg)^{2m} \bigg\}.
\end{equation}
 We assert that $\mathscr{D} = \mathscr{C}_{\mathcal{D}}.$ 
 To prove this assertion, let us take $c \in \mathscr{C}_{\mathcal{D}}.$ Here, we observe, by \eqref{Code2}, that there exists  $r=(d+ue,f) \in \mathcal{R}^m$ such that $$c= \Big( (d+ue)\cdot(w_1+uw_2) +uf\cdot w_3 \Big)_{(w_1+uw_2,w_3) \in \mathcal{D}}.$$ By embedding $\mathbb{F}_q^m$ into $\Big(\frac{\mathbb{F}_q[u]}{\langle u^2 \rangle}\Big)^{m},$ we may regard $r$ as an element of $\Big(\frac{\mathbb{F}_q[u]}{\langle u^2 \rangle}\Big)^{2m}.$ This implies that  $c=r \mathcal{G},$ which,  by \eqref{Di}, further implies that $c \in \mathscr{D}.$ This shows that $\mathscr{C}_{\mathcal{D}} \subseteq \mathscr{D}.$

On the other hand, let $b \in \mathscr{D}.$ Here, by \eqref{Di}, we see that there exists an element $z=(x+uh,y+ug) \in \Big(\frac{\mathbb{F}_q[u]}{\langle u^2 \rangle}\Big)^{2m}$ such that $x,y,h,g \in \mathbb{F}_q^m$ and $b = z\mathcal{G}.$ We further note, by  \eqref{lambda}, that the last $m$ rows of $\mathcal{G}$ are multiples of $u.$ This implies that $b=z \mathcal{G} = (x+uh,y) \mathcal{G} = \Big( (x+uh)\cdot(w_1+uw_2) +uy\cdot w_3 \Big)_{(w_1+uw_2,w_3) \in \mathcal{D}},$ which, by \eqref{Code2}, gives $\mathscr{D} \subseteq \mathscr{C}_{\mathcal{D}}.$ This proves the assertion that $ \mathscr{C}_{\mathcal{D}}=\mathscr{D}.$  \end{proof}
In the following lemma, we establish a connection between spanning sets of a linear code   over $\frac{\mathbb{F}_q[u]}{\langle u^2 \rangle}$ and its Gray image under the map $\Phi.$
\begin{lemma}\label{Lem7}
Let $\mathscr{D}$ be a linear code of length $\mathrm{n}$ over $\frac{\mathbb{F}_q[u]}{\langle u^2 \rangle}$ with a spanning set $\mathcal{I} = \{R_1, R_2, \ldots, R_k \}.$ The Gray image $\Phi(\mathscr{D})$ is a linear code of length $2\mathrm{n}$ over $\mathbb{F}_q$ with a spanning set $\mathcal{J}=\{\Phi(R_1), \Phi(R_2), \ldots, \Phi(R_k),\break \Phi(uR_1), \Phi(uR_2), \ldots, \Phi(uR_k) \}.$
\end{lemma}
\begin{proof}
To prove the result, we first recall, from Section \ref{Prelim3}, that the Gray map $\Phi$  induces an $\mathbb{F}_q$-linear isometry between the spaces $\big(\big(\frac{\mathbb{F}_q[u]}{\langle u^2 \rangle}\big)^\mathrm{n}, wt_L\big)$ and $\big(\mathbb{F}_q^{2\mathrm{n}}, wt_H\big).$
Now, let $\mathscr{B}$ be a linear code of length $2\mathrm{n}$ over $\mathbb{F}_q$ with a spanning set
$\mathcal{J} = \{\Phi(R_1), \Phi(R_2), \ldots, \Phi(R_k), \Phi(uR_1), \Phi(uR_2), \ldots, \Phi(uR_k)\}.$

We first note, for $1 \leq i \leq k,$ that both $R_i, uR_i \in \mathscr{D},$ and hence both $\Phi(R_i), \Phi(uR_i) \in \Phi(\mathscr{D}).$ This implies that $\mathscr{B} \subseteq \Phi(\mathscr{D}).$
On the other hand, let us take $x \in \Phi(\mathscr{D}).$ Then there exists a codeword $y \in \mathscr{D}$ such that $x = \Phi(y).$ Since $\mathcal{I}$ is a spanning set of the code $\mathscr{D},$ the codeword $y$ can be expressed as
\begin{equation}\label{y}
   y = (a_1 + u b_1) R_1 + (a_2 + u b_2) R_2 + \cdots + (a_k + u b_k) R_k 
\end{equation}
for some $a_1,a_2,\ldots,a_k, b_1, b_2, \ldots, b_k \in \mathbb{F}_q.$  Using the fact that the Gray map $\Phi$ is $\mathbb{F}_q$-linear and equation \eqref{y}, we obtain
\begin{equation*}
   x=\Phi(y) = a_1 \Phi(R_1) + b_1 \Phi(uR_1) + a_2 \Phi(R_2) + b_2 \Phi(uR_2) + \cdots + a_k \Phi(R_k) + b_k \Phi(uR_k). 
\end{equation*}
This implies that $x = \Phi(y)$ belongs to  the $\mathbb{F}_q$-linear span of $\mathcal{J},$ and hence $x \in \mathscr{B}.$ Thus, we have $\Phi(\mathscr{D}) \subseteq \mathscr{B}.$ 
\\This proves  that $\Phi(\mathscr{D}) = \mathscr{B}.$ 
\end{proof}

Now, let us define two maps $\pi_1$ and $\pi_2$ from $\mathcal{R}^m$ into $ \mathbb{F}_q^{3m}$  as
\begin{equation}\label{pi}
  \pi_1(w_1+uw_2,w_3) =(w_2,w_3,w_1) \text{ and }  \pi_2(w_1+uw_2,w_3) =(w_1+w_2,w_3,w_1) 
\end{equation}for all $(w_1+uw_2,w_3)\in\mathcal{R}^{m}.$ In the following theorem, we determine a spanning matrix of the Gray image $\Phi(\mathscr{C}_{\mathcal{D}})$  of a linear code $\mathscr{C}_{\mathcal{D}}$ over $\frac{\mathbb{F}_q[u]}{\langle u^2 \rangle},$ with defining set $\mathcal{D}\subseteq \mathcal{R}^m,$   in terms of the set $\mathcal{D}$  and  the maps $\pi_1$ and  $\pi_2.$ 
\begin{thm}\label{Th6} Let $\mathcal{D}$ be a non-empty subset of $\mathcal{R}^m,$ and let $\mathscr{C}_{\mathcal{D}}$ be the corresponding linear code over $\frac{\mathbb{F}_q[u]}{\langle u^2\rangle}$ with defining set $\mathcal{D}.$ Then the Gray image $\Phi(\mathscr{C}_{\mathcal{D}})$ is a linear code of  length $2|\mathcal{D}|$ over $\mathbb{F}_q$ with a spanning matrix $\mathscr{G}\in M_{3m \times 2|\mathcal{D}|}(\mathbb{F}_q) $ whose columns are $\pi_1(s)^T \in M_{3m \times 1}(\mathbb{F}_q)$ and $\pi_2(s)^T\in M_{3m \times 1}(\mathbb{F}_q)$ arranged in consecutive odd and even positions, respectively, as $s$ ranges over the elements of $\mathcal{D}.$

\end{thm}
\begin{proof}
To prove the result, we see, by Theorem \ref{Th5}, that the code $\mathscr{C}_{\mathcal{D}}$ has a spanning matrix $\mathcal{G}$ whose columns are the vectors $\lambda(s)^T \in M_{2m\times 1}\big(\frac{\mathbb{F}_q[u]}{\langle u^2 \rangle}\big),$ where $s$ runs over the elements of the defining set $\mathcal{D}.$  Now, let $R_j$ denote the $j$-th row of the matrix $\mathcal{G}$ for $1 \leq j \leq 2m.$ We note, by \eqref{lambda}, that the row $R_j$ is a multiple of $u$ for all $m+1 \leq j \leq 2m.$
Since the rows of the matrix $\mathcal{G}$ form a spanning set of the code $\mathscr{C}_{\mathcal{D}},$ we see, by Lemma \ref{Lem7}, that the set
$T = \{ \Phi(R_1),\Phi(R_2), \ldots, \Phi(R_m), \Phi(u R_1),\Phi(u R_2), \ldots, \Phi(u R_m) \}$
forms a spanning set of the code $\Phi(\mathscr{C}_{\mathcal{D}}).$ 
Moreover, one can easily see, by \eqref{pi}, that the rows of the matrix $\mathscr{G}$ are precisely the elements of the set $T.$ This shows that $\mathscr{G}$ is a spanning matrix of the code $\Phi(\mathscr{C}_{\mathcal{D}}).$ 
\end{proof}

In the following section, we will  construct four infinite families of linear codes over  $\frac{\mathbb{F}_q[u]}{\langle u^2 \rangle},$ whose defining sets are certain subsets of $\mathcal{R}^m.$ We will also explicitly determine their parameters and Lee weight distributions.

\section{Four infinite families of linear codes over $\frac{\mathbb{F}_q[u]}{\langle u^2 \rangle}$}\label{Newcodes}
Throughout this paper, let $\mathpzc{A},\ \mathpzc{B},$ and $\mathpzc{C}$ be non-empty subsets of $[m] := \{1, 2, \ldots, m\}.$ Let $\Delta_{\mathpzc{A}},$ $\Delta_{\mathpzc{B}},$ and $\Delta_{\mathpzc{C}}$ be the simplicial complexes of $\mathbb{F}_q^m$ with supports $\mathpzc{A},$ $\mathpzc{B},$ and $\mathpzc{C},$ respectively. Now, let us define the following four subsets of $\mathcal{R}^m$: 
\begin{eqnarray}\label{S1}
\mathcal{S}_1 &=& \{(w_1+uw_2,w_3)\in \mathcal{R}^m : w_1 \in \Delta_{\mathpzc{A}}, w_2 \in \Delta_{\mathpzc{B}},w_3 \in \Delta_{\mathpzc{C}}^c  \text{ with }|\mathpzc{C}|<m\},\\
\label{S2}
\mathcal{S}_2 &=& \{(w_1+uw_2,w_3)\in \mathcal{R}^m : w_1 \in \Delta_{\mathpzc{A}}^c, w_2 \in \Delta_{\mathpzc{B}},w_3 \in \Delta_{\mathpzc{C}} \text{ with }|\mathpzc{A}|<m\},\\
\label{S3}
\mathcal{S}_3 &=& \{(w_1+uw_2,w_3)\in \mathcal{R}^m : w_1 \in \Delta_{\mathpzc{A}}, w_2 \in \Delta_{\mathpzc{B}}^c, w_3 \in \Delta_{\mathpzc{C}} \text{ with }|\mathpzc{B}|<m\}, 
\text{ and }\\
\label{S4}
\mathcal{S}_4 &=& \{(w_1+uw_2,w_3)\in \mathcal{R}^m : w_1 \in \Delta_{\mathpzc{A}}^\ast, w_2 \in \Delta_{\mathpzc{B}}, w_3 \in \Delta_{\mathpzc{C}}^c \text{ with }|\mathpzc{C}|<m \text{ and } |\mathpzc{A}|\geq 2\}.
\end{eqnarray}
\begin{remark}\label{Remark3.01}
The exclusion of $\{\mathbf{0}\}$ from $\Delta_{\mathpzc{A}}$ in the definition of
$\mathcal{S}_4$ is imposed to avoid vectors whose first component is of the form
$uw_2,$ which leads to codes with parameters distinct from those obtained using the defining set $\mathcal{S}_1$ (see Theorems \ref{Th1} and \ref{Th4}).  The defining set $\mathcal{S}_4$ will also play a role  in Section \ref{Sec5}, where  a certain subset of
$\mathcal{S}_4$ is used  to construct projective codes, which are shown to have locality either $2$ or $3$ in  Theorem \ref{Thm5.13}.
\end{remark}

For $1 \leq i \leq 4,$ let $\mathscr{C}_{\mathcal{S}_i}$ be a linear code of length $|\mathcal{S}_i|$ over $\frac{\mathbb{F}_q[u]}{\langle u^2 \rangle}$ with defining set $\mathcal{S}_i \subseteq \mathcal{R}^m,$ as defined by \eqref{Code2}. Let us define a map $\mathpzc{T}_i : \mathcal{R}^m \rightarrow \mathscr{C}_{\mathcal{S}_i}$ as 
\begin{equation*}
\mathpzc{T}_i(r) =c_r= \big( \langle r,s \rangle \big)_{s \in \mathcal{S}_i}
\end{equation*}
for all $r \in \mathcal{R}^m.$
Since the map $\langle \cdot, \cdot \rangle$ (as defined by \eqref{IP}) is a bilinear form, the map $\mathpzc{T}_i$ is a surjective $\big(\frac{\mathbb{F}_q[u]}{\langle u^2 \rangle}\big)$-module homomorphism. This implies that
\begin{equation*}
|\mathscr{C}_{\mathcal{S}_i}| = \frac{q^{3m}}{|\ker(\mathpzc{T}_i)|},
\end{equation*} where $\ker(\mathpzc{T}_i)$ denotes the kernel of the map $\mathpzc{T}_i.$ Moreover, as noted in Section \ref{Prelim3}, the image $\Phi(\mathscr{C}_{\mathcal{S}_i})$ is a linear code of length $2|\mathcal{S}_i|$ over $\mathbb{F}_q.$

In a recent work, Mondal and Lee \cite{Mondal2024} focused on the case   $q = 2$ and investigated the linear codes $\mathscr{C}_{\mathcal{S}_2}$ and $\mathscr{C}_{\mathcal{S}_3}$ over $\frac{\mathbb{F}_2[u]}{\langle u^2 \rangle}.$ In this paper, we will consider the case $q \geq 2$ and  extend the techniques employed by Mondal and Lee \cite{Mondal2024} to study the linear codes $\mathscr{C}_{\mathcal{S}_1},$ $\mathscr{C}_{\mathcal{S}_2},$ $\mathscr{C}_{\mathcal{S}_3},$ and $\mathscr{C}_{\mathcal{S}_4}$ over  $\frac{\mathbb{F}_q[u]}{\langle u^2 \rangle}.$   We will also determine their parameters and Lee weight distributions,
 and examine their Gray images. 

First of all, we see that the finite field  $\mathbb{F}_{q^m}$ of order $q^m$ can also be viewed as a vector space of dimension $m$ over $\mathbb{F}_q.$ More precisely, for an ordered basis $\bm{\alpha} = \{\alpha_1<\alpha_2<\cdots<\alpha_m\}$  of $\mathbb{F}_{q^m}$ over $\mathbb{F}_q,$ the map $\psi: \mathbb{F}_{q}^m \rightarrow \mathbb{F}_{q^m},$ defined as
$$\psi(v) = \sum\limits_{i=1}^m \alpha_i v_i \text{ for all } v=(v_1,v_2,\ldots,v_m) \in \mathbb{F}_q^m,$$ is an $\mathbb{F}_q$-linear isomorphism.  We further see, by Theorem 2.29 of \cite{Lidl1994}, that there exists a unique trace dual ordered basis $\bm{\beta} = \{\beta_1< \beta_2< \cdots< \beta_m \}$ of $\mathbb{F}_{q^m}$ over $\mathbb{F}_q,$ $\ie$ $Tr(\alpha_i\beta_j)=\delta_{ij}$ for $1 \leq i,j \leq m,$ where $\delta_{ij}$ denotes the Kronecker delta function and $Tr(\cdot)$ denotes  the trace function from $\mathbb{F}_{q^m}$ onto $\mathbb{F}_q.$ We next observe that each element $w \in \mathbb{F}_{q^m}$ can be uniquely written as  $w = \sum\limits_{i=1}^m Tr(w\beta_i)\alpha_i.$ Accordingly,  the support of the element $w \in \mathbb{F}_{q^m}$ with respect to the 
ordered basis $\bm{\alpha}$ is defined as $\mathcal{S}_{\bm{\alpha}}(w) = \{ i \in [m] : Tr(w\beta_i) \neq 0 \}.$ 

Moreover, for a non-empty subset $\mathcal{D}$ of $\mathcal{R}^m, $ let us define the corresponding subset of $\mathtt{R}=\frac{\mathbb{F}_{q^m}[u]}{\langle u^2 \rangle} \times \mathbb{F}_{q^m}$ as follows: $$\psi(\mathcal{D})=\big\{\big(\psi(w_1)+u\,\psi(w_2), \psi(w_3)\big): (w_1+uw_2, w_3) \in \mathcal{D} \big\},$$ where $\frac{\mathbb{F}_{q^m}[u]}{\langle u^2 \rangle}$ is the quasi-Galois ring with  maximal ideal $\langle u \rangle$ of nilpotency index $2$ and  residue field $\mathbb{F}_{q^m}.$ Note that $\mathtt{R}=\frac{\mathbb{F}_{q^m}[u]}{\langle u^2 \rangle} \times \mathbb{F}_{q^m}$ can be viewed as a module over $\frac{\mathbb{F}_{q^m}[u]}{\langle u^2 \rangle}.$ From now on, we will represent the elements of $\mathtt{R}$ as $(\theta_1+u\theta_2,\theta_3),$ where $\theta_1,\theta_2,\theta_3 \in \mathbb{F}_{q^m}.$ We further define \begin{equation*}
\mathcal{C}_{\psi(\mathcal{D})}=\left\{\big(Tr( \theta_1\delta_1)+u \, Tr( \theta_2\delta_1+ \theta_1\delta_2+ \theta_3\delta_3)\big)_{(\delta_1+u\delta_2, \delta_3) \in \psi(\mathcal{D})}: (\theta_1+u \theta_2,\theta_3) \in \mathtt{R} \right\}.\end{equation*} It is easy to see that the code $\mathcal{C}_{\psi(\mathcal{D})}$ is a linear code over $\frac{\mathbb{F}_q[u]}{\langle u^2 \rangle}.$ Furthermore, working as in Section II(A) of Luo and Cao \cite{Luo2021}, one can easily see that \begin{equation*}
\mathscr{C}_{\mathcal{D}}=\mathcal{C}_{\psi(\mathcal{D})}.\end{equation*}

In view of this observation,  throughout this section,  we shall identify the elements of $\mathbb{F}_q^m$ with those of $\mathbb{F}_{q^m},$  via the $\mathbb{F}_q$-linear isomorphism $\psi.$ Accordingly,  the simplicial complex $\Delta_A $ of $\mathbb{F}_{q}^m,$ with support $A \subseteq [m],$ will be regarded as a subset of $\mathbb{F}_{q^m}$ and referred to as a simplicial complex of $\mathbb{F}_{q^m}.$ Under this identification, the corresponding simplicial complex $\Delta_A$ of $\mathbb{F}_{q^m}$ has  the same support $A$ with respect to the ordered basis $\bm{\alpha}.$  Unless stated otherwise, we shall henceforth consider the support of a simplicial complex of $\mathbb{F}_{q^m}$ with respect to the ordered basis $\bm{\alpha}.$ Moreover, we see, for $1 \leq i \leq 4,$ that the code $\mathscr{C}_{\mathcal{S}_i}$ can be viewed as the image of the surjective $\big(\frac{\mathbb{F}_{q^m}[u]}{\langle u^2 \rangle}\big)$-linear homomorphism $\mu_i$ from $\mathtt{R} $ onto $\mathscr{C}_{\mathcal{S}_i},$ defined by \begin{equation*}\label{mu}\mu_i(\theta)=\mathtt{c}_{\theta}:= \Big(Tr( \theta_1 \delta_1) + u \, Tr( \theta_2 \delta_1 +  \theta_1 \delta_2 + \theta_3 \delta_3)\Big)_{(\delta_1+u\delta_2, \delta_3) \in \psi(\mathcal{S}_i)}\end{equation*} for all $\theta=(\theta_1+u\,\theta_2, \theta_3) \in \mathtt{R}$ with $\theta_1,\theta_2,\theta_3 \in \mathbb{F}_{q^m}.$ This implies, for $1 \leq i \leq 4,$ that
\begin{equation}\label{Nsize}
|\mathscr{C}_{\mathcal{S}_i}| = \frac{q^{3m}}{|\ker(\mu_i)|},
\end{equation} where $\ker(\mu_i)$ denotes the kernel of the map $\mu_i.$  Regardless of whether elements of $\mathbb{F}_q^m$ are identified with those of $\mathbb{F}_{q^m},$ we shall henceforth denote the linear codes over $\frac{\mathbb{F}_q[u]}{\langle u^2 \rangle}$ with defining sets $\mathcal{S}_1,$ $\mathcal{S}_2,$ $\mathcal{S}_3,$ and $\mathcal{S}_4$ by $\mathscr{C}_{\mathcal{S}_1},$ $\mathscr{C}_{\mathcal{S}_2},$ $\mathscr{C}_{\mathcal{S}_3},$ and $\mathscr{C}_{\mathcal{S}_4},$ respectively.

Next, let $M $ be a $k$-dimensional $\mathbb{F}_q$-linear subspace of $\mathbb{F}_{q^m}.$ The trace dual of $M,$ denoted by $M^{\perp_{Tr}},$ is defined as $$M^{\perp_{Tr}} = \{w \in \mathbb{F}_{q^m} : Tr(wb)=0 \text{ for all } b\in M\}.$$
We note, by Proposition 2.4 of \cite{Grove2008}, that $M^{\perp_{Tr}}$ is an $(m-k)$-dimensional $\mathbb{F}_q$-linear subspace of $\mathbb{F}_{q^m}.$ We next observe the following:
\begin{remark}\label{KRK}Let $\Delta_{\mathpzc{P}}$ be the simplicial complex of $\mathbb{F}_{q^m}$ with support $\mathpzc{P} \subseteq [m].$ Here, one can easily see that $\Delta_{\mathpzc{P}}$ has a basis $\{\alpha_i : i\in \mathpzc{P} \}$ and its trace dual $\Delta_{\mathpzc{P}}^{\perp_{Tr}}$ has a basis $\{\beta_j : j\in [m]\setminus \mathpzc{P} \}.$ From this, it follows that $b\in \Delta_{\mathpzc{P}}^{\perp_{Tr}}$ if and only if $\mathcal{S}_{\bm{\beta}}(b) \cap \mathpzc{P} = \emptyset,$ where $\mathcal{S}_{\bm{\beta}}(b) = 
\{ i \in [m] : Tr(b\alpha_i) \neq 0 \}$ is the support of the element $b \in \mathbb{F}_{q^m}$ with respect to the 
ordered basis $\bm{\beta}.$ \end{remark}

 Further, let $\chi(\cdot)$ denote the canonical additive character of $\mathbb{F}_q,$ $\ie$ $\chi(x) = \zeta^{Tr_p^q(x)}$ for all $x \in \mathbb{F}_q,$ where $\zeta$ is a complex primitive $p$-th root of unity and $Tr_p^q(\cdot)$ denotes the trace function from $\mathbb{F}_q$ onto $\mathbb{F}_p.$ We need the following lemma involving some special character sums.

\begin{lemma}\cite[p. 4908]{Yang2015} \label{Lem5}    Let $\mathpzc{P}$ be a non-empty subset of $  [m],$ and let $\Delta_{\mathpzc{P}}$ be the simplicial complex of $\mathbb{F}_{q^m}$ with support  $\mathpzc{P}.$  For  $b \in \mathbb{F}_{q^m},$ let us define $$F_b(\Delta_{\mathpzc{P}})=\sum\limits_{w \in \Delta_{\mathpzc{P}}} \chi \big( Tr(bw) \big)\text{~ and ~} F_b(\Delta_{\mathpzc{P}}^c)=\sum\limits_{w \in \Delta_{\mathpzc{P}}^c} \chi \big( Tr(bw) \big).$$
The following hold.
\begin{itemize}
\item[(a)]  For $b  \in \mathbb{F}_{q^m},$ we have $$F_b(\Delta_{\mathpzc{P}})= \left\{\begin{array}{cl} q^{|\mathpzc{P}|} & \text{if } \mathcal{S}_{\bm{\beta}}(b) \cap \mathpzc{P} = \emptyset;\\
0 & \text{otherwise,}
\end{array}\right. \text{
~and ~} F_b(\Delta_{\mathpzc{P}}^c)= \left\{\begin{array}{cl} q^m-q^{|\mathpzc{P}|} & \text{if } b=0;\\-q^{|\mathpzc{P}|} & \text{if } b\neq 0 \text{ and } \mathcal{S}_{\bm{\beta}}(b) \cap \mathpzc{P} = \emptyset;\\
0 & \text{otherwise.}
\end{array}\right.$$
\item[(b)]  For all non-zero $\gamma \in \mathbb{F}_q,$ we have \begin{eqnarray*}F_b(\Delta_{\mathpzc{P}})&=&F_b(\gamma \Delta_{\mathpzc{P}}) = F_{\gamma b}(\Delta_{\mathpzc{P}}) \text{ and}\\F_b(\Delta_{\mathpzc{P}}^c)&=&F_b(\gamma \Delta_{\mathpzc{P}}^c)=F_{\gamma b}(\Delta_{\mathpzc{P}}^c),\end{eqnarray*} where $\gamma \Delta_{\mathpzc{P}} = \{\gamma w : w \in \Delta_{\mathpzc{P}}\}$ and $\gamma \Delta_{\mathpzc{P}}^c = \{\gamma w : w \in \Delta_{\mathpzc{P}}^c\}.$
\end{itemize}
\end{lemma}
\begin{proof} \begin{description}\item[(a)] First of all, using equation (13) of Yang \etal \cite{Yang2015} and Remark \ref{KRK}, we obtain  \begin{equation}\label{sum1} F_b(\Delta_{\mathpzc{P}})= \left\{\begin{array}{cl} q^{|\mathpzc{P}|} & \text{if } \mathcal{S}_{\bm{\beta}}(b) \cap \mathpzc{P} = \emptyset;\\
0 & \text{otherwise.}
\end{array}\right.  \end{equation}

Next, to determine the sum $F_b(\Delta_{\mathpzc{P}}^c),$ one can easily see that $F_b(\Delta_{\mathpzc{P}}^c)=|\Delta_{\mathpzc{P}}^c |=q^m-q^{|\mathpzc{P}|}$ when $b=0.$ On the other hand,  for all  non-zero $b  \in \mathbb{F}_{q^m},$  we see that $F_b(\Delta_{\mathpzc{P}}) +F_b(\Delta_{\mathpzc{P}}^c)=\sum\limits_{w \in \mathbb{F}_{q^m}} \chi \big( Tr(bw) \big)= 0,$  which gives $ F_b(\Delta_{\mathpzc{P}}) = -F_b(\Delta_{\mathpzc{P}}^c).$ From this and using \eqref{sum1},  we get the desired result.

\item[(b)] To prove part (b), one can easily see that  $\Delta_{\mathpzc{P}}$ is an $\mathbb{F}_q$-linear subspace of $\mathbb{F}_{q^m}.$ This implies that $\gamma \Delta_{\mathpzc{P}} =\Delta_{\mathpzc{P}}$ and $\gamma \Delta_{\mathpzc{P}}^c = \Delta_{\mathpzc{P}}^c$ for all non-zero $\gamma \in \mathbb{F}_q.$ From this and using the $\mathbb{F}_q$-linearity of the trace function $Tr(\cdot),$ part (b) follows.\vspace{-2mm}\end{description}
\vspace{-2mm}\end{proof}
In the following theorem, we determine the parameters and  Lee weight distribution of the code $\mathscr{C}_{\mathcal{S}_1}$ over $\frac{\mathbb{F}_q[u]}{\langle u^2 \rangle}.$ 
\begin{thm}\label{Th1}
The code $\mathscr{C}_{\mathcal{S}_1}$  is a linear code over $\frac{\mathbb{F}_q[u]}{\langle u^2\rangle}$ with parameters 
$\big(q^{|\mathpzc{A}|+|\mathpzc{B}|}(q^m-q^{|\mathpzc{C}|}),~ q^{m+|\mathpzc{A}|+|\mathpzc{A}\cup\mathpzc{B}|},~ \epsilon(q-1)q^{|\mathpzc{A}|+|\mathpzc{B}|-1}(q^m-q^{|\mathpzc{C}|})\big )$ and  Lee weight distribution given by Table \ref{T1}, where $\epsilon=2$ if $\mathpzc{A}\subseteq \mathpzc{B},$ while $\epsilon=1$ if $\mathpzc{A}\not\subseteq \mathpzc{B}.$ As a consequence, $\mathscr{C}_{\mathcal{S}_1}$ is a $2$-weight code if $\mathpzc{A}\subseteq \mathpzc{B},$ whereas it  is a $4$-weight code if $\mathpzc{A}\not\subseteq \mathpzc{B}.$
\begin{table}[H]
\centering
    \begin{tabular}{ |c|c|} 
 \hline Lee (\textit{resp.} Hamming) weight & Frequency  \\ \hline
$0$ & $1$ \\
 \hline
$(q-1)q^{|\mathpzc{A}|+|\mathpzc{B}|-1}(q^m-q^{|\mathpzc{C}|})$ & $2(q^{|\mathpzc{A} \cup \mathpzc{B}|-|\mathpzc{B}|}-1)$ \\
 \hline
$2(q-1)q^{|\mathpzc{A}|+|\mathpzc{B}|-1}(q^m-q^{|\mathpzc{C}|})$ & $q^{m+|\mathpzc{A}|+|\mathpzc{A}\cup \mathpzc{B}|} - 2q^{m-|\mathpzc{C}|+|\mathpzc{A} \cup \mathpzc{B}| - |\mathpzc{B}|} + q^{m-|\mathpzc{C}|}$ \\
 \hline
$(q-1)q^{|\mathpzc{A}|+|\mathpzc{B}|-1}(2q^m-q^{|\mathpzc{C}|})$ & $2(q^{|\mathpzc{A} \cup \mathpzc{B}|-|\mathpzc{B}|}-1)(q^{m-|\mathpzc{C}|}-1)$ \\
 \hline
$2(q-1)q^{m+|\mathpzc{A}|+|\mathpzc{B}|-1}$ & $q^{m-|\mathpzc{C}|}-1$ \\
 \hline
\end{tabular}
\caption{The Lee (\textit{resp.} Hamming) weight distribution of the code $\mathscr{C}_{\mathcal{S}_1}$ (\textit{resp.} $\Phi(\mathscr{C}_{\mathcal{S}_1})$) }\label{T1} 
\end{table}
\end{thm}

\begin{proof}
To prove the result, we first recall that the code $\mathscr{C}_{\mathcal{S}_1}$  is a linear code of length $|\mathcal{S}_1|$ over $\frac{\mathbb{F}_q[u]}{\langle u^2 \rangle},$ where the set $\mathcal{S}_1$ is as defined by \eqref{S1} and the code $\mathscr{C}_{\mathcal{S}_1}$ is as defined by  \eqref{Code2}. Further, by \eqref{Scard} and \eqref{S1}, we have \begin{equation}\label{EQS1}|\mathcal{S}_1| =  |\Delta_{\mathpzc{A}}| |\Delta_{\mathpzc{B}}||\Delta_{\mathpzc{C}}^c| =  q^{|\mathpzc{A}| + |\mathpzc{B}|}(q^m - q^{|\mathpzc{C}|}).\end{equation} Thus, $\mathscr{C}_{\mathcal{S}_1}$ is a linear code of length $ q^{|\mathpzc{A}| + |\mathpzc{B}|}(q^m - q^{|\mathpzc{C}|})$ over $\frac{\mathbb{F}_q[u]}{\langle u^2 \rangle}.$

Now, to determine the size of the code $\mathscr{C}_{\mathcal{S}_1},$ we see, by \eqref{Nsize}, that it suffices to determine $|\ker(\mu_1)|,$ where $\mu_1 : \mathtt{R} \rightarrow \mathscr{C}_{\mathcal{S}_1}$ is a surjective $\big(\frac{\mathbb{F}_q[u]}{\langle u^2 \rangle}\big)$-module homomorphism, defined by 

 $$\mu_1(\theta) = \mathtt{c}_{\theta}:= \big(Tr( \theta_1 \delta_1) + u \, Tr( \theta_2\delta_1 +  \theta_1 \delta_2 +  \theta_3 \delta_3)\big)_{(\delta_1+u\delta_2, \delta_3) \in \mathcal{S}_1}$$
for all $\theta=(\theta_1+u\theta_2, \theta_3) \in \mathtt{R}.$  From this and using equation \eqref{GIe}, we obtain
  \begin{align*}
 wt_L(\mathtt{c}_{\theta}) & =  \sum\limits_{\delta_1 \in \Delta_\mathpzc{A}} \sum\limits_{\delta_2 \in \Delta_\mathpzc{B}}\sum\limits_{\delta_3 \in \Delta_\mathpzc{C}^c} \bigg( wt_H\Big( Tr( \theta_2\delta_1 +  \theta_1\delta_2 +  \theta_3\delta_3)\Big) +  wt_H\Big( Tr( \theta_2\delta_1 +  \theta_1\delta_2 +  \theta_3\delta_3 + \theta_1\delta_1)\Big) \bigg)\\
 & = |\mathcal{S}_1| - \frac{1}{q}\sum\limits_{\gamma \in \mathbb{F}_q}  \sum\limits_{\delta_1 \in \Delta_\mathpzc{A}} \sum\limits_{\delta_2 \in \Delta_\mathpzc{B}}\sum\limits_{\delta_3 \in \Delta_\mathpzc{C}^c} \chi\big(\gamma Tr(\theta_3\delta_3+\theta_2\delta_1+\theta_1\delta_2) \big) 
 \\ & ~~+  |\mathcal{S}_1| - \frac{1}{q}\sum\limits_{\gamma \in \mathbb{F}_q}  \sum\limits_{\delta_1 \in \Delta_\mathpzc{A}} \sum\limits_{\delta_2 \in \Delta_\mathpzc{B}}\sum\limits_{\delta_3 \in \Delta_\mathpzc{C}^c} \chi\big(\gamma Tr(\theta_3\delta_3+(\theta_1+\theta_2)\delta_1+\theta_1\delta_2) \big) \\
& = \frac{2(q-1)}{q}|\mathcal{S}_1| - \frac{1}{q}\sum\limits_{\gamma \in \mathbb{F}_q^\ast}  \sum\limits_{\delta_1 \in \Delta_\mathpzc{A}} \chi\big(\gamma Tr(\theta_2\delta_1)\big) \sum\limits_{\delta_2 \in \Delta_\mathpzc{B}} \chi\big(\gamma Tr(\theta_1\delta_2)\big)\sum\limits_{\delta_3 \in \Delta_\mathpzc{C}^c} \chi\big(\gamma Tr(\theta_3\delta_3)\big)  \\ & ~~ - \frac{1}{q}\sum\limits_{\gamma \in \mathbb{F}_q^\ast}  \sum\limits_{\delta_1 \in \Delta_\mathpzc{A}} \chi\big(\gamma Tr((\theta_1+\theta_2)\delta_1)\big) \sum\limits_{\delta_2 \in \Delta_\mathpzc{B}} \chi\big(\gamma Tr(\theta_1\delta_2)\big)\sum\limits_{\delta_3 \in \Delta_\mathpzc{C}^c} \chi\big(\gamma Tr(\theta_3\delta_3)\big),
\end{align*} where $\mathbb{F}_q^\ast=\mathbb{F}_q\setminus \{0\}.$ This, by Lemma \ref{Lem5}(b), implies that
\begin{equation}\label{Eq3.1}
wt_L(\mathtt{c}_{\theta}) = \frac{(q-1)}{q} \Big(2|\mathcal{S}_1| - F_{\theta_1}(\Delta_\mathpzc{B})F_{\theta_3}(\Delta_\mathpzc{C}^c) \big(F_{\theta_2}(\Delta_\mathpzc{A})+F_{\theta_1+\theta_2}(\Delta_\mathpzc{A})\big) \Big).
\end{equation} 
Consequently, we have $wt_L(\mathtt{c}_{\theta})=0$ if and only if
$2|\mathcal{S}_1| =  F_{\theta_1}(\Delta_\mathpzc{B})F_{\theta_3}(\Delta_\mathpzc{C}^c)\big(F_{\theta_2}(\Delta_\mathpzc{A})+F_{\theta_1+\theta_2}(\Delta_\mathpzc{A})\big) ,$ which,   by Lemma \ref{Lem5}(a) and using equation \eqref{EQS1},  holds if and only if   $F_{\theta_1}(\Delta_\mathpzc{B}) = |\Delta_\mathpzc{B}|,$ $F_{\theta_3}(\Delta_\mathpzc{C}^c)=|\Delta_\mathpzc{C}^c|,$ and $F_{\theta_2}(\Delta_\mathpzc{A})+F_{\theta_1+\theta_2}(\Delta_\mathpzc{A}) = 2|\Delta_\mathpzc{A}|.$  Now, by Lemma \ref{Lem5}(a), we get
\begin{equation}\label{Ker}
    \ker(\mu_1)=\{ (\theta_1+u\theta_2,\theta_3) \in \mathtt{R} : \mathcal{S}_{\bm{\beta}}(\theta_1) \cap (\mathpzc{A} \cup \mathpzc{B}) = \emptyset,~ \mathcal{S}_{\bm{\beta}}(\theta_2) \cap \mathpzc{A} = \emptyset,~ \text{ and }~\theta_3 = 0 \}.
\end{equation}
This, by  Lemma \ref{Lem3},  implies that \begin{equation}\label{kermu}|\ker(\mu_1)| = |\mathpzc{X}_{\mathpzc{A} \cup \mathpzc{B}}||\mathpzc{X}_{\mathpzc{A}}|=q^{2m-|\mathpzc{A}|-|\mathpzc{A}\cup \mathpzc{B}|}.\end{equation} From this and using equation \eqref{Nsize}, we get $|\mathscr{C}_{\mathcal{S}_1}|=q^{m+|\mathpzc{A}|+|\mathpzc{A} \cup \mathpzc{B}|}.$

Next, to determine the Lee weight distribution of the code $\mathscr{C}_{\mathcal{S}_1},$ we assume, throughout  the proof, that $\theta = (\theta_1 + u \theta_2, \theta_3) \in \mathtt{R}$ is such that $\theta \notin \ker(\mu_1).$ Using \eqref{Ker} and Lemma \ref{Lem5}(a),  we note that the element $\theta = (\theta_1 + u \theta_2, \theta_3) \in \mathtt{R}$ satisfies exactly one of the following eight conditions:  
\begin{itemize}
    \item[(I)]  $\mathcal{S}_{\bm{\beta}}(\theta_1) \cap (\mathpzc{A}\cup \mathpzc{B}) = \emptyset,$ $\mathcal{S}_{\bm{\beta}}(\theta_2) \cap \mathpzc{A} \neq \emptyset,$ and $\theta_3=0.$
    \item[(II)]  $\mathcal{S}_{\bm{\beta}}(\theta_1) \cap (\mathpzc{A}\cup \mathpzc{B}) \neq \emptyset,$ $\mathcal{S}_{\bm{\beta}}(\theta_2) \cap \mathpzc{A} = \emptyset,$ and $\theta_3=0.$ 
    \item[(III)]  $\mathcal{S}_{\bm{\beta}}(\theta_1) \cap (\mathpzc{A}\cup \mathpzc{B}) \neq \emptyset,$ $\mathcal{S}_{\bm{\beta}}(\theta_2) \cap \mathpzc{A} \neq \emptyset,$ and $\theta_3=0.$
    \item[(IV)] $\theta_3 \neq 0$ and $\mathcal{S}_{\bm{\beta}}(\theta_3) \cap \mathpzc{C} \neq \emptyset.$ 
    \item[(V)]  $\mathcal{S}_{\bm{\beta}}(\theta_1) \cap (\mathpzc{A}\cup \mathpzc{B}) = \emptyset,$ $\mathcal{S}_{\bm{\beta}}(\theta_2) \cap \mathpzc{A} \neq \emptyset,$ $\theta_3 \neq 0,$ and $\mathcal{S}_{\bm{\beta}}(\theta_3) \cap \mathpzc{C} = \emptyset.$  
    \item[(VI)]  $\mathcal{S}_{\bm{\beta}}(\theta_1) \cap (\mathpzc{A} \cup \mathpzc{B}) \neq \emptyset,$ $\mathcal{S}_{\bm{\beta}}(\theta_2) \cap \mathpzc{A} = \emptyset,$ $\theta_3 \neq 0,$ and $\mathcal{S}_{\bm{\beta}}(\theta_3) \cap \mathpzc{C} = \emptyset.$ 
    \item[(VII)]  $\mathcal{S}_{\bm{\beta}}(\theta_1) \cap (\mathpzc{A} \cup \mathpzc{B}) \neq \emptyset,$ $\mathcal{S}_{\bm{\beta}}(\theta_2) \cap \mathpzc{A} \neq \emptyset,$ $\theta_3 \neq 0,$ and $\mathcal{S}_{\bm{\beta}}(\theta_3) \cap \mathpzc{C} = \emptyset.$ 
    \item[(VIII)]  $\mathcal{S}_{\bm{\beta}}(\theta_1) \cap (\mathpzc{A} \cup \mathpzc{B}) = \emptyset,$ $\mathcal{S}_{\bm{\beta}}(\theta_2) \cap \mathpzc{A} = \emptyset,$ $\theta_3 \neq 0,$ and $\mathcal{S}_{\bm{\beta}}(\theta_3) \cap \mathpzc{C} = \emptyset.$
\end{itemize} 
We  next proceed to determine the Lee weight of the codeword $\mathtt{c}_{\theta} \in \mathscr{C}_{\mathcal{S}_1}$ for each $\theta = (\theta_1 + u \theta_2, \theta_3) \in \mathtt{R}\setminus \ker(\mu_1) $ satisfying  exactly one of the conditions (I) -- (VIII) above and the  number of choices for $\theta = (\theta_1 + u \theta_2, \theta_3) \in \mathtt{R}$ satisfying each of these eight conditions.
\begin{description}
    \item[(I)]  Let  $\mathcal{S}_{\bm{\beta}}(\theta_1) \cap (\mathpzc{A}\cup \mathpzc{B}) = \emptyset,$  $\mathcal{S}_{\bm{\beta}}(\theta_2) \cap \mathpzc{A} \neq \emptyset,$ and $\theta_3 = 0.$ Here, we have $\mathcal{S}_{\bm{\beta}}(\theta_1+\theta_2) \cap \mathpzc{A} \neq \emptyset.$ In this case, we see, by Lemma \ref{Lem5}(a) and equations \eqref{EQS1} and \eqref{Eq3.1}, that $wt_L(\mathtt{c}_{\theta}) = \frac{2(q-1)}{q}|\mathcal{S}_1|=2(q-1)q^{|\mathpzc{A}|+|\mathpzc{B}|-1}(q^m-q^{|\mathpzc{C}|}).$ Furthermore, by Lemma \ref{Lem3}, the element $\theta = (\theta_1 + u \theta_2, \theta_3) \in \mathtt{R}$ satisfying the condition (I) has precisely $|\mathpzc{X}_{\mathpzc{A} \cup \mathpzc{B}}||\mathpzc{X}_\mathpzc{A}^c|=q^{m-|\mathpzc{A}  \cup  \mathpzc{B}|}(q^m-q^{m-|\mathpzc{A}|})$ distinct choices. 
    \item[(II)] Let  $\mathcal{S}_{\bm{\beta}}(\theta_1) \cap (\mathpzc{A} \cup \mathpzc{B}) \neq \emptyset,$ $\mathcal{S}_{\bm{\beta}}(\theta_2) \cap \mathpzc{A} = \emptyset,$ and $\theta_3=0.$  Here, we will distinguish the following three cases:
(i)  $\mathcal{S}_{\bm{\beta}}(\theta_1) \cap \mathpzc{A} = \emptyset$ and $\mathcal{S}_{\bm{\beta}}(\theta_1) \cap \mathpzc{B} \neq \emptyset,$ (ii)  $\mathcal{S}_{\bm{\beta}}(\theta_1) \cap \mathpzc{A}  \neq \emptyset$ and $\mathcal{S}_{\bm{\beta}}(\theta_1) \cap \mathpzc{B}= \emptyset,$ and 
      (iii) $\mathcal{S}_{\bm{\beta}}(\theta_1) \cap \mathpzc{A}  \neq \emptyset$ and $\mathcal{S}_{\bm{\beta}}(\theta_1) \cap \mathpzc{B} \neq \emptyset.$
    \begin{description}
        \item[(i)] Suppose that  $\mathcal{S}_{\bm{\beta}}(\theta_1) \cap \mathpzc{A} = \emptyset$ and $\mathcal{S}_{\bm{\beta}}(\theta_1) \cap \mathpzc{B} \neq \emptyset.$ Here, we have $\mathcal{S}_{\bm{\beta}}(\theta_1+\theta_2) \cap \mathpzc{A} = \emptyset.$ In this case, we see, by Lemma \ref{Lem5}(a) and equations \eqref{EQS1} and \eqref{Eq3.1}, that $wt(\mathtt{c}_{\theta}) = \frac{2(q-1)}{q}|\mathcal{S}_1|=2(q-1)q^{|\mathpzc{A}|+|\mathpzc{B}|-1}(q^m-q^{|\mathpzc{C}|}).$ We further note, by Lemma \ref{Lem3}, that the element $\theta = (\theta_1 + u \theta_2, \theta_3) \in \mathtt{R}$ satisfying the conditions $\theta_3=0,$ $\mathcal{S}_{\bm{\beta}}(\theta_2) \cap \mathpzc{A} = \emptyset,$ $\mathcal{S}_{\bm{\beta}}(\theta_1) \cap \mathpzc{A} = \emptyset,$ and $\mathcal{S}_{\bm{\beta}}(\theta_1) \cap \mathpzc{B} \neq \emptyset$ has precisely $|\mathpzc{X}_\mathpzc{A}||\mathpzc{Y}_{\mathpzc{A},\mathpzc{B}}|=q^{m-|\mathpzc{A}|}(q^{m-|\mathpzc{A}|}-q^{m-|\mathpzc{A}  \cup  \mathpzc{B}|})$ distinct choices.
        \item[(ii)] Suppose that $\mathcal{S}_{\bm{\beta}}(\theta_1) \cap \mathpzc{A}  \neq \emptyset$ and $\mathcal{S}_{\bm{\beta}}(\theta_1) \cap \mathpzc{B}= \emptyset.$ Here,  we see, by Lemma \ref{Lem5}(a) and equations \eqref{EQS1} and \eqref{Eq3.1}, that $wt(\mathtt{c}_{\theta}) = \frac{(q-1)}{q}|\mathcal{S}_1|=(q-1)q^{|\mathpzc{A}|+|\mathpzc{B}|-1}(q^m-q^{|\mathpzc{C}|}).$ We further note, by Lemma \ref{Lem3}, that the element $\theta = (\theta_1 + u \theta_2, \theta_3) \in \mathtt{R}$ satisfying the conditions $\theta_3=0,$ $\mathcal{S}_{\bm{\beta}}(\theta_2) \cap \mathpzc{A} = \emptyset,$ $\mathcal{S}_{\bm{\beta}}(\theta_1) \cap \mathpzc{A}  \neq \emptyset,$ and $\mathcal{S}_{\bm{\beta}}(\theta_1) \cap \mathpzc{B}= \emptyset$ has precisely $|\mathpzc{X}_\mathpzc{A}||\mathpzc{Y}_{\mathpzc{B},\mathpzc{A}}|=q^{m-|\mathpzc{A}|}(q^{m-|\mathpzc{B}|}-q^{m-|\mathpzc{A}  \cup  \mathpzc{B}|})$ distinct choices.
        \item[(iii)] Suppose that $\mathcal{S}_{\bm{\beta}}(\theta_1) \cap \mathpzc{A}  \neq \emptyset$ and $\mathcal{S}_{\bm{\beta}}(\theta_1) \cap \mathpzc{B} \neq \emptyset.$ Here, we have $\mathcal{S}_{\bm{\beta}}(\theta_1+\theta_2) \cap \mathpzc{A} \neq \emptyset.$ In this case, we see, by Lemma \ref{Lem5}(a) and equations \eqref{EQS1} and \eqref{Eq3.1}, that $wt(\mathtt{c}_{\theta}) = \frac{2(q-1)}{q}|\mathcal{S}_1|=2(q-1)q^{|\mathpzc{A}|+|\mathpzc{B}|-1}(q^m-q^{|\mathpzc{C}|}).$ We further note, by Lemma \ref{Lem3}, that there are precisely $|\mathpzc{X}_\mathpzc{A}||\mathpzc{Z}_{\mathpzc{B},\mathpzc{A}}|=q^{m-|\mathpzc{A}|}(q^m -q^{m-|\mathpzc{A}|}-q^{m-|\mathpzc{B}|}+q^{m-|\mathpzc{A}  \cup  \mathpzc{B}|})$ distinct choices for the element  $\theta = (\theta_1 + u \theta_2, \theta_3) \in \mathtt{R}$ satisfying the conditions $\theta_3=0,$ $\mathcal{S}_{\bm{\beta}}(\theta_2) \cap \mathpzc{A} = \emptyset,$ $\mathcal{S}_{\bm{\beta}}(\theta_1) \cap \mathpzc{A}  \neq \emptyset,$ and $\mathcal{S}_{\bm{\beta}}(\theta_1) \cap \mathpzc{B} \neq \emptyset.$ 
    \end{description}
    \item[(III)] Let  $\mathcal{S}_{\bm{\beta}}(\theta_1) \cap (\mathpzc{A}\cup \mathpzc{B}) \neq \emptyset,$ $\mathcal{S}_{\bm{\beta}}(\theta_2) \cap \mathpzc{A} \neq \emptyset,$ and $\theta_3=0.$ Here, we will distinguish the following three cases:
   (i)  $\mathcal{S}_{\bm{\beta}}(\theta_1) \cap \mathpzc{A} = \emptyset$ and $\mathcal{S}_{\bm{\beta}}(\theta_1) \cap \mathpzc{B} \neq \emptyset,$ (ii)  $\mathcal{S}_{\bm{\beta}}(\theta_1) \cap \mathpzc{A}  \neq \emptyset$ and $\mathcal{S}_{\bm{\beta}}(\theta_1) \cap \mathpzc{B}= \emptyset,$ and 
        (iii)  $\mathcal{S}_{\bm{\beta}}(\theta_1) \cap \mathpzc{A}  \neq \emptyset$ and $\mathcal{S}_{\bm{\beta}}(\theta_1) \cap \mathpzc{B} \neq \emptyset.$
    \begin{description}
    \item[(i)] Suppose that $\mathcal{S}_{\bm{\beta}}(\theta_1) \cap \mathpzc{A} = \emptyset$ and $\mathcal{S}_{\bm{\beta}}(\theta_1) \cap \mathpzc{B} \neq \emptyset.$ Here,  by Lemma \ref{Lem5}(a) and equations \eqref{EQS1} and \eqref{Eq3.1}, we get $wt(\mathtt{c}_{\theta}) = \frac{2(q-1)}{q}|\mathcal{S}_1|=2(q-1)q^{|\mathpzc{A}|+|\mathpzc{B}|-1}(q^m-q^{|\mathpzc{C}|}).$ We further note, by Lemma \ref{Lem3}, that the element  $\theta = (\theta_1 + u \theta_2, \theta_3) \in \mathtt{R}$ satisfying  $\theta_3=0,$ $\mathcal{S}_{\bm{\beta}}(\theta_2) \cap \mathpzc{A} \neq \emptyset,$ $\mathcal{S}_{\bm{\beta}}(\theta_1) \cap \mathpzc{A} = \emptyset,$ and $\mathcal{S}_{\bm{\beta}}(\theta_1) \cap \mathpzc{B} \neq \emptyset$ has precisely $|\mathpzc{X}^c_\mathpzc{A}||\mathpzc{Y}_{\mathpzc{A},\mathpzc{B}}|=(q^m-q^{m-|\mathpzc{A}|})(q^{m-|\mathpzc{A}|}-q^{m-|\mathpzc{A}  \cup  \mathpzc{B}|})$ distinct choices.
        
    \item[(ii)] Suppose that  $\mathcal{S}_{\bm{\beta}}(\theta_1) \cap \mathpzc{A}  \neq \emptyset$ and $\mathcal{S}_{\bm{\beta}}(\theta_1) \cap \mathpzc{B}= \emptyset.$ 
        
        First of all, let us assume that $\mathcal{S}_{\bm{\beta}}(\theta_1+\theta_2) \cap \mathpzc{A} = \emptyset.$ In this case, we see, by Lemma \ref{Lem5}(a) and equations \eqref{EQS1} and  \eqref{Eq3.1}, that $wt(\mathtt{c}_{\theta}) = \frac{(q-1)}{q}|\mathcal{S}_1|=(q-1)q^{|\mathpzc{A}|+|\mathpzc{B}|-1}(q^m-q^{|\mathpzc{C}|}).$ We further observe, by Lemma \ref{Lem4}, that the element $\theta = (\theta_1 + u \theta_2, \theta_3) \in \mathtt{R}$ satisfying  $\theta_3=0,$ $\mathcal{S}_{\bm{\beta}}(\theta_2) \cap \mathpzc{A} \neq \emptyset,$ $\mathcal{S}_{\bm{\beta}}(\theta_1) \cap \mathpzc{A}  \neq \emptyset,$ $\mathcal{S}_{\bm{\beta}}(\theta_1) \cap \mathpzc{B}= \emptyset,$ and $\mathcal{S}_{\bm{\beta}}(\theta_1+\theta_2) \cap \mathpzc{A} = \emptyset$ has precisely $|\widehat{\mathpzc{M}}_{\mathpzc{A},\mathpzc{B}}|= q^{2m-|\mathpzc{A}|-|\mathpzc{A}  \cup  \mathpzc{B}|}(q^{|\mathpzc{A}|-|\mathpzc{A}\cap \mathpzc{B}|}-1)$ distinct choices.

        Next, let us suppose that $\mathcal{S}_{\bm{\beta}}(\theta_1+\theta_2) \cap \mathpzc{A} \neq \emptyset.$ In this case, we see, by Lemma \ref{Lem5}(a) and equations \eqref{EQS1} and  \eqref{Eq3.1}, that $wt(\mathtt{c}_{\theta}) = \frac{2(q-1)}{q}|\mathcal{S}_1|=2(q-1)q^{|\mathpzc{A}|+|\mathpzc{B}|-1}(q^m-q^{|\mathpzc{C}|}).$ We further note, by Lemma \ref{Lem4}, that the element  $\theta = (\theta_1 + u \theta_2, \theta_3) \in \mathtt{R}$ satisfying  $\theta_3=0,$ $\mathcal{S}_{\bm{\beta}}(\theta_2) \cap \mathpzc{A} \neq \emptyset,$ $\mathcal{S}_{\bm{\beta}}(\theta_1) \cap \mathpzc{A}  \neq \emptyset,$ $\mathcal{S}_{\bm{\beta}}(\theta_1) \cap \mathpzc{B}= \emptyset,$ and $\mathcal{S}_{\bm{\beta}}(\theta_1+\theta_2) \cap \mathpzc{A} \neq \emptyset$ has precisely $|\widetilde{\mathpzc{M}}_{\mathpzc{A},\mathpzc{B}}|= q^{2m-|\mathpzc{A}|-|\mathpzc{A}  \cup  \mathpzc{B}|}\big((q^{|\mathpzc{A}|}-1)(q^{|\mathpzc{A}  \cup  \mathpzc{B}|-|\mathpzc{B}|}-1)-(q^{|\mathpzc{A}|-|\mathpzc{A}\cap \mathpzc{B}|}-1)\big)$ distinct choices.
        \item[(iii)] Suppose that  $\mathcal{S}_{\bm{\beta}}(\theta_1) \cap \mathpzc{A}  \neq \emptyset$ and $\mathcal{S}_{\bm{\beta}}(\theta_1) \cap \mathpzc{B} \neq \emptyset.$ In this case, we note, by Lemma \ref{Lem5}(a) and equations \eqref{EQS1} and \eqref{Eq3.1}, that $wt(\mathtt{c}_{\theta}) = \frac{2(q-1)}{q}|\mathcal{S}_1|=2(q-1)q^{|\mathpzc{A}|+|\mathpzc{B}|-1}(q^m-q^{|\mathpzc{C}|}).$ We further note, by Lemma \ref{Lem3}, that the element $\theta = (\theta_1 + u \theta_2, \theta_3) \in \mathtt{R}$ satisfying $\theta_3=0,$ $\mathcal{S}_{\bm{\beta}}(\theta_2) \cap \mathpzc{A} \neq \emptyset,$ $\mathcal{S}_{\bm{\beta}}(\theta_1) \cap \mathpzc{A}  \neq \emptyset,$ and $\mathcal{S}_{\bm{\beta}}(\theta_1) \cap \mathpzc{B} \neq \emptyset$ has precisely $|\mathpzc{X}_\mathpzc{A}^c||\mathpzc{Z}_{\mathpzc{B},\mathpzc{A}}|=(q^m-q^{m-|\mathpzc{A}|})(q^m -q^{m-|\mathpzc{A}|}-q^{m-|\mathpzc{B}|}+q^{m-|\mathpzc{A}  \cup  \mathpzc{B}|})$ distinct choices.
    \end{description}
    \item[(IV)] Let $\theta_3 \neq 0$ and $\mathcal{S}_{\bm{\beta}}(\theta_3) \cap \mathpzc{C} \neq \emptyset.$ In this case, we see, by Lemma \ref{Lem5}(a) and equations \eqref{EQS1} and \eqref{Eq3.1}, that $wt(\mathtt{c}_{\theta}) = \frac{2(q-1)}{q}|\mathcal{S}_1|=2(q-1)q^{|\mathpzc{A}|+|\mathpzc{B}|-1}(q^m-q^{|\mathpzc{C}|}).$
     We also observe, by Lemma \ref{Lem3}, that  there are precisely $|\mathpzc{X}_\mathpzc{C}^c|q^{2m}=(q^m-q^{m-|\mathpzc{C}|})q^{2m}$ distinct choices for the element $\theta = (\theta_1 + u \theta_2, \theta_3) \in \mathtt{R}$ satisfying the condition (IV).
     \item[(V)] Let  $\mathcal{S}_{\bm{\beta}}(\theta_1) \cap (\mathpzc{A}\cup \mathpzc{B}) = \emptyset,$ $\mathcal{S}_{\bm{\beta}}(\theta_2) \cap \mathpzc{A} \neq \emptyset,$ $\theta_3 \neq 0,$ and $\mathcal{S}_{\bm{\beta}}(\theta_3) \cap \mathpzc{C} = \emptyset.$ Here, we have $\mathcal{S}_{\bm{\beta}}(\theta_1+\theta_2) \cap \mathpzc{A} \neq \emptyset.$ In this case, we see, by Lemma \ref{Lem5}(a) and equations \eqref{EQS1} and \eqref{Eq3.1}, that $wt(\mathtt{c}_{\theta}) = \frac{2(q-1)}{q}|\mathcal{S}_1|=2(q-1)q^{|\mathpzc{A}|+|\mathpzc{B}|-1}(q^m-q^{|\mathpzc{C}|}).$ We further note, by Lemma \ref{Lem3}, that the element  $\theta = (\theta_1 + u \theta_2, \theta_3) \in \mathtt{R}$ satisfying the condition (V) has precisely $(|\mathpzc{X}_{\mathpzc{C}}|-1)|\mathpzc{X}_{\mathpzc{A} \cup \mathpzc{B}}||\mathpzc{X}_\mathpzc{A}^c|=(q^{m-|\mathpzc{C}|}-1)q^{m-|\mathpzc{A}  \cup  \mathpzc{B}|}(q^m-q^{m-|\mathpzc{A}|})$ distinct choices. 
     \item[(VI)] Let  $\mathcal{S}_{\bm{\beta}}(\theta_1) \cap (\mathpzc{A} \cup \mathpzc{B}) \neq \emptyset,$ $\mathcal{S}_{\bm{\beta}}(\theta_2) \cap \mathpzc{A} = \emptyset,$ $\theta_3 \neq 0,$ and $\mathcal{S}_{\bm{\beta}}(\theta_3) \cap \mathpzc{C} = \emptyset.$ Here, we will distinguish the following three cases:
   (i)  $\mathcal{S}_{\bm{\beta}}(\theta_1) \cap \mathpzc{A} = \emptyset$ and $\mathcal{S}_{\bm{\beta}}(\theta_1) \cap \mathpzc{B} \neq \emptyset,$
        (ii)  $\mathcal{S}_{\bm{\beta}}(\theta_1) \cap \mathpzc{A}  \neq \emptyset$ and $\mathcal{S}_{\bm{\beta}}(\theta_1) \cap \mathpzc{B}= \emptyset,$ and 
        (iii)  $\mathcal{S}_{\bm{\beta}}(\theta_1) \cap \mathpzc{A}  \neq \emptyset$ and $\mathcal{S}_{\bm{\beta}}(\theta_1) \cap \mathpzc{B} \neq \emptyset.$
    \begin{description}
        \item[(i)] Suppose that  $\mathcal{S}_{\bm{\beta}}(\theta_1) \cap \mathpzc{A} = \emptyset$ and $\mathcal{S}_{\bm{\beta}}(\theta_1) \cap \mathpzc{B} \neq \emptyset.$ Here, by Lemma \ref{Lem5}(a) and equations \eqref{EQS1} and \eqref{Eq3.1}, we obtain $wt(\mathtt{c}_{\theta}) = \frac{2(q-1)}{q}|\mathcal{S}_1|=2(q-1)q^{|\mathpzc{A}|+|\mathpzc{B}|-1}(q^m-q^{|\mathpzc{C}|}).$ Further, we see, by Lemma \ref{Lem3}, that there are precisely $(|\mathpzc{X}_{\mathpzc{C}}|-1)|\mathpzc{X}_\mathpzc{A}||\mathpzc{Y}_{\mathpzc{A},\mathpzc{B}}|=(q^{m-|\mathpzc{C}|}-1)q^{m-|\mathpzc{A}|}(q^{m-|\mathpzc{A}|}-q^{m-|\mathpzc{A}  \cup  \mathpzc{B}|})$  distinct choices for the element $\theta = (\theta_1 + u \theta_2, \theta_3) \in \mathtt{R}$ satisfying  $\theta_3 \neq 0,$ $\mathcal{S}_{\bm{\beta}}(\theta_3) \cap \mathpzc{C} = \emptyset,$ $\mathcal{S}_{\bm{\beta}}(\theta_2) \cap \mathpzc{A} = \emptyset,$ $\mathcal{S}_{\bm{\beta}}(\theta_1) \cap \mathpzc{A} = \emptyset,$ and $\mathcal{S}_{\bm{\beta}}(\theta_1) \cap \mathpzc{B} \neq \emptyset.$
        \item[(ii)]Suppose that  $\mathcal{S}_{\bm{\beta}}(\theta_1) \cap \mathpzc{A}  \neq \emptyset$ and $\mathcal{S}_{\bm{\beta}}(\theta_1) \cap \mathpzc{B}= \emptyset.$ Here, by Lemma \ref{Lem5}(a) and equations \eqref{EQS1} and \eqref{Eq3.1}, we obtain $wt(\mathtt{c}_{\theta}) = \frac{(q-1)}{q}(2|\mathcal{S}_1|+q^{|\mathpzc{A}|+|\mathpzc{B}|+|\mathpzc{C}|})=(q-1)q^{|\mathpzc{A}|+|\mathpzc{B}|-1}(2q^m-q^{|\mathpzc{C}|}).$ We further see, by Lemma \ref{Lem3}, that the element $\theta = (\theta_1 + u \theta_2, \theta_3) \in \mathtt{R}$ satisfying  $\theta_3 \neq 0,$ $\mathcal{S}_{\bm{\beta}}(\theta_3) \cap \mathpzc{C} = \emptyset,$ $\mathcal{S}_{\bm{\beta}}(\theta_2) \cap \mathpzc{A} = \emptyset,$ $\mathcal{S}_{\bm{\beta}}(\theta_1) \cap \mathpzc{A}  \neq \emptyset,$ and $\mathcal{S}_{\bm{\beta}}(\theta_1) \cap \mathpzc{B}= \emptyset$ has precisely $(|\mathpzc{X}_{\mathpzc{C}}|-1)|\mathpzc{X}_\mathpzc{A}||\mathpzc{Y}_{\mathpzc{B},\mathpzc{A}}|=(q^{m-|\mathpzc{C}|}-1)q^{m-|\mathpzc{A}|}(q^{m-|\mathpzc{B}|}-q^{m-|\mathpzc{A}  \cup  \mathpzc{B}|})$ distinct choices.
        \item[(iii)]Suppose that $\mathcal{S}_{\bm{\beta}}(\theta_1) \cap \mathpzc{A}  \neq \emptyset$ and $\mathcal{S}_{\bm{\beta}}(\theta_1) \cap \mathpzc{B} \neq \emptyset.$ Here, by Lemma \ref{Lem5}(a) and equations \eqref{EQS1} and \eqref{Eq3.1}, we obtain $wt(\mathtt{c}_{\theta}) = \frac{2(q-1)}{q}|\mathcal{S}_1|=2(q-1)q^{|\mathpzc{A}|+|\mathpzc{B}|-1}(q^m-q^{|\mathpzc{C}|}).$ We further see, by Lemma \ref{Lem3}, that there are precisely $(|\mathpzc{X}_{\mathpzc{C}}|-1)|\mathpzc{X}_\mathpzc{A}||\mathpzc{Z}_{\mathpzc{B},\mathpzc{A}}|=(q^{m-|\mathpzc{C}|}-1)q^{m-|\mathpzc{A}|}(q^m -q^{m-|\mathpzc{A}|}-q^{m-|\mathpzc{B}|}+q^{m-|\mathpzc{A}  \cup  \mathpzc{B}|})$ distinct choices for the element  $\theta = (\theta_1 + u \theta_2, \theta_3) \in \mathtt{R}$ satisfying  $\theta_3 \neq 0,$ $\mathcal{S}_{\bm{\beta}}(\theta_3) \cap \mathpzc{C} = \emptyset,$ $\mathcal{S}_{\bm{\beta}}(\theta_2) \cap \mathpzc{A} = \emptyset,$ $\mathcal{S}_{\bm{\beta}}(\theta_1) \cap \mathpzc{A}  \neq \emptyset,$ and $\mathcal{S}_{\bm{\beta}}(\theta_1) \cap \mathpzc{B} \neq \emptyset.$ 
    \end{description}
   \item[(VII)] Let  $\mathcal{S}_{\bm{\beta}}(\theta_1) \cap (\mathpzc{A} \cup \mathpzc{B}) \neq \emptyset,$ $\mathcal{S}_{\bm{\beta}}(\theta_2) \cap \mathpzc{A} \neq \emptyset,$ $\theta_3 \neq 0,$ and $\mathcal{S}_{\bm{\beta}}(\theta_3) \cap \mathpzc{C} = \emptyset.$ Here, we will consider the following three cases separately: (i) 
     $\mathcal{S}_{\bm{\beta}}(\theta_1) \cap \mathpzc{A} = \emptyset$ and $\mathcal{S}_{\bm{\beta}}(\theta_1) \cap \mathpzc{B} \neq \emptyset,$
       (ii)  $\mathcal{S}_{\bm{\beta}}(\theta_1) \cap \mathpzc{A}  \neq \emptyset$ and $\mathcal{S}_{\bm{\beta}}(\theta_1) \cap \mathpzc{B}= \emptyset,$ and 
        (iii)  $\mathcal{S}_{\bm{\beta}}(\theta_1) \cap \mathpzc{A}  \neq \emptyset$ and $\mathcal{S}_{\bm{\beta}}(\theta_1) \cap \mathpzc{B} \neq \emptyset.$
 \begin{description}
\item[(i)] Suppose that  $\mathcal{S}_{\bm{\beta}}(\theta_1) \cap \mathpzc{A} = \emptyset$ and $\mathcal{S}_{\bm{\beta}}(\theta_1) \cap \mathpzc{B} \neq \emptyset.$ Here, by  Lemma \ref{Lem5}(a) and equations \eqref{EQS1} and \eqref{Eq3.1}, we obtain $wt(\mathtt{c}_{\theta}) = \frac{2(q-1)}{q}|\mathcal{S}_1|=2(q-1)q^{|\mathpzc{A}|+|\mathpzc{B}|-1}(q^m-q^{|\mathpzc{C}|}).$ We further see, by Lemma \ref{Lem3}, that there are precisely $(|\mathpzc{X}_{\mathpzc{C}}|-1)|\mathpzc{X}^c_\mathpzc{A}||\mathpzc{Y}_{\mathpzc{A},\mathpzc{B}}|=(q^{m-|\mathpzc{C}|}-1)(q^m-q^{m-|\mathpzc{A}|})(q^{m-|\mathpzc{A}|}-q^{m-|\mathpzc{A}  \cup  \mathpzc{B}|})$ distinct choices for the element  $\theta = (\theta_1 + u \theta_2, \theta_3) \in \mathtt{R}$ satisfying  $\theta_3 \neq 0,$ $\mathcal{S}_{\bm{\beta}}(\theta_3) \cap \mathpzc{C} = \emptyset,$ $\mathcal{S}_{\bm{\beta}}(\theta_2) \cap \mathpzc{A} \neq \emptyset,$ $\mathcal{S}_{\bm{\beta}}(\theta_1) \cap \mathpzc{A} = \emptyset,$ and $\mathcal{S}_{\bm{\beta}}(\theta_1) \cap \mathpzc{B} \neq \emptyset.$

        \item[(ii)] Suppose that  $\mathcal{S}_{\bm{\beta}}(\theta_1) \cap \mathpzc{A}  \neq \emptyset$ and $\mathcal{S}_{\bm{\beta}}(\theta_1) \cap \mathpzc{B}= \emptyset.$ 

First of all, let us assume that $\mathcal{S}_{\bm{\beta}}(\theta_1+\theta_2) \cap \mathpzc{A} = \emptyset.$ Here, by  Lemma \ref{Lem5}(a) and equations \eqref{EQS1} and \eqref{Eq3.1}, we obtain $wt(\mathtt{c}_{\theta}) = \frac{(q-1)}{q}(2|\mathcal{S}_1|+q^{|\mathpzc{A}|+|\mathpzc{B}|+|\mathpzc{C}|})=(q-1)q^{|\mathpzc{A}|+|\mathpzc{B}|-1}(2q^m-q^{|\mathpzc{C}|}).$ We further see, by Lemma \ref{Lem3}, that there are precisely $(|\mathpzc{X}_{\mathpzc{C}}|-1)|\widehat{\mathpzc{M}}_{\mathpzc{A},\mathpzc{B}}|= (q^{m-|\mathpzc{C}|}-1)q^{2m-|\mathpzc{A}|-|\mathpzc{A}  \cup  \mathpzc{B}|}(q^{|\mathpzc{A}|-|\mathpzc{A}\cap \mathpzc{B}|}-1)$ distinct choices for the element $\theta = (\theta_1 + u \theta_2, \theta_3) \in \mathtt{R}$ satisfying $\theta_3 \neq 0,$ $\mathcal{S}_{\bm{\beta}}(\theta_3) \cap \mathpzc{C} = \emptyset,$ $\mathcal{S}_{\bm{\beta}}(\theta_2) \cap \mathpzc{A} \neq \emptyset,$ $\mathcal{S}_{\bm{\beta}}(\theta_1) \cap \mathpzc{A}  \neq \emptyset,$ $\mathcal{S}_{\bm{\beta}}(\theta_1) \cap \mathpzc{B}= \emptyset,$ and $\mathcal{S}_{\bm{\beta}}(\theta_1+\theta_2) \cap \mathpzc{A} = \emptyset.$ 

Next, let us suppose that $\mathcal{S}_{\bm{\beta}}(\theta_1+\theta_2) \cap \mathpzc{A} \neq \emptyset.$ Here, by Lemma \ref{Lem5}(a) and equations \eqref{EQS1} and \eqref{Eq3.1}, we obtain $wt(\mathtt{c}_{\theta}) = \frac{2(q-1)}{q}|\mathcal{S}_1|=2(q-1)q^{|\mathpzc{A}|+|\mathpzc{B}|-1}(q^m-q^{|\mathpzc{C}|}).$ We further see, by Lemma \ref{Lem3}, that the element $\theta = (\theta_1 + u \theta_2, \theta_3) \in \mathtt{R}$ satisfying $\theta_3 \neq 0,$ $\mathcal{S}_{\bm{\beta}}(\theta_3) \cap \mathpzc{C} = \emptyset,$ $\mathcal{S}_{\bm{\beta}}(\theta_2) \cap \mathpzc{A} \neq \emptyset,$ $\mathcal{S}_{\bm{\beta}}(\theta_1) \cap \mathpzc{A}  \neq \emptyset,$ $\mathcal{S}_{\bm{\beta}}(\theta_1) \cap \mathpzc{B}= \emptyset,$ and $\mathcal{S}_{\bm{\beta}}(\theta_1+\theta_2) \cap \mathpzc{A} \neq \emptyset$ has precisely $(|\mathpzc{X}_{\mathpzc{C}}|-1)|\widetilde{\mathpzc{M}}_{\mathpzc{A},\mathpzc{B}}|= (q^{m-|\mathpzc{C}|}-1)q^{2m-|\mathpzc{A}|-|\mathpzc{A}  \cup  \mathpzc{B}|}\big((q^{|\mathpzc{A}|}-1)(q^{|\mathpzc{A}  \cup  \mathpzc{B}|-|\mathpzc{B}|}-1)-(q^{|\mathpzc{A}|-|\mathpzc{A}\cap \mathpzc{B}|}-1)\big)$ distinct choices.

\item[(iii)] Suppose that  $\mathcal{S}_{\bm{\beta}}(\theta_1) \cap \mathpzc{A}  \neq \emptyset$ and $\mathcal{S}_{\bm{\beta}}(\theta_1) \cap \mathpzc{B} \neq \emptyset.$ Here, by  Lemma \ref{Lem5}(a) and equations \eqref{EQS1} and \eqref{Eq3.1}, we obtain $wt(\mathtt{c}_{\theta}) = \frac{2(q-1)}{q}|\mathcal{S}_1|=2(q-1)q^{|\mathpzc{A}|+|\mathpzc{B}|-1}(q^m-q^{|\mathpzc{C}|}).$ We further note,  by Lemma \ref{Lem3}, that the element  $\theta = (\theta_1 + u \theta_2, \theta_3) \in \mathtt{R}$ satisfying   $\theta_3 \neq 0,$ $\mathcal{S}_{\bm{\beta}}(\theta_3) \cap \mathpzc{C} = \emptyset,$ $\mathcal{S}_{\bm{\beta}}(\theta_2) \cap \mathpzc{A} \neq \emptyset,$ $\mathcal{S}_{\bm{\beta}}(\theta_1) \cap \mathpzc{A}  \neq \emptyset,$ and $\mathcal{S}_{\bm{\beta}}(\theta_1) \cap \mathpzc{B} \neq \emptyset$ has precisely    $(|\mathpzc{X}_{\mathpzc{C}}|-1)|\mathpzc{X}_\mathpzc{A}^c||\mathpzc{Z}_{\mathpzc{B},\mathpzc{A}}|=(q^{m-|\mathpzc{C}|}-1)(q^m-q^{m-|\mathpzc{A}|})(q^m -q^{m-|\mathpzc{A}|}-q^{m-|\mathpzc{B}|}+q^{m-|\mathpzc{A}  \cup  \mathpzc{B}|})$ distinct choices.
    \end{description}
\item[(VIII)] Finally, let  $\mathcal{S}_{\bm{\beta}}(\theta_1) \cap (\mathpzc{A}\cup \mathpzc{B}) = \emptyset,$ $\mathcal{S}_{\bm{\beta}}(\theta_2) \cap \mathpzc{A} = \emptyset,$ $\theta_3 \neq 0,$ and $\mathcal{S}_{\bm{\beta}}(\theta_3) \cap \mathpzc{C} = \emptyset.$ Here, we have $\mathcal{S}_{\bm{\beta}}(\theta_1+\theta_2) \cap \mathpzc{A} = \emptyset.$ In this case, we see, by Lemma \ref{Lem5}(a) and equations \eqref{EQS1} and \eqref{Eq3.1}, that $wt(\mathtt{c}_{\theta}) = \frac{2(q-1)}{q}(|\mathcal{S}_1|+q^{|\mathpzc{A}|+|\mathpzc{B}|+|\mathpzc{C}|})=2(q-1)q^{m+|\mathpzc{A}|+|\mathpzc{B}|-1}.$ Further, by Lemma \ref{Lem3}, we see that there are precisely $(|\mathpzc{X}_{\mathpzc{C}}|-1)|\mathpzc{X}_{\mathpzc{A} \cup \mathpzc{B}}||\mathpzc{X}_{\mathpzc{A}}|=(q^{m-|\mathpzc{C}|}-1)(q^{m-|\mathpzc{A}  \cup  \mathpzc{B}|})(q^{m-|\mathpzc{A}|})$ distinct choices for the element  $\theta = (\theta_1 + u \theta_2, \theta_3) \in \mathtt{R}$ satisfying the condition (VIII).  
\end{description}

Furthermore, we note, by \eqref{kermu}, that $|\ker(\mu_1)| = q^{2m - |\mathpzc{A}| - |\mathpzc{A} \cup \mathpzc{B}|}.$ Using this fact and by combining the cases (I) -- (VIII) above,  we observe  that Table \ref{T1} provides the Lee weight distribution of the code $\mathscr{C}_{\mathcal{S}_1}$ over $\frac{\mathbb{F}_q[u]}{\langle u^2 \rangle}.$ 
\end{proof}


In the following theorem, we determine the parameters and Lee weight distribution of the code $\mathscr{C}_{\mathcal{S}_2}$ over $\frac{\mathbb{F}_q[u]}{\langle u^2 \rangle}.$ 

\begin{thm}\label{Th2}
The code $\mathscr{C}_{\mathcal{S}_2}$ is a linear code over $\frac{\mathbb{F}_q[u]}{\langle u^2\rangle}$ with parameters $\big( q^{|\mathpzc{B}|+|\mathpzc{C}|}(q^m-q^{|\mathpzc{A}|}),~ q^{2m+|\mathpzc{C}|},~ d_L(\mathscr{C}_{\mathcal{S}_2} )\big )$ and  Lee weight distribution given by Table \ref{T2}, where 
\begin{equation*}
d_L(\mathscr{C}_{\mathcal{S}_2}) = \left\{\begin{array}{ll}  
2(q-1)q^{|\mathpzc{B}|+|\mathpzc{C}|-1}(q^m-q^{|\mathpzc{A}|}) & \text{if }~ \mathpzc{B} = [m];\\
(q-1)q^{m+|\mathpzc{B}|+|\mathpzc{C}|-1} & \text{if }~ \mathpzc{A} \subseteq \mathpzc{B} \subsetneq [m];\\
(q-1)q^{|\mathpzc{B}|+|\mathpzc{C}|-1}(q^m-q^{|\mathpzc{A}|}) & \text{if }~ \mathpzc{A} \not\subseteq \mathpzc{B}.
\end{array}\right.
\end{equation*}
As a consequence, the code $\mathscr{C}_{\mathcal{S}_2}$ is
\begin{itemize}
\item a $2$-weight code if either $\mathpzc{B} = [m],$ or
$\mathpzc{A} \subseteq \mathpzc{B} \subsetneq [m],$ $q = 2,$ and
$|\mathpzc{A}| = m-1.$
\item a $3$-weight code if
$\mathpzc{A} \subseteq \mathpzc{B} \subsetneq [m]$ with either
$q \neq 2,$ or $|\mathpzc{A}| \neq m-1.$
\item a $4$-weight code if
$\mathpzc{A} \not\subseteq \mathpzc{B} \subsetneq [m]$ with either
$\mathpzc{A} \cup \mathpzc{B} = [m],$ or $q = 2$ and
$|\mathpzc{A}| = m-1.$
\item a $5$-weight code otherwise.
\end{itemize}

\begin{table}[H]
\centering
    \begin{tabular}{ |c|c|} 
 \hline Lee (\textit{resp.} Hamming) weight & Frequency \\ \hline
$0$ & $1$ \\
 \hline
$(q-1)(q^m-q^{|\mathpzc{A}|})q^{|\mathpzc{B}|+|\mathpzc{C}|-1}$ & $2(q^{m-|\mathpzc{B}|}-q^{m-|\mathpzc{A} \cup \mathpzc{B}|})$ \\
 \hline
$(q-1)q^{m+|\mathpzc{B}|+|\mathpzc{C}|-1}$ & $2(q^{m-|\mathpzc{A} \cup \mathpzc{B}|}-1)$ \\
\hline
$2(q-1)(q^m-q^{|\mathpzc{A}|})q^{|\mathpzc{B}|+|\mathpzc{C}|-1}$ & $q^{2m+|\mathpzc{C}|} - 2q^{m-|\mathpzc{A}|}(q^{m-|\mathpzc{B}|}-q^{m-|\mathpzc{A}\cup\mathpzc{B}|}) - q^{2m-|\mathpzc{A}\cup\mathpzc{B}|-|\mathpzc{A}|}$ \\
\hline
$(q-1)(2q^m-q^{|\mathpzc{A}|})q^{|\mathpzc{B}|+|\mathpzc{C}|-1}$ & $2(q^{m-|\mathpzc{B}|}-q^{m-|\mathpzc{A} \cup \mathpzc{B}|})(q^{m-|\mathpzc{A}|} - 1)$ \\
\hline
$2(q-1)q^{m+|\mathpzc{B}|+|\mathpzc{C}|-1}$ & $(q^{m-|\mathpzc{A}|}-1)(q^{m-|\mathpzc{A} \cup \mathpzc{B}|}-1) + q^{m-|\mathpzc{A}|} - q^{m-|\mathpzc{A}\cup\mathpzc{B}|}$ \\
 \hline
\end{tabular}
\caption{The Lee (\textit{resp.} Hamming) weight distribution of the code $\mathscr{C}_{\mathcal{S}_2}$ (\textit{resp.} $\phi(\mathscr{C}_{\mathcal{S}_2})$) }\label{T2} 
\end{table}
\end{thm}
\begin{proof} 
To prove the result, we first recall that the code $\mathscr{C}_{\mathcal{S}_2}$  is a linear code of length $|\mathcal{S}_2|$ over $\frac{\mathbb{F}_q[u]}{\langle u^2 \rangle},$ where the set $\mathcal{S}_2$ is as defined by \eqref{S2} and the code $\mathscr{C}_{\mathcal{S}_2}$ is as defined by  \eqref{Code2}. Further, by \eqref{Scard} and \eqref{S2}, we have \begin{equation}\label{EQS2}
|\mathcal{S}_2| =  |\Delta_{\mathpzc{A}}^c| |\Delta_{\mathpzc{B}}||\Delta_{\mathpzc{C}}| =  q^{|\mathpzc{B}| + |\mathpzc{C}|}(q^m - q^{|\mathpzc{A}|}).\end{equation} Thus, $\mathscr{C}_{\mathcal{S}_2}$ is a linear code of length $ q^{|\mathpzc{B}| + |\mathpzc{C}|}(q^m - q^{|\mathpzc{A}|})$ over $\frac{\mathbb{F}_q[u]}{\langle u^2 \rangle}.$

Now, to determine the size of the code $\mathscr{C}_{\mathcal{S}_2},$ we see, by \eqref{Nsize}, that it suffices to determine $|\ker(\mu_2)|,$ where $\mu_2 : \mathtt{R} \rightarrow \mathscr{C}_{\mathcal{S}_2}$ is a surjective $\big(\frac{\mathbb{F}_q[u]}{\langle u^2 \rangle}\big)$-module homomorphism, defined by 
$$\mu_2(\theta) = \mathtt{c}_{\theta}:= \big(Tr( \theta_1 \delta_1) + u \, Tr( \theta_2\delta_1 +  \theta_1 \delta_2 +  \theta_3 \delta_3)\big)_{(\delta_1+u\delta_2, \delta_3) \in \mathcal{S}_2}$$
for all $\theta=(\theta_1+u\theta_2, \theta_3) \in \mathtt{R}.$ From this, following the approach used in Theorem \ref{Th1}, and using equation \eqref{GIe} together with Lemma \ref{Lem5}(b), we obtain
\begin{equation}\label{Eq3.1s2}
wt_L(\mathtt{c}_{\theta}) = \frac{(q-1)}{q} \Big(2|\mathcal{S}_2| - F_{\theta_1}(\Delta_\mathpzc{B})F_{\theta_3}(\Delta_\mathpzc{C}) \big(F_{\theta_2}(\Delta_\mathpzc{A}^c)+F_{\theta_1+\theta_2}(\Delta_\mathpzc{A}^c)\big) \Big).
\end{equation} 
Consequently, we have $wt_L(\mathtt{c}_{\theta})=0$ if and only if
$2|\mathcal{S}_2| =  F_{\theta_1}(\Delta_\mathpzc{B})F_{\theta_3}(\Delta_\mathpzc{C})\big(F_{\theta_2}(\Delta_\mathpzc{A}^c)+F_{\theta_1+\theta_2}(\Delta_\mathpzc{A}^c)\big) ,$ which,   by Lemma \ref{Lem5}(a) and using equation \eqref{EQS2},  holds if and only if   $F_{\theta_1}(\Delta_\mathpzc{B}) = |\Delta_\mathpzc{B}|,$ $F_{\theta_3}(\Delta_\mathpzc{C})=|\Delta_\mathpzc{C}|,$ and $F_{\theta_2}(\Delta_\mathpzc{A}^c)+F_{\theta_1+\theta_2}(\Delta_\mathpzc{A}^c) = 2|\Delta_\mathpzc{A}^c|.$  Now, by Lemma \ref{Lem5}(a), we get
\begin{equation}\label{Ker2}
\ker(\mu_2)=\{ (\theta_1+u\theta_2,\theta_3) \in \mathtt{R} : \theta_1 = \theta_2 = 0 \text{ and } \mathcal{S}_{\bm{\beta}}(\theta_3) \cap \mathpzc{C} = \emptyset \}.
\end{equation}
This, by  Lemma \ref{Lem3}, implies that \begin{equation}\label{kermu2}
|\ker(\mu_2)| = |\mathpzc{X}_{\mathpzc{C}}|=q^{m-|\mathpzc{C}|}.\end{equation} 
From this and using equation \eqref{Nsize}, we get $|\mathscr{C}_{\mathcal{S}_2}|=q^{2m+|\mathpzc{C}|}.$

Next, to determine the Lee weight distribution of the code $\mathscr{C}_{\mathcal{S}_2},$ we assume, throughout  the proof, that $\theta = (\theta_1 + u \theta_2, \theta_3) \in \mathtt{R}$ is such that $\theta \notin \ker(\mu_2).$ Using \eqref{Ker2} and Lemma \ref{Lem5}(a),  we note that the element $\theta = (\theta_1 + u \theta_2, \theta_3) \in \mathtt{R}$ satisfies exactly one of the following five conditions:
\begin{itemize}
\item[(I)]  $\mathcal{S}_{\bm{\beta}}(\theta_1) \cap \mathpzc{B} = \emptyset$ and $\mathcal{S}_{\bm{\beta}}(\theta_3) \cap \mathpzc{C} = \emptyset$ with either
$\theta_2=0,$ $\theta_1\neq0,$ and $\mathcal{S}_{\bm{\beta}}(\theta_1) \cap \mathpzc{A} \neq \emptyset,$ or $\theta_1+\theta_2=0,$ $\theta_2\neq0,$ and $\mathcal{S}_{\bm{\beta}}(\theta_2) \cap \mathpzc{A} \neq \emptyset.$

\item[(II)] $\mathcal{S}_{\bm{\beta}}(\theta_1) \cap \mathpzc{B} = \emptyset$ and $\mathcal{S}_{\bm{\beta}}(\theta_3) \cap \mathpzc{C} = \emptyset$ with either $\theta_2=0,$ $\theta_1\neq0,$ and $\mathcal{S}_{\bm{\beta}}(\theta_1) \cap \mathpzc{A} = \emptyset,$ or $\theta_1+\theta_2=0,$ $\theta_2\neq0,$ and $\mathcal{S}_{\bm{\beta}}(\theta_2) \cap \mathpzc{A} = \emptyset.$

\item[(III)] $\mathcal{S}_{\bm{\beta}}(\theta_1) \cap \mathpzc{B} = \emptyset,$ $\mathcal{S}_{\bm{\beta}}(\theta_3) \cap \mathpzc{C} = \emptyset,$ $\theta_2 \neq 0,$ and $\theta_1+\theta_2 \neq 0$ with either $\mathcal{S}_{\bm{\beta}}(\theta_2) \cap \mathpzc{A} \neq \emptyset$ and $\mathcal{S}_{\bm{\beta}}(\theta_1 + \theta_2) \cap \mathpzc{A} = \emptyset,$ or $\mathcal{S}_{\bm{\beta}}(\theta_2) \cap \mathpzc{A} = \emptyset$ and $\mathcal{S}_{\bm{\beta}}(\theta_1 + \theta_2) \cap \mathpzc{A} \neq \emptyset.$

\item[(IV)] $\mathcal{S}_{\bm{\beta}}(\theta_1) \cap \mathpzc{B} = \emptyset,$ $\mathcal{S}_{\bm{\beta}}(\theta_3) \cap \mathpzc{C} = \emptyset,$ $\theta_2 \neq 0,$ $\theta_1+\theta_2 \neq 0,$ $\mathcal{S}_{\bm{\beta}}(\theta_2) \cap \mathpzc{A} = \emptyset,$ and $\mathcal{S}_{\bm{\beta}}(\theta_1 + \theta_2) \cap \mathpzc{A} = \emptyset.$

\item[(V)] Either $\mathcal{S}_{\bm{\beta}}(\theta_1) \cap \mathpzc{B} \neq \emptyset,$ or $\mathcal{S}_{\bm{\beta}}(\theta_1) \cap \mathpzc{B} = \emptyset$ and $\mathcal{S}_{\bm{\beta}}(\theta_3) \cap \mathpzc{C} \neq \emptyset,$ or $\mathcal{S}_{\bm{\beta}}(\theta_1) \cap \mathpzc{B} = \emptyset,$  $\mathcal{S}_{\bm{\beta}}(\theta_3) \cap \mathpzc{C} = \emptyset,$ $\theta_2 \neq 0,$ $\theta_1+\theta_2 \neq 0,$ $\mathcal{S}_{\bm{\beta}}(\theta_2) \cap \mathpzc{A} \neq \emptyset,$ and $\mathcal{S}_{\bm{\beta}}(\theta_1 + \theta_2) \cap \mathpzc{A} \neq \emptyset.$
\end{itemize}

We  next proceed to determine the Lee weight of the codeword $\mathtt{c}_{\theta} \in \mathscr{C}_{\mathcal{S}_2}$ for each $\theta = (\theta_1 + u \theta_2, \theta_3) \in \mathtt{R}\setminus \ker(\mu_2) $ satisfying  exactly one of the conditions (I) -- (V) above and the  number of choices for $\theta = (\theta_1 + u \theta_2, \theta_3) \in \mathtt{R}$ satisfying each of these five conditions.
\begin{description}
\item[(I)] Let  $\mathcal{S}_{\bm{\beta}}(\theta_1) \cap \mathpzc{B} = \emptyset$ and $\mathcal{S}_{\bm{\beta}}(\theta_3) \cap \mathpzc{C} = \emptyset$ with either $\theta_2=0,$ $\theta_1\neq0,$ and $\mathcal{S}_{\bm{\beta}}(\theta_1) \cap \mathpzc{A} \neq \emptyset,$ or $\theta_1+\theta_2=0,$ $\theta_2\neq0,$ and $\mathcal{S}_{\bm{\beta}}(\theta_2) \cap \mathpzc{A} \neq \emptyset.$ In this case, we see, by Lemma \ref{Lem5}(a) and equations \eqref{EQS2} and \eqref{Eq3.1s2}, that $wt_L(\mathtt{c}_{\theta}) = \frac{(q-1)}{q} \big(2|\mathcal{S}_2| - q^{|\mathpzc{B}|+|\mathpzc{C}|}(q^m - q^{|\mathpzc{A}|}) \big) = (q-1) q^{|\mathpzc{B}|+|\mathpzc{C}|-1}(q^m-q^{|\mathpzc{A}|}).$ Moreover, working as in Theorem \ref{Th1} and by applying Lemma \ref{Lem3}, we see that the element  $\theta = (\theta_1 + u \theta_2, \theta_3) \in \mathtt{R}$ satisfying the condition (I) has precisely $2|\mathpzc{Y}_{\mathpzc{B},\mathpzc{A}}||\mathpzc{X}_{\mathpzc{C}}|=2(q^{m-|\mathpzc{B}|}-q^{m-|\mathpzc{A} \cup \mathpzc{B}|})q^{m-|\mathpzc{C}|}$ distinct choices.

\item[(II)] Let $\mathcal{S}_{\bm{\beta}}(\theta_1) \cap \mathpzc{B} = \emptyset$ and $\mathcal{S}_{\bm{\beta}}(\theta_3) \cap \mathpzc{C} = \emptyset$ with either $\theta_2=0,$ $\theta_1\neq0,$ and $\mathcal{S}_{\bm{\beta}}(\theta_1) \cap \mathpzc{A} = \emptyset,$ or $\theta_1+\theta_2=0,$ $\theta_2\neq0,$ and $\mathcal{S}_{\bm{\beta}}(\theta_2) \cap \mathpzc{A} = \emptyset.$ In this case, we see, by Lemma \ref{Lem5}(a) and equations \eqref{EQS2} and \eqref{Eq3.1s2}, that $wt_L(\mathtt{c}_{\theta}) = \frac{(q-1)}{q} \big(2|\mathcal{S}_2| - q^{|\mathpzc{B}|+|\mathpzc{C}|}(q^m - 2q^{|\mathpzc{A}|}) \big) = (q-1) q^{m+|\mathpzc{B}|+|\mathpzc{C}|-1}.$ Moreover, working as in Theorem \ref{Th1} and by applying Lemma \ref{Lem3}, we see that the element  $\theta = (\theta_1 + u \theta_2, \theta_3) \in \mathtt{R}$ satisfying the condition (II) has precisely $2(|\mathpzc{X}_{\mathpzc{A}\cup\mathpzc{B}}|-1)|\mathpzc{X}_{\mathpzc{C}}| = 2(q^{m-|\mathpzc{A} \cup \mathpzc{B}|}-1)q^{m-|\mathpzc{C}|}$ distinct choices.

\item[(III)] Let $\mathcal{S}_{\bm{\beta}}(\theta_1) \cap \mathpzc{B} = \emptyset,$ $\mathcal{S}_{\bm{\beta}}(\theta_3) \cap \mathpzc{C} = \emptyset,$ $\theta_2 \neq 0,$ and $\theta_1+\theta_2 \neq 0$ with either $\mathcal{S}_{\bm{\beta}}(\theta_2) \cap \mathpzc{A} \neq \emptyset$ and $\mathcal{S}_{\bm{\beta}}(\theta_1 + \theta_2) \cap \mathpzc{A} = \emptyset,$ or $\mathcal{S}_{\bm{\beta}}(\theta_2) \cap \mathpzc{A} = \emptyset$ and $\mathcal{S}_{\bm{\beta}}(\theta_1 + \theta_2) \cap \mathpzc{A} \neq \emptyset.$ In this case, we see, by Lemma \ref{Lem5}(a) and equations \eqref{EQS2} and \eqref{Eq3.1s2}, that $wt_L(\mathtt{c}_{\theta}) = \frac{(q-1)}{q} \big(2|\mathcal{S}_2| + q^{|\mathpzc{A}|+|\mathpzc{B}|+|\mathpzc{C}|} \big) = (q-1)(2q^m-q^{|\mathpzc{A}|})q^{|\mathpzc{B}|+|\mathpzc{C}|-1}.$ Moreover, working as in Theorem \ref{Th1} and by applying Lemma \ref{Lem3}, we see that the element  $\theta = (\theta_1 + u \theta_2, \theta_3) \in \mathtt{R}$ satisfying the condition (III) has precisely $2|\mathpzc{Y}_{\mathpzc{B},\mathpzc{A}}|(|\mathpzc{X}_\mathpzc{A}|-1)|\mathpzc{X}_\mathpzc{C}| = 2(q^{m-|\mathpzc{B}|}-q^{m-|\mathpzc{A} \cup \mathpzc{B}|})(q^{m-|\mathpzc{A}|} - 1)q^{m-|\mathpzc{C}|}$ distinct choices.

\item[(IV)] Let $\mathcal{S}_{\bm{\beta}}(\theta_1) \cap \mathpzc{B} = \emptyset,$ $\mathcal{S}_{\bm{\beta}}(\theta_3) \cap \mathpzc{C} = \emptyset,$ $\theta_2 \neq 0,$ $\theta_1+\theta_2 \neq 0,$ $\mathcal{S}_{\bm{\beta}}(\theta_2) \cap \mathpzc{A} = \emptyset,$ and $\mathcal{S}_{\bm{\beta}}(\theta_1 + \theta_2) \cap \mathpzc{A} = \emptyset.$
In this case, we see, by Lemma \ref{Lem5}(a) and equations \eqref{EQS2} and \eqref{Eq3.1s2}, that $wt_L(\mathtt{c}_{\theta}) = \frac{(q-1)}{q} \big(2|\mathcal{S}_2| + 2q^{|\mathpzc{A}|+|\mathpzc{B}|+|\mathpzc{C}|} \big) = 2(q-1)q^{m+|\mathpzc{B}|+|\mathpzc{C}|-1}.$ Moreover, following the approach in Theorem \ref{Th1} and by applying Lemma \ref{Lem3}, the element  $\theta = (\theta_1 + u \theta_2, \theta_3) \in \mathtt{R}$ satisfying the condition (IV) has precisely $(|\mathpzc{X}_\mathpzc{A}||\mathpzc{X}_{\mathpzc{A}\cup\mathpzc{B}}| - 2|\mathpzc{X}_{\mathpzc{A}\cup\mathpzc{B}}| +1)|\mathpzc{X}_\mathpzc{C}| = \big((q^{m-|\mathpzc{A}|}-1)(q^{m-|\mathpzc{A} \cup \mathpzc{B}|}-1) + q^{m-|\mathpzc{A}|} - q^{m-|\mathpzc{A}\cup\mathpzc{B}|}\big) q^{m-|\mathpzc{C}|}$ distinct choices.

\item[(V)] Suppose that either $\mathcal{S}_{\bm{\beta}}(\theta_1) \cap \mathpzc{B} \neq \emptyset,$ or $\mathcal{S}_{\bm{\beta}}(\theta_1) \cap \mathpzc{B} = \emptyset$ and $\mathcal{S}_{\bm{\beta}}(\theta_3) \cap \mathpzc{C} \neq \emptyset,$ or $\mathcal{S}_{\bm{\beta}}(\theta_1) \cap \mathpzc{B} = \emptyset,$ $\mathcal{S}_{\bm{\beta}}(\theta_3) \cap \mathpzc{C} = \emptyset,$ $\theta_2 \neq 0,$ $\theta_1+\theta_2 \neq 0,$ $\mathcal{S}_{\bm{\beta}}(\theta_2) \cap \mathpzc{A} \neq \emptyset,$ and $\mathcal{S}_{\bm{\beta}}(\theta_1 + \theta_2) \cap \mathpzc{A} \neq \emptyset.$
In this case, we see, by Lemma \ref{Lem5}(a) and equations \eqref{EQS2} and \eqref{Eq3.1s2}, that $wt_L(\mathtt{c}_{\theta}) = \frac{2(q-1)}{q} |\mathcal{S}_2| = 2(q-1)(q^m-q^{|\mathpzc{A}|})q^{|\mathpzc{B}|+|\mathpzc{C}|-1}.$ Moreover, working as in Theorem \ref{Th1} and by applying Lemmas \ref{Lem3} and \ref{Lem4}, we see that the element  $\theta = (\theta_1 + u \theta_2, \theta_3) \in \mathtt{R}$ satisfying the condition (V) has precisely $|\mathpzc{X}_\mathpzc{B}^c|q^{2m} + |\mathpzc{X}_\mathpzc{B}||\mathpzc{X}_\mathpzc{C}^c|q^m + (|\widetilde{\mathpzc{M}}_{\mathpzc{A},\mathpzc{B}}|+|\mathpzc{X}_{\mathpzc{A}\cup\mathpzc{B}}||\mathpzc{X}_\mathpzc{A}^c|)|\mathpzc{X}_\mathpzc{C}| = q^{m-|\mathpzc{C}|}\big( q^{2m+|\mathpzc{C}|} - 2q^{m-|\mathpzc{A}|}(q^{m-|\mathpzc{B}|}-q^{m-|\mathpzc{A}\cup\mathpzc{B}|}) - q^{2m-|\mathpzc{A}\cup\mathpzc{B}|-|\mathpzc{A}|} \big)$ distinct choices.
\end{description}
Furthermore, we note, by \eqref{kermu2}, that $|\ker(\mu_2)| = q^{m - |\mathpzc{C}|}.$ Using this fact and by combining the cases (I) -- (V) above,  one can easily observe that Table \ref{T2} provides the Lee weight distribution of the code $\mathscr{C}_{\mathcal{S}_2}$ over $\frac{\mathbb{F}_q[u]}{\langle u^2 \rangle}.$ 
\end{proof}
\begin{remark}\label{Rk1}
Theorem 3 of Mondal and Lee \cite{Mondal2024} follows, as a special case, from  Theorem \ref{Th2} upon setting $q = 2.$
\end{remark}
In a subsequent work \cite{Josehull}, we show that the Gray image $\Phi(\mathscr{C}_{\mathcal{S}_2})$ is Galois self-orthogonal for every automorphism of $\mathbb{F}_q$ over $\mathbb{F}_p.$ We note that the parameters of $\Phi(\mathscr{C}_{\mathcal{S}_2})$  were not correctly reported in Theorem 2.2 of \cite{Josehull}. Accordingly, in view of Theorem \ref{Th2} and Remark \ref{RRk},  the Galois self-orthogonal code $\Phi(\mathscr{C}_{\mathcal{S}_2})$ considered in \cite[Th. 4.3]{Josehull} has Hamming distance  $2 (q-1) q^{|\mathpzc{B}|+|\mathpzc{C}|-1}(q^m-q^{|\mathpzc{A}|})$ if $\mathpzc{B}=[m];$ $(q-1)q^{m+|\mathpzc{B}|+|\mathpzc{C}|-1}$ if $\mathpzc{A}\subseteq\mathpzc{B}\subsetneq[m];$ and $(q-1)q^{|\mathpzc{B}|+|\mathpzc{C}|-1}(q^m-q^{|\mathpzc{A}|})$ otherwise. Also,  in \cite[Tab. 6]{Josehull}, the code $\Phi(\mathscr{C}_{\mathcal{S}_2})$ of length $2^{7}(2^{6}-2^{2})$ and dimension $7$ has Hamming distance $(2^{2}-1)2^{10}.$ The resulting EAQECCs constructed in \cite[Th. 5.5(a) \&\ 5.6(a)]{Josehull}  have distance $d_i$ satisfying $(q-1)(q^m-\epsilon q^{|\mathpzc{A}|})q^{|\mathpzc{B}|+|\mathpzc{C}|-1} \leq d_i \leq (q-1)(q^m-\epsilon q^{|\mathpzc{A}|})q^{|\mathpzc{B}|+|\mathpzc{C}|-1} + i,$ where $\epsilon=0$ if $\mathpzc{A}\subseteq \mathpzc{B},$ while $\epsilon=1$ otherwise.

In the following theorem, we determine the parameters and  Lee weight distribution of the code $\mathscr{C}_{\mathcal{S}_3}$ over $\frac{\mathbb{F}_q[u]}{\langle u^2 \rangle}.$ 
\begin{thm}\label{Th3}
The code $\mathscr{C}_{\mathcal{S}_3}$ is a linear code over $\frac{\mathbb{F}_q[u]}{\langle u^2\rangle}$ with parameters $\big( q^{|\mathpzc{A}|+|\mathpzc{C}|}(q^m-q^{|\mathpzc{B}|}),~ q^{m+|\mathpzc{A}|+|\mathpzc{C}|},~ 2(q-1)q^{|\mathpzc{A}|+|\mathpzc{C}|-1}(q^m-q^{|\mathpzc{B}|})\big )$ and  Lee weight distribution given by Table \ref{T3}. As a consequence, the code $\mathscr{C}_{\mathcal{S}_3}$ is \begin{itemize}
\item  a $2$-weight code if either $[m]=\mathpzc{A}\cup\mathpzc{B},$ or $\mathpzc{A}\subseteq \mathpzc{B}.$ 
\item  a $3$-weight code otherwise.\end{itemize}
\begin{table}[H]
\centering
    \begin{tabular}{ |c|c|} 
 \hline Lee (\textit{resp.} Hamming) weight & Frequency \\ \hline
$0$ & $1$ \\
\hline
$2(q-1)q^{|\mathpzc{A}|+|\mathpzc{C}|-1}(q^m-q^{|\mathpzc{B}|})$ & $q^{m+|\mathpzc{A}|+|\mathpzc{C}|}-2q^{m-|\mathpzc{B}|}+q^{m-|\mathpzc{A} \cup \mathpzc{B}|}$\\
 \hline
 $(q-1)q^{|\mathpzc{A}|+|\mathpzc{C}|-1}(2q^m-q^{|\mathpzc{B}|})$ & $2(q^{m-|\mathpzc{B}|}-q^{m-|\mathpzc{A} \cup \mathpzc{B}|})$\\
 \hline
 $2(q-1)q^{m+|\mathpzc{A}|+|\mathpzc{C}|-1}$ & $q^{m-|\mathpzc{A}\cup \mathpzc{B}|}-1$\\
 \hline
\end{tabular}
\caption{The Lee (\textit{resp.} Hamming) weight distribution of the code $\mathscr{C}_{\mathcal{S}_3}$ (\textit{resp.} $\Phi(\mathscr{C}_{\mathcal{S}_3})$) }\label{T3} 
\end{table}
\end{thm}
\begin{proof}
To prove the result, we first recall that the code $\mathscr{C}_{\mathcal{S}_3}$  is a linear code of length $|\mathcal{S}_3|$ over $\frac{\mathbb{F}_q[u]}{\langle u^2 \rangle},$ where the set $\mathcal{S}_3$ is as defined by \eqref{S3} and the code $\mathscr{C}_{\mathcal{S}_3}$ is as defined by  \eqref{Code2}. Further, by \eqref{Scard} and \eqref{S3}, we have \begin{equation}\label{EQS3}
|\mathcal{S}_3| =  |\Delta_{\mathpzc{A}}| |\Delta_{\mathpzc{B}}^c||\Delta_{\mathpzc{C}}| =  q^{|\mathpzc{A}| + |\mathpzc{C}|}(q^m - q^{|\mathpzc{B}|}).
\end{equation} Thus, $\mathscr{C}_{\mathcal{S}_3}$ is a linear code of length $ q^{|\mathpzc{A}| + |\mathpzc{C}|}(q^m - q^{|\mathpzc{B}|})$ over $\frac{\mathbb{F}_q[u]}{\langle u^2 \rangle}.$

Now, to determine the size of the code $\mathscr{C}_{\mathcal{S}_3},$ we see, by \eqref{Nsize}, that it suffices to determine $|\ker(\mu_3)|,$ where $\mu_3 : \mathtt{R} \rightarrow \mathscr{C}_{\mathcal{S}_3}$ is a surjective $\big(\frac{\mathbb{F}_q[u]}{\langle u^2 \rangle}\big)$-module homomorphism, defined by 
$$\mu_3(\theta) = \mathtt{c}_{\theta}:= \big(Tr( \theta_1 \delta_1) + u \, Tr( \theta_2\delta_1 +  \theta_1 \delta_2 +  \theta_3 \delta_3)\big)_{(\delta_1+u\delta_2, \delta_3) \in \mathcal{S}_3}$$
for all $\theta=(\theta_1+u\theta_2, \theta_3) \in \mathtt{R}.$ From this, following the approach used in Theorem \ref{Th1}, and using equation \eqref{GIe} together with Lemma \ref{Lem5}(b), we obtain
\begin{equation}\label{Eq3.1s3}
wt_L(\mathtt{c}_{\theta}) = \frac{(q-1)}{q} \Big(2|\mathcal{S}_3| - F_{\theta_1}(\Delta_\mathpzc{B}^c)F_{\theta_3}(\Delta_\mathpzc{C}) \big(F_{\theta_2}(\Delta_\mathpzc{A})+F_{\theta_1+\theta_2}(\Delta_\mathpzc{A})\big) \Big).
\end{equation} 
Consequently, we have $wt_L(\mathtt{c}_{\theta})=0$ if and only if
$2|\mathcal{S}_3| =  F_{\theta_1}(\Delta_\mathpzc{B}^c)F_{\theta_3}(\Delta_\mathpzc{C})\big(F_{\theta_2}(\Delta_\mathpzc{A})+F_{\theta_1+\theta_2}(\Delta_\mathpzc{A})\big) ,$ which,   by Lemma \ref{Lem5}(a) and using equation \eqref{EQS3},  holds if and only if   $F_{\theta_1}(\Delta_\mathpzc{B}^c) = |\Delta_\mathpzc{B}^c|,$ $F_{\theta_3}(\Delta_\mathpzc{C})=|\Delta_\mathpzc{C}|,$ and $F_{\theta_2}(\Delta_\mathpzc{A})+F_{\theta_1+\theta_2}(\Delta_\mathpzc{A}) = 2|\Delta_\mathpzc{A}|.$  Now, by Lemma \ref{Lem5}(a), we get
\begin{equation}\label{Ker3}
\ker(\mu_3)=\{ (\theta_1+u\theta_2,\theta_3) \in \mathtt{R} : \theta_1 = 0,~ \mathcal{S}_{\bm{\beta}}(\theta_2) \cap \mathpzc{A} = \emptyset, \text{ and } \mathcal{S}_{\bm{\beta}}(\theta_3) \cap \mathpzc{C} = \emptyset \}.
\end{equation}
This, by  Lemma \ref{Lem3}, implies that 
\begin{equation}\label{kermu3}
|\ker(\mu_3)| = |\mathpzc{X}_{\mathpzc{A}}||\mathpzc{X}_{\mathpzc{C}}|=q^{2m-|\mathpzc{A}|-|\mathpzc{C}|}.
\end{equation} 
From this and using equation \eqref{Nsize}, we get $|\mathscr{C}_{\mathcal{S}_3}|=q^{m+|\mathpzc{A}|+|\mathpzc{C}|}.$

Next, to determine the Lee weight distribution of the code $\mathscr{C}_{\mathcal{S}_3},$ we assume, throughout  the proof, that $\theta = (\theta_1 + u \theta_2, \theta_3) \in \mathtt{R}$ is such that $\theta \notin \ker(\mu_3).$ By \eqref{Ker3} and Lemma \ref{Lem5}(a),  we note that the element $\theta = (\theta_1 + u \theta_2, \theta_3) \in \mathtt{R}$ satisfies exactly one of the following three conditions:
\begin{itemize}
\item[(I)] $\theta_1 \neq 0,$  $\mathcal{S}_{\bm{\beta}}(\theta_1) \cap \mathpzc{B} = \emptyset,$ and $\mathcal{S}_{\bm{\beta}}(\theta_3) \cap \mathpzc{C} = \emptyset$ with either $\mathcal{S}_{\bm{\beta}}(\theta_2) \cap \mathpzc{A} = \emptyset$ and $\mathcal{S}_{\bm{\beta}}(\theta_1 + \theta_2) \cap \mathpzc{A} \neq \emptyset,$ or $\mathcal{S}_{\bm{\beta}}(\theta_2) \cap \mathpzc{A} \neq \emptyset$ and $\mathcal{S}_{\bm{\beta}}(\theta_1 + \theta_2) \cap \mathpzc{A} = \emptyset.$

\item[(II)] $\theta_1 \neq 0,$  $\mathcal{S}_{\bm{\beta}}(\theta_1) \cap \mathpzc{B} = \emptyset,$ $\mathcal{S}_{\bm{\beta}}(\theta_3) \cap \mathpzc{C} = \emptyset,$ $\mathcal{S}_{\bm{\beta}}(\theta_2) \cap \mathpzc{A} = \emptyset,$ and $\mathcal{S}_{\bm{\beta}}(\theta_1 + \theta_2) \cap \mathpzc{A} = \emptyset.$

\item[(III)] Either $\mathcal{S}_{\bm{\beta}}(\theta_3) \cap \mathpzc{C} \neq \emptyset,$ or $\mathcal{S}_{\bm{\beta}}(\theta_3) \cap \mathpzc{C} = \emptyset,$ $\mathcal{S}_{\bm{\beta}}(\theta_2) \cap \mathpzc{A} \neq \emptyset,$ and $\mathcal{S}_{\bm{\beta}}(\theta_1+\theta_2) \cap \mathpzc{A} \neq \emptyset,$ or $\mathcal{S}_{\bm{\beta}}(\theta_3) \cap \mathpzc{C} = \emptyset,$ $\mathcal{S}_{\bm{\beta}}(\theta_2) \cap \mathpzc{A} = \emptyset,$  $\mathcal{S}_{\bm{\beta}}(\theta_1+\theta_2) \cap \mathpzc{A} = \emptyset,$ $\theta_1 \neq 0,$ and  $\mathcal{S}_{\bm{\beta}}(\theta_1) \cap \mathpzc{B} \neq \emptyset,$ or $\mathcal{S}_{\bm{\beta}}(\theta_3) \cap \mathpzc{C} = \emptyset,$ $\mathcal{S}_{\bm{\beta}}(\theta_2) \cap \mathpzc{A} \neq \emptyset,$  $\mathcal{S}_{\bm{\beta}}(\theta_1+\theta_2) \cap \mathpzc{A} = \emptyset,$ $\theta_1 \neq 0,$ and  $\mathcal{S}_{\bm{\beta}}(\theta_1) \cap \mathpzc{B} \neq \emptyset,$ or $\mathcal{S}_{\bm{\beta}}(\theta_3) \cap \mathpzc{C} = \emptyset,$ $\mathcal{S}_{\bm{\beta}}(\theta_2) \cap \mathpzc{A} = \emptyset,$  $\mathcal{S}_{\bm{\beta}}(\theta_1+\theta_2) \cap \mathpzc{A} \neq \emptyset,$ $\theta_1 \neq 0,$ and  $\mathcal{S}_{\bm{\beta}}(\theta_1) \cap \mathpzc{B} \neq \emptyset.$
\end{itemize}

We  next proceed to determine the Lee weight of the codeword $\mathtt{c}_{\theta} \in \mathscr{C}_{\mathcal{S}_3}$ for each $\theta = (\theta_1 + u \theta_2, \theta_3) \in \mathtt{R}\setminus \ker(\mu_3) $ satisfying  exactly one of the conditions (I) -- (III) above and the  number of choices for $\theta = (\theta_1 + u \theta_2, \theta_3) \in \mathtt{R}$ satisfying each of these three conditions.
\begin{description}
\item[(I)] Let  $\theta_1 \neq 0,$  $\mathcal{S}_{\bm{\beta}}(\theta_1) \cap \mathpzc{B} = \emptyset,$ and $\mathcal{S}_{\bm{\beta}}(\theta_3) \cap \mathpzc{C} = \emptyset$ with either $\mathcal{S}_{\bm{\beta}}(\theta_2) \cap \mathpzc{A} = \emptyset$ and $\mathcal{S}_{\bm{\beta}}(\theta_1 + \theta_2) \cap \mathpzc{A} \neq \emptyset,$ or $\mathcal{S}_{\bm{\beta}}(\theta_2) \cap \mathpzc{A} \neq \emptyset$ and $\mathcal{S}_{\bm{\beta}}(\theta_1 + \theta_2) \cap \mathpzc{A} = \emptyset.$ In this case, we see, by Lemma \ref{Lem5}(a) and equations \eqref{EQS3} and \eqref{Eq3.1s3}, that $wt_L(\mathtt{c}_{\theta}) = \frac{(q-1)}{q} \big(2|\mathcal{S}_3| + q^{|\mathpzc{A}|+|\mathpzc{B}|+|\mathpzc{C}|} \big) = (q-1)q^{|\mathpzc{A}|+|\mathpzc{C}|-1}(2q^m-q^{|\mathpzc{B}|}).$ Moreover, working as in Theorem \ref{Th1} and by applying Lemma \ref{Lem3}, we see that the element  $\theta = (\theta_1 + u \theta_2, \theta_3) \in \mathtt{R}$ satisfying the condition (I) has precisely $2|\mathpzc{Y}_{\mathpzc{B},\mathpzc{A}}||\mathpzc{X}_\mathpzc{A}||\mathpzc{X}_\mathpzc{C}| = 2(q^{m-|\mathpzc{B}|}-q^{m-|\mathpzc{A} \cup \mathpzc{B}|})q^{2m-|\mathpzc{A}|-|\mathpzc{C}|}$ distinct choices.

\item[(II)] Let  $\theta_1 \neq 0,$  $\mathcal{S}_{\bm{\beta}}(\theta_1) \cap \mathpzc{B} = \emptyset,$ $\mathcal{S}_{\bm{\beta}}(\theta_3) \cap \mathpzc{C} = \emptyset,$ $\mathcal{S}_{\bm{\beta}}(\theta_2) \cap \mathpzc{A} = \emptyset,$ and $\mathcal{S}_{\bm{\beta}}(\theta_1 + \theta_2) \cap \mathpzc{A} = \emptyset.$ In this case, we see, by Lemma \ref{Lem5}(a) and equations \eqref{EQS3} and \eqref{Eq3.1s3}, that $wt_L(\mathtt{c}_{\theta}) = \frac{(q-1)}{q} \big(2|\mathcal{S}_3| + 2q^{|\mathpzc{A}|+|\mathpzc{B}|+|\mathpzc{C}|} \big) = 2(q-1)q^{m+|\mathpzc{A}|+|\mathpzc{C}|-1}.$ Moreover, working as in Theorem \ref{Th1} and by applying Lemma \ref{Lem3}, we see that the element  $\theta = (\theta_1 + u \theta_2, \theta_3) \in \mathtt{R}$ satisfying the condition (II) has precisely $(|\mathpzc{X}_{\mathpzc{A}\cup\mathpzc{B}}|-1) |\mathpzc{X}_\mathpzc{A}| |\mathpzc{X}_\mathpzc{C}| = (q^{m-|\mathpzc{A}\cup\mathpzc{B}|}-1)q^{2m-|\mathpzc{A}|-|\mathpzc{C}|}$ distinct choices.

\item[(III)] Suppose that either $\mathcal{S}_{\bm{\beta}}(\theta_3) \cap \mathpzc{C} \neq \emptyset,$ or $\mathcal{S}_{\bm{\beta}}(\theta_3) \cap \mathpzc{C} = \emptyset,$ $\mathcal{S}_{\bm{\beta}}(\theta_2) \cap \mathpzc{A} \neq \emptyset,$ and $\mathcal{S}_{\bm{\beta}}(\theta_1+\theta_2) \cap \mathpzc{A} \neq \emptyset,$ or $\mathcal{S}_{\bm{\beta}}(\theta_3) \cap \mathpzc{C} = \emptyset,$ $\mathcal{S}_{\bm{\beta}}(\theta_2) \cap \mathpzc{A} = \emptyset,$  $\mathcal{S}_{\bm{\beta}}(\theta_1+\theta_2) \cap \mathpzc{A} = \emptyset,$ $\theta_1 \neq 0,$ and  $\mathcal{S}_{\bm{\beta}}(\theta_1) \cap \mathpzc{B} \neq \emptyset,$ or $\mathcal{S}_{\bm{\beta}}(\theta_3) \cap \mathpzc{C} = \emptyset,$ $\mathcal{S}_{\bm{\beta}}(\theta_2) \cap \mathpzc{A} \neq \emptyset,$  $\mathcal{S}_{\bm{\beta}}(\theta_1+\theta_2) \cap \mathpzc{A} = \emptyset,$ $\theta_1 \neq 0,$ and  $\mathcal{S}_{\bm{\beta}}(\theta_1) \cap \mathpzc{B} \neq \emptyset,$ or $\mathcal{S}_{\bm{\beta}}(\theta_3) \cap \mathpzc{C} = \emptyset,$ $\mathcal{S}_{\bm{\beta}}(\theta_2) \cap \mathpzc{A} = \emptyset,$  $\mathcal{S}_{\bm{\beta}}(\theta_1+\theta_2) \cap \mathpzc{A} \neq \emptyset,$ $\theta_1 \neq 0,$ and  $\mathcal{S}_{\bm{\beta}}(\theta_1) \cap \mathpzc{B} \neq \emptyset.$ In this case, we see, by Lemma \ref{Lem5}(a) and equations \eqref{EQS3} and \eqref{Eq3.1s3}, that $wt_L(\mathtt{c}_{\theta}) = \frac{2(q-1)}{q} |\mathcal{S}_3| = 2(q-1)q^{|\mathpzc{A}|+|\mathpzc{C}|-1}(q^m-q^{|\mathpzc{B}|}).$ Moreover, working as in Theorem \ref{Th1} and by applying Lemmas \ref{Lem3} and \ref{Lem4}, we see that the element  $\theta = (\theta_1 + u \theta_2, \theta_3) \in \mathtt{R}$ satisfying the condition (III) has precisely $|\mathpzc{X}_\mathpzc{C}^c|q^{2m} + |\mathpzc{X}_\mathpzc{C}||\mathpzc{X}_\mathpzc{A}||\mathpzc{X}_\mathpzc{A}^c| + (|\widetilde{\mathpzc{M}}_{\mathpzc{A},\mathpzc{B}}| + |\widetilde{\mathpzc{N}}_{\mathpzc{A},\mathpzc{B}}|) |\mathpzc{X}_\mathpzc{C}| + |\mathpzc{X}_\mathpzc{C}| (q^m|\mathpzc{X}_\mathpzc{B}^c| - |\mathpzc{Y}_{\mathpzc{A},\mathpzc{B}}||\mathpzc{X}_\mathpzc{A}^c| - |\widetilde{\mathpzc{N}}_{\mathpzc{A},\mathpzc{B}}|) = q^{2m-|\mathpzc{A}|-|\mathpzc{C}|} \big(q^{m+|\mathpzc{A}|+|\mathpzc{C}|}-2q^{m-|\mathpzc{B}|}+q^{m-|\mathpzc{A} \cup \mathpzc{B}|} \big)$ distinct choices.
\end{description}
Furthermore, we note, by \eqref{kermu3}, that $|\ker(\mu_3)| = q^{2m - |\mathpzc{A}| - |\mathpzc{C}|}.$ Using this fact and by combining the cases (I) -- (III) above,  we observe  that Table \ref{T3} provides the Lee weight distribution of the code $\mathscr{C}_{\mathcal{S}_3}$ over $\frac{\mathbb{F}_q[u]}{\langle u^2 \rangle}.$
\end{proof}
\begin{remark}\label{Rk2}Theorem 2 of Mondal and Lee \cite{Mondal2024} follows, as a special case, from  Theorem \ref{Th3} by setting $q = 2.$
\end{remark}
In the following theorem, we determine the parameters and  Lee weight distribution of the code $\mathscr{C}_{\mathcal{S}_4}.$ 
\begin{thm}\label{Th4} 
The code $\mathscr{C}_{\mathcal{S}_4}$ is a linear  code over $\frac{\mathbb{F}_q[u]}{\langle u^2\rangle}$ with Lee weight distribution given by Table \ref{T4}.  Furthermore, when $\mathpzc{A}\subseteq\mathpzc{B},$ the code $\mathscr{C}_{\mathcal{S}_4}$ has parameters
$\big( (q^m-q^{|\mathpzc{C}|})(q^{|\mathpzc{A}|}-1)q^{|\mathpzc{B}|}, ~q^{m+|\mathpzc{A}|+|\mathpzc{B}|},~2(q-1)q^{|\mathpzc{B}|-1}\big((q^m-q^{|\mathpzc{C}|})(q^{|\mathpzc{A}|}-1)-q^{|\mathpzc{C}|}\big),$
and is 
\begin{itemize}
\item a $3$-weight code if $|\mathpzc{A}| + |\mathpzc{C}| = m.$
\item a $4$-weight code if $|\mathpzc{A}| + |\mathpzc{C}| \neq m.$
\end{itemize}
When $\mathpzc{A} \not\subseteq \mathpzc{B},$ the code $\mathscr{C}_{\mathcal{S}_4}$ has parameters
$\big( (q^m-q^{|\mathpzc{C}|})(q^{|\mathpzc{A}|}-1)q^{|\mathpzc{B}|},~ q^{m+|\mathpzc{A}|+|\mathpzc{A}\cup \mathpzc{B}|},~(q-1)(q^{m}-q^{|\mathpzc{C}|})q^{|\mathpzc{A}|+|\mathpzc{B}|-1}\big ),$
and  is
\begin{itemize}
\item a $4$-weight code if $|\mathpzc{A}| + |\mathpzc{C}| = m + 1,$ $q=2,$ and $|\mathpzc{A}| = 2.$
\item a $5$-weight code if either $|\mathpzc{A}| + |\mathpzc{C}| = m,$ or
$|\mathpzc{A}| + |\mathpzc{C}| = m + 1,$ $q=2,$ and $|\mathpzc{A}| \geq 3.$
\item a $6$-weight code otherwise.
\end{itemize} 
\begin{table}[H]
\centering
    \begin{tabular}{ |c|c|} 
 \hline Lee (\textit{resp.} Hamming) weight & Frequency  \\ \hline
$0$ & $1$ \\
\hline
$(q-1)(q^{m}-q^{|\mathpzc{C}|})q^{|\mathpzc{A}|+|\mathpzc{B}|-1}$ & $2(q^{|\mathpzc{A}\cup \mathpzc{B}|-|\mathpzc{B}|}-1)$\\
\hline
$2(q-1)q^{|\mathpzc{B}|-1}\big((q^m-q^{|\mathpzc{C}|})(q^{|\mathpzc{A}|}-1)-q^{|\mathpzc{C}|}\big)$ & $\big((q^{|\mathpzc{A}|}-1)+(q^{|\mathpzc{A}|}-2)(q^{|\mathpzc{A}\cup \mathpzc{B}|-|\mathpzc{B}|}-1)\big)(q^{m-|\mathpzc{C}|}-1)$ \\
\hline
$2(q-1)(q^m-q^{|\mathpzc{C}|})(q^{|\mathpzc{A}|}-1)q^{|\mathpzc{B}|-1}$ & $q^{m+|\mathpzc{A}|+|\mathpzc{A} \cup \mathpzc{B}|}-q^{m+|\mathpzc{A}|+|\mathpzc{A} \cup \mathpzc{B}|-|\mathpzc{C}|-|\mathpzc{B}|}$ \\
\hline
$2(q-1)(q^m-q^{|\mathpzc{C}|})q^{|\mathpzc{A}|+|\mathpzc{B}|-1}$ & $(q^{|\mathpzc{A}|}-1)+(q^{|\mathpzc{A}|}-2)(q^{|\mathpzc{A} \cup \mathpzc{B}|-|\mathpzc{B}|}-1)$\\
\hline
$(q-1)q^{|\mathpzc{B}|-1}\big((2q^m-q^{|\mathpzc{C}|})(q^{|\mathpzc{A}|}-1)-q^{|\mathpzc{C}|}\big)$ & $2(q^{|\mathpzc{A} \cup \mathpzc{B}|-|\mathpzc{B}|}-1)(q^{m-|\mathpzc{C}|}-1)$ \\
\hline
$2(q-1)q^{m+|\mathpzc{B}|-1}(q^{|\mathpzc{A}|}-1)$ & $q^{m-|\mathpzc{C}|}-1$ \\
\hline
\end{tabular}
\caption{The Lee (\textit{resp.} Hamming) weight distribution of the code $\mathscr{C}_{\mathcal{S}_4}$ (\textit{resp.} $\Phi(\mathscr{C}_{\mathcal{S}_4})$)}\label{T4} 
\end{table}
\end{thm}
\begin{proof}
To prove the result, we first recall that the code $\mathscr{C}_{\mathcal{S}_4}$  is a linear code of length $|\mathcal{S}_4|$ over $\frac{\mathbb{F}_q[u]}{\langle u^2 \rangle},$ where the set $\mathcal{S}_4$ is as defined by \eqref{S4} and the code $\mathscr{C}_{\mathcal{S}_4}$ is as defined by  \eqref{Code2}. Further, by \eqref{Scard} and \eqref{S4}, we have \begin{equation}\label{EQS4}
|\mathcal{S}_4| =  |\Delta_{\mathpzc{A}}^\ast| |\Delta_{\mathpzc{B}}||\Delta_{\mathpzc{C}}^c| =  (q^{|\mathpzc{A}|}-1)q^{|\mathpzc{B}|}(q^m - q^{|\mathpzc{C}|}).
\end{equation} Thus, $\mathscr{C}_{\mathcal{S}_4}$ is a linear code of length $(q^{|\mathpzc{A}|}-1)q^{|\mathpzc{B}|}(q^m - q^{|\mathpzc{C}|})$ over $\frac{\mathbb{F}_q[u]}{\langle u^2 \rangle}.$

Now, to determine the size of the code $\mathscr{C}_{\mathcal{S}_4},$ we see, by \eqref{Nsize}, that it suffices to determine $|\ker(\mu_4)|,$ where $\mu_4 : \mathtt{R} \rightarrow \mathscr{C}_{\mathcal{S}_4}$ is a surjective $\big(\frac{\mathbb{F}_q[u]}{\langle u^2 \rangle}\big)$-module homomorphism, defined by 
$$\mu_4(\theta) = \mathtt{c}_{\theta}:= \big(Tr( \theta_1 \delta_1) + u \, Tr( \theta_2\delta_1 +  \theta_1 \delta_2 +  \theta_3 \delta_3)\big)_{(\delta_1+u\delta_2, \delta_3) \in \mathcal{S}_4}$$
for all $\theta=(\theta_1+u\theta_2, \theta_3) \in \mathtt{R}.$ From this, following the approach used in Theorem \ref{Th1}, and using equation \eqref{GIe} together with Lemma \ref{Lem5}(b), we obtain
\begin{equation}\label{Eq3.1s4}
wt_L(\mathtt{c}_{\theta}) = \frac{(q-1)}{q} \Big(2|\mathcal{S}_4| - F_{\theta_1}(\Delta_\mathpzc{B})F_{\theta_3}(\Delta_\mathpzc{C}^c) \big(F_{\theta_2}(\Delta_\mathpzc{A}^\ast)+F_{\theta_1+\theta_2}(\Delta_\mathpzc{A}^\ast)\big) \Big),
\end{equation}
where $F_{\theta_2}(\Delta_\mathpzc{A}^\ast) = F_{\theta_2}(\Delta_\mathpzc{A}) - 1$ and $F_{\theta_1+\theta_2}(\Delta_\mathpzc{A}^\ast) = F_{\theta_1+\theta_2}(\Delta_\mathpzc{A}) - 1.$ Consequently, we have $wt_L(\mathtt{c}_{\theta})=0$ if and only if
$2|\mathcal{S}_4| =  F_{\theta_1}(\Delta_\mathpzc{B})F_{\theta_3}(\Delta_\mathpzc{C}^c)\big(F_{\theta_2}(\Delta_\mathpzc{A}^\ast)+F_{\theta_1+\theta_2}(\Delta_\mathpzc{A}^\ast)\big) ,$ which,   by Lemma \ref{Lem5}(a) and using equation \eqref{EQS4},  holds if and only if   $F_{\theta_1}(\Delta_\mathpzc{B}) = |\Delta_\mathpzc{B}|,$ $F_{\theta_3}(\Delta_\mathpzc{C}^c)=|\Delta_\mathpzc{C}^c|,$ and $F_{\theta_2}(\Delta_\mathpzc{A}^\ast)+F_{\theta_1+\theta_2}(\Delta_\mathpzc{A}^\ast) = 2|\Delta_\mathpzc{A}^\ast|.$  Now, by Lemma \ref{Lem5}(a), we get
\begin{equation}\label{Ker4}
\ker(\mu_4)=\{ (\theta_1+u\theta_2,\theta_3) \in \mathtt{R} : \mathcal{S}_{\bm{\beta}}(\theta_1) \cap (\mathpzc{A} \cup \mathpzc{B}) = \emptyset,~ \mathcal{S}_{\bm{\beta}}(\theta_2) \cap \mathpzc{A} = \emptyset, ~\text{ and }~ \theta_3 = 0 \}.
\end{equation}
This, by  Lemma \ref{Lem3},   implies that 
\begin{equation}\label{kermu4}
|\ker(\mu_4)| = |\mathpzc{X}_{\mathpzc{A}\cup\mathpzc{B}}||\mathpzc{X}_{\mathpzc{A}}|=q^{2m-|\mathpzc{A}|-|\mathpzc{A}\cup\mathpzc{B}|}.
\end{equation} 
From this and using equation \eqref{Nsize}, we get $|\mathscr{C}_{\mathcal{S}_4}|=q^{m+|\mathpzc{A}|+|\mathpzc{A}\cup\mathpzc{B}|}.$

Next, to determine the Lee weight distribution of the code $\mathscr{C}_{\mathcal{S}_4},$ we assume, throughout  the proof, that $\theta = (\theta_1 + u \theta_2, \theta_3) \in \mathtt{R}$ is such that $\theta \notin \ker(\mu_4).$ By \eqref{Ker2},  we note that the element $\theta = (\theta_1 + u \theta_2, \theta_3) \in \mathtt{R}$ satisfies exactly one of the following six conditions:
\begin{itemize}
\item[(I)]  $\mathcal{S}_{\bm{\beta}}(\theta_1) \cap \mathpzc{B} = \emptyset$ and $\theta_3 = 0$ with either $\mathcal{S}_{\bm{\beta}}(\theta_2) \cap \mathpzc{A} = \emptyset$ and $\mathcal{S}_{\bm{\beta}}(\theta_1+\theta_2) \cap \mathpzc{A} \neq \emptyset,$ or $\mathcal{S}_{\bm{\beta}}(\theta_2) \cap \mathpzc{A} \neq \emptyset$ and $\mathcal{S}_{\bm{\beta}}(\theta_1+\theta_2) \cap \mathpzc{A} = \emptyset.$

\item[(II)] $\mathcal{S}_{\bm{\beta}}(\theta_1) \cap \mathpzc{B} = \emptyset,$ $\theta_3 \neq 0,$ $\mathcal{S}_{\bm{\beta}}(\theta_3) \cap \mathpzc{C} = \emptyset,$ $\mathcal{S}_{\bm{\beta}}(\theta_2) \cap \mathpzc{A} \neq \emptyset,$ and $\mathcal{S}_{\bm{\beta}}(\theta_1+\theta_2) \cap \mathpzc{A} \neq \emptyset.$

\item[(III)] $\mathcal{S}_{\bm{\beta}}(\theta_1) \cap \mathpzc{B} = \emptyset,$ $\theta_3 = 0,$ $\mathcal{S}_{\bm{\beta}}(\theta_2) \cap \mathpzc{A} \neq \emptyset,$ and $\mathcal{S}_{\bm{\beta}}(\theta_1+\theta_2) \cap \mathpzc{A} \neq \emptyset.$

\item[(IV)] $\mathcal{S}_{\bm{\beta}}(\theta_1) \cap \mathpzc{B} = \emptyset,$ $\theta_3 \neq 0,$ and $\mathcal{S}_{\bm{\beta}}(\theta_3) \cap \mathpzc{C} = \emptyset$ with either $\mathcal{S}_{\bm{\beta}}(\theta_2) \cap \mathpzc{A} = \emptyset$ and $\mathcal{S}_{\bm{\beta}}(\theta_1+\theta_2) \cap \mathpzc{A} \neq \emptyset,$ or $\mathcal{S}_{\bm{\beta}}(\theta_2) \cap \mathpzc{A} \neq \emptyset$ and $\mathcal{S}_{\bm{\beta}}(\theta_1+\theta_2) \cap \mathpzc{A} = \emptyset.$

\item[(V)] $\mathcal{S}_{\bm{\beta}}(\theta_1) \cap \mathpzc{B} = \emptyset,$ $\theta_3 \neq 0,$ $\mathcal{S}_{\bm{\beta}}(\theta_3) \cap \mathpzc{C} = \emptyset,$  $\mathcal{S}_{\bm{\beta}}(\theta_2) \cap \mathpzc{A} = \emptyset,$ and $\mathcal{S}_{\bm{\beta}}(\theta_1+\theta_2) \cap \mathpzc{A} = \emptyset.$

\item[(VI)] Either $\mathcal{S}_{\bm{\beta}}(\theta_1) \cap \mathpzc{B} \neq \emptyset,$ or $\mathcal{S}_{\bm{\beta}}(\theta_1) \cap \mathpzc{B} = \emptyset,$ $\theta_3 \neq 0,$ and $\mathcal{S}_{\bm{\beta}}(\theta_3) \cap \mathpzc{C} \neq \emptyset.$
\end{itemize}

We  next proceed to determine the Lee weight of the codeword $\mathtt{c}_{\theta} \in \mathscr{C}_{\mathcal{S}_3}$ for each $\theta = (\theta_1 + u \theta_2, \theta_3) \in \mathtt{R}\setminus \ker(\mu_4) $ satisfying  exactly one of the conditions (I) -- (VI) above and the  number of choices for $\theta = (\theta_1 + u \theta_2, \theta_3) \in \mathtt{R}$ satisfying each of these six conditions.
\begin{description}
\item[(I)] Let  $\mathcal{S}_{\bm{\beta}}(\theta_1) \cap \mathpzc{B} = \emptyset$ and $\theta_3 = 0$ with either $\mathcal{S}_{\bm{\beta}}(\theta_2) \cap \mathpzc{A} = \emptyset$ and $\mathcal{S}_{\bm{\beta}}(\theta_1+\theta_2) \cap \mathpzc{A} \neq \emptyset,$ or $\mathcal{S}_{\bm{\beta}}(\theta_2) \cap \mathpzc{A} \neq \emptyset$ and $\mathcal{S}_{\bm{\beta}}(\theta_1+\theta_2) \cap \mathpzc{A} = \emptyset.$ In this case, we see, by Lemma \ref{Lem5}(a) and equations \eqref{EQS4} and \eqref{Eq3.1s4}, that $wt_L(\mathtt{c}_{\theta}) = \frac{(q-1)}{q} \big(2|\mathcal{S}_4| - q^{|\mathpzc{B}|}(q^m - q^{|\mathpzc{C}|})(q^{|\mathpzc{A}|}-2) \big) = (q-1)(q^m-q^{|\mathpzc{C}|})q^{|\mathpzc{A}|+|\mathpzc{B}|-1}.$ Moreover, following the approach in Theorem \ref{Th1} and by applying Lemma \ref{Lem3}, we see that the element  $\theta = (\theta_1 + u \theta_2, \theta_3) \in \mathtt{R}$ satisfying the condition (I) has precisely $2|\mathpzc{Y}_{\mathpzc{B},\mathpzc{A}}||\mathpzc{X}_\mathpzc{A}| = 2(q^{|\mathpzc{A}\cup \mathpzc{B}|-|\mathpzc{B}|}-1) q^{2m-|\mathpzc{A}|-|\mathpzc{A}\cup\mathpzc{B}|}$ distinct choices.

\item[(II)] Let  $\mathcal{S}_{\bm{\beta}}(\theta_1) \cap \mathpzc{B} = \emptyset,$ $\theta_3 \neq 0,$ $\mathcal{S}_{\bm{\beta}}(\theta_3) \cap \mathpzc{C} = \emptyset,$ $\mathcal{S}_{\bm{\beta}}(\theta_2) \cap \mathpzc{A} \neq \emptyset,$ and $\mathcal{S}_{\bm{\beta}}(\theta_1+\theta_2) \cap \mathpzc{A} \neq \emptyset.$ In this case, we see, by Lemma \ref{Lem5}(a) and equations \eqref{EQS4} and \eqref{Eq3.1s4}, that $wt_L(\mathtt{c}_{\theta}) = \frac{(q-1)}{q} \big(2|\mathcal{S}_4| - 2q^{|\mathpzc{B}|+|\mathpzc{C}|} \big) = 2(q-1)q^{|\mathpzc{B}|-1}\big((q^m-q^{|\mathpzc{C}|})(q^{|\mathpzc{A}|}-1)-q^{|\mathpzc{C}|}\big).$ Moreover, following the approach in Theorem \ref{Th1} and by applying Lemmas \ref{Lem3} and \ref{Lem4}, we see that the element  $\theta = (\theta_1 + u \theta_2, \theta_3) \in \mathtt{R}$ satisfying the condition (II) has precisely $(|\mathpzc{X}_\mathpzc{C}|-1) (|\mathpzc{X}_{\mathpzc{A}\cup\mathpzc{B}}||\mathpzc{X}_\mathpzc{A}^c| + |\widetilde{\mathpzc{M}}_{\mathpzc{A},\mathpzc{B}}|) = \big((q^{|\mathpzc{A}|}-1)+(q^{|\mathpzc{A}|}-2)(q^{|\mathpzc{A}\cup \mathpzc{B}|-|\mathpzc{B}|}-1)\big)(q^{m-|\mathpzc{C}|}-1) q^{2m-|\mathpzc{A}|-|\mathpzc{A}\cup\mathpzc{B}|}$ distinct choices.

\item[(III)] Let  $\mathcal{S}_{\bm{\beta}}(\theta_1) \cap \mathpzc{B} = \emptyset,$ $\theta_3 = 0,$ $\mathcal{S}_{\bm{\beta}}(\theta_2) \cap \mathpzc{A} \neq \emptyset,$ and $\mathcal{S}_{\bm{\beta}}(\theta_1+\theta_2) \cap \mathpzc{A} \neq \emptyset.$ In this case, we see, by Lemma \ref{Lem5}(a) and equations \eqref{EQS4} and \eqref{Eq3.1s4}, that $wt_L(\mathtt{c}_{\theta}) = \frac{(q-1)}{q} \big(2|\mathcal{S}_4| + 2q^{|\mathpzc{B}|}(q^m-q^{|\mathpzc{C}|}) \big) = 2(q-1)(q^m-q^{|\mathpzc{C}|})q^{|\mathpzc{A}|+|\mathpzc{B}|-1}.$ Moreover, following the approach in Theorem \ref{Th1} and by applying Lemmas \ref{Lem3} and \ref{Lem4}, we see that the element  $\theta = (\theta_1 + u \theta_2, \theta_3) \in \mathtt{R}$ satisfying the condition (III) has precisely $|\mathpzc{X}_{\mathpzc{A}\cup\mathpzc{B}}||\mathpzc{X}_\mathpzc{A}^c| + |\widetilde{\mathpzc{M}}_{\mathpzc{A},\mathpzc{B}}| = \big((q^{|\mathpzc{A}|}-1)+(q^{|\mathpzc{A}|}-2)(q^{|\mathpzc{A} \cup \mathpzc{B}|-|\mathpzc{B}|}-1)\big) q^{2m-|\mathpzc{A}|-|\mathpzc{A}\cup\mathpzc{B}|}$ distinct choices.

\item[(IV)] Let  $\mathcal{S}_{\bm{\beta}}(\theta_1) \cap \mathpzc{B} = \emptyset,$ $\theta_3 \neq 0,$ and $\mathcal{S}_{\bm{\beta}}(\theta_3) \cap \mathpzc{C} = \emptyset$ with either $\mathcal{S}_{\bm{\beta}}(\theta_2) \cap \mathpzc{A} = \emptyset$ and $\mathcal{S}_{\bm{\beta}}(\theta_1+\theta_2) \cap \mathpzc{A} \neq \emptyset,$ or $\mathcal{S}_{\bm{\beta}}(\theta_2) \cap \mathpzc{A} \neq \emptyset$ and $\mathcal{S}_{\bm{\beta}}(\theta_1+\theta_2) \cap \mathpzc{A} = \emptyset.$ In this case, we see, by Lemma \ref{Lem5}(a) and equations \eqref{EQS4} and \eqref{Eq3.1s4}, that $wt_L(\mathtt{c}_{\theta}) = \frac{(q-1)}{q} \big(2|\mathcal{S}_4| + q^{|\mathpzc{B}|+|\mathpzc{C}|}(q^{|\mathpzc{A}|}-2) \big) = (q-1)q^{|\mathpzc{B}|-1}\big((2q^m-q^{|\mathpzc{C}|})(q^{|\mathpzc{A}|}-1)-q^{|\mathpzc{C}|}\big).$ Moreover, following the approach in Theorem \ref{Th1} and by applying Lemma \ref{Lem3}, we see that the element  $\theta = (\theta_1 + u \theta_2, \theta_3) \in \mathtt{R}$ satisfying the condition (IV) has precisely $2|\mathpzc{Y}_{\mathpzc{B},\mathpzc{A}}| |\mathpzc{X}_\mathpzc{A}| (|\mathpzc{X}_\mathpzc{C}|-1) = 2(q^{|\mathpzc{A} \cup \mathpzc{B}|-|\mathpzc{B}|}-1)(q^{m-|\mathpzc{C}|}-1) q^{2m-|\mathpzc{A}|-|\mathpzc{A}\cup\mathpzc{B}|}$ distinct choices.

\item[(V)] Let  $\mathcal{S}_{\bm{\beta}}(\theta_1) \cap \mathpzc{B} = \emptyset,$ $\theta_3 \neq 0,$ $\mathcal{S}_{\bm{\beta}}(\theta_3) \cap \mathpzc{C} = \emptyset,$  $\mathcal{S}_{\bm{\beta}}(\theta_2) \cap \mathpzc{A} = \emptyset,$ and $\mathcal{S}_{\bm{\beta}}(\theta_1+\theta_2) \cap \mathpzc{A} = \emptyset.$ In this case, we see, by Lemma \ref{Lem5}(a) and equations \eqref{EQS4} and \eqref{Eq3.1s4}, that $wt_L(\mathtt{c}_{\theta}) = \frac{(q-1)}{q} \big(2|\mathcal{S}_4| + 2q^{|\mathpzc{B}|+|\mathpzc{C}|}(q^{|\mathpzc{A}|}-1) \big) = 2(q-1)q^{m+|\mathpzc{B}|-1}(q^{|\mathpzc{A}|}-1).$ Moreover, following the approach in Theorem \ref{Th1} and by applying Lemma \ref{Lem3}, we see that the element  $\theta = (\theta_1 + u \theta_2, \theta_3) \in \mathtt{R}$ satisfying the condition (V) has precisely $(|\mathpzc{X}_\mathpzc{C}|-1)|\mathpzc{X}_\mathpzc{A}| |\mathpzc{X}_{\mathpzc{A}\cup\mathpzc{B}}| = (q^{m-|\mathpzc{C}|}-1) q^{2m-|\mathpzc{A}|-|\mathpzc{A}\cup\mathpzc{B}|}$ distinct choices.

\item[(VI)] Suppose that either  $\mathcal{S}_{\bm{\beta}}(\theta_1) \cap \mathpzc{B} \neq \emptyset,$ or $\mathcal{S}_{\bm{\beta}}(\theta_1) \cap \mathpzc{B} = \emptyset,$ $\theta_3 \neq 0,$ and $\mathcal{S}_{\bm{\beta}}(\theta_3) \cap \mathpzc{C} \neq \emptyset.$ In this case, we see, by Lemma \ref{Lem5}(a) and equations \eqref{EQS4} and \eqref{Eq3.1s4}, that $wt_L(\mathtt{c}_{\theta}) = \frac{2(q-1)}{q} |\mathcal{S}_4| = 2(q-1)(q^m-q^{|\mathpzc{C}|})(q^{|\mathpzc{A}|}-1)q^{|\mathpzc{B}|-1}.$ Moreover, following the approach in Theorem \ref{Th1} and by applying Lemma \ref{Lem3}, we see that the element  $\theta = (\theta_1 + u \theta_2, \theta_3) \in \mathtt{R}$ satisfying the condition (VI) has precisely $|\mathpzc{X}_\mathpzc{B}^c|q^{2m} + |\mathpzc{X}_\mathpzc{B}||\mathpzc{X}_\mathpzc{C}^c|q^m = (q^{m+|\mathpzc{A}|+|\mathpzc{A} \cup \mathpzc{B}|}-q^{m+|\mathpzc{A}|+|\mathpzc{A} \cup \mathpzc{B}|-|\mathpzc{C}|-|\mathpzc{B}|}) q^{2m-|\mathpzc{A}|-|\mathpzc{A}\cup\mathpzc{B}|}$ distinct choices.
\end{description}
Furthermore, we note, by \eqref{kermu4}, that $|\ker(\mu_4)| = q^{2m-|\mathpzc{A}|-|\mathpzc{A}\cup\mathpzc{B}|}.$ Using this fact and by combining the cases (I) -- (VI) above,  we observe  that Table \ref{T4} provides the Lee weight distribution of the code $\mathscr{C}_{\mathcal{S}_4}$ over $\frac{\mathbb{F}_q[u]}{\langle u^2 \rangle}.$
\end{proof}

\section{Parameters of the Gray images of  $\mathscr{C}_{\mathcal{S}_1},$ $\mathscr{C}_{\mathcal{S}_2},$ $\mathscr{C}_{\mathcal{S}_3}$ and $\mathscr{C}_{\mathcal{S}_4}$ over $\frac{\mathbb{F}_q[u]}{\langle u^2 \rangle}$}\label{GI}
In this section, we will investigate the Gray images of the codes $\mathscr{C}_{\mathcal{S}_1},$ $\mathscr{C}_{\mathcal{S}_2},$ $\mathscr{C}_{\mathcal{S}_3},$ and $\mathscr{C}_{\mathcal{S}_4}$ over $\frac{\mathbb{F}_q[u]}{\langle u^2 \rangle}$ under the Gray map $\Phi$ and explicitly determine their parameters and Hamming weight distributions.  For this, we shall henceforth represent the vectors of $\mathbb{F}_q^{3m}$ in the form $(w_1,w_2,w_3),$ where $w_1, w_2,w_3 \in \mathbb{F}_q^m.$   We further define the following four multisets consisting of the vectors of $\mathbb{F}_q^{3m}$:
\begin{align}\label{D1}
    \mathcal{N}_1  &=  \big\{\big\{
(w_2+\omega,w_3,w_1)\in \mathbb{F}_q^{3m} : w_1 \in \Delta_{\mathpzc{A}}, ~w_2 \in \Delta_{\mathpzc{B}},~w_3 \in \Delta_{\mathpzc{C}}^c,~ \omega\in \{\textbf{0}, w_1\} \text{ with }|\mathpzc{C}|<m\big\}\big\},\\ 
\label{D2}
\mathcal{N}_2  &=  \big\{\big\{
(w_2+\omega,w_3,w_1)\in \mathbb{F}_q^{3m} : w_1 \in \Delta_{\mathpzc{A}}^c, ~w_2 \in \Delta_{\mathpzc{B}},~w_3 \in \Delta_{\mathpzc{C}},~ \omega\in \{\textbf{0}, w_1\} \text{ with }|\mathpzc{A}|<m\big\}\big\},\\
\label{D3}
\mathcal{N}_3  &=  \big\{\big\{
(w_2+\omega,w_3,w_1)\in \mathbb{F}_q^{3m} : w_1 \in \Delta_{\mathpzc{A}}, ~w_2 \in \Delta_{\mathpzc{B}}^c,~w_3 \in \Delta_{\mathpzc{C}},~ \omega\in \{\textbf{0}, w_1\} \text{ with }|\mathpzc{B}|<m\big\}\big\}, \text{ and }\\
\label{D4}
\mathcal{N}_4  &=  \big\{\big\{
 (w_2+\omega, w_3,w_1)\in \mathbb{F}_q^{3m} : w_1 \in \Delta_{\mathpzc{A}}^\ast, ~w_2 \in \Delta_{\mathpzc{B}},~w_3 \in \Delta_{\mathpzc{C}}^c,~ \omega\in \{\textbf{0}, w_1\} \text{ with }|\mathpzc{C}|<m \text{ and } |\mathpzc{A}| \geq 2\big\}\big\},
\end{align}
where $\{\!\{\,\cdot\,\}\!\}$ denotes a multiset, allowing elements to appear with multiplicities, from this point on.  We next make a key observation in the following lemma.

\begin{lemma}\label{Matrix}
For $1\leq i \leq 4,$ let $\mathscr{G}_{\mathcal{N}_i}\in M_{3m \times 2|\mathcal{S}_{i}|}(\mathbb{F}_q)$ be a matrix whose transposed columns form the multiset $\mathcal{N}_i,$ considered up to permutation of columns. Then the Gray image $\Phi(\mathscr{C}_{\mathcal{S}_i})$ of $\mathscr{C}_{\mathcal{S}_i}$ is a linear code of length $2|\mathcal{S}_{i}|$ over $\mathbb{F}_q$ with  $\mathscr{G}_{\mathcal{N}_i}$ as its spanning matrix.
\end{lemma}
\begin{proof}
It follows immediately from the definition of  $\mathscr{C}_{\mathcal{S}_i}$ and  Theorem \ref{Th6}.
\end{proof}

To study the Gray image $\Phi(\mathscr{C}_{\mathcal{S}_1})$,  we begin with the following observation.

\begin{remark}\label{comp1}
By Lemma \ref{Matrix}, we see that the code $\Phi(\mathscr{C}_{\mathcal{S}_1})$ is precisely the row span of the matrix  $\mathscr{G}_{\mathcal{N}_1}$ over $\mathbb{F}_q.$ Moreover, when $\mathpzc{A} \subseteq \mathpzc{B},$ we have $\Delta_\mathpzc{A}\subseteq\Delta_\mathpzc{B}.$ In this case, we further observe  that each column of the matrix $\mathscr{G}_{\mathcal{N}_1}$ appears exactly twice. More precisely, up to permutation equivalence, the matrix $\mathscr{G}_{\mathcal{N}_1}$ is of the form 
$\mathscr{G}_{\mathcal{N}_1}=\big[\widehat{\mathscr{G}}_{\mathcal{N}_1}\mid \widehat{\mathscr{G}}_{\mathcal{N}_1}\big],$ where $\widehat{\mathscr{G}}_{\mathcal{N}_1}\in M_{3m \times |\mathcal{S}_{1}|}(\mathbb{F}_q)$ is a matrix whose transposed columns form the set $\widehat{\mathcal{N}}_1 = \{(w_2,w_3,w_1)\in \mathbb{F}_q^{3m} : w_1 \in \Delta_{\mathpzc{A}}, ~w_2 \in \Delta_{\mathpzc{B}},~w_3 \in \Delta_{\mathpzc{C}}^c \text{ with }|\mathpzc{C}|<m \}.$   As a consequence,  the code $\Phi(\mathscr{C}_{\mathcal{S}_1})$ is a double repetition of the code $\mathcal{D}_{1}$ spanned by the rows of  $\widehat{\mathscr{G}}_{\mathcal{N}_1}$ over $\mathbb{F}_q.$ 

 Additionally, we observe that all columns of $\widehat{\mathscr{G}}_{\mathcal{N}_1}\in M_{3m \times |\mathcal{S}_{1}|}(\mathbb{F}_q)$ are distinct and that there are exactly $2m-|\mathpzc{A}| -|\mathpzc{B}|$  zero rows in $\widehat{\mathscr{G}}_{\mathcal{N}_1}.$ By applying a suitable row permutation to $\widehat{\mathscr{G}}_{\mathcal{N}_1}$ (if necessary) so that its last $2m-|\mathpzc{A}| -|\mathpzc{B}|$ rows are zero, the transposed columns of $\widehat{\mathscr{G}}_{\mathcal{N}_1}$ can be regarded as elements of $\mathbb{F}_q^{m+|\mathpzc{A}|+|\mathpzc{B}|} \setminus \Delta_{\mathpzc{L}_1},$ where $\Delta_{{\mathpzc{L}}_1} \subseteq \mathbb{F}_q^{m+|\mathpzc{A}|+|\mathpzc{B}|}$ is a simplicial complex with support 
\begin{equation*}
    \mathpzc{L}_1 = \{1,2,\ldots,|\mathpzc{B}|\} \cup \{i+|\mathpzc{B}| : i \in \mathpzc{C}\} \cup \{m+|\mathpzc{B}|+1, m+|\mathpzc{B}|+2, \ldots,m+|\mathpzc{B}|+|\mathpzc{A}|\}. 
\end{equation*}
 Thus, the code $\mathcal{D}_1$ is a linear code over $\mathbb{F}_q$ with defining set $\mathbb{F}_q^{m+|\mathpzc{A}|+|\mathpzc{B}|} \setminus \Delta_{\mathpzc{L}_1},$ and hence it belongs to the well-known  family of Solomon–Stiffler codes \cite{Solomon1965}.  Further, the  parameters and Hamming weight distribution of $\mathcal{D}_1,$ and hence of the code $\Phi(\mathscr{C}_{\mathcal{S}_1})$ in the case $\mathpzc{A} \subseteq \mathpzc{B},$  can be obtained by applying Theorem 2 of Hu \etal \cite{Hu2024}.   \end{remark} In view of Remark \ref{comp1},   we consider the case  $\mathpzc{A} \not\subseteq \mathpzc{B}$ and explicitly determine the parameters and Hamming weight distribution of the Gray image $\Phi(\mathscr{C}_{\mathcal{S}_1})$  over $\mathbb{F}_q$ in the following theorem.  
\begin{thm}\label{ThN1}
When $\mathpzc{A}\not\subseteq \mathpzc{B},$ the Gray image $\Phi(\mathscr{C}_{\mathcal{S}_1})$ of the code $\mathscr{C}_{\mathcal{S}_1}$ is a $4$-weight linear code over $\mathbb{F}_q$ with parameters $$\big[2q^{|\mathpzc{A}|+|\mathpzc{B}|}(q^m-q^{|\mathpzc{C}|}), ~{m+|\mathpzc{A}|+|\mathpzc{A}\cup\mathpzc{B}|}, ~(q-1)q^{|\mathpzc{A}|+|\mathpzc{B}|-1}(q^m-q^{|\mathpzc{C}|})\big ]$$ and Hamming weight distribution given by  Table \ref{T1}. 
Furthermore, the code $\Phi(\mathscr{C}_{\mathcal{S}_1})$ is self-orthogonal  if either $q=2,$ or $q=3.$
\end{thm}
\begin{proof}
The result follows immediately by applying Theorem \ref{Th1}, Lemma \ref{Matrix}, Theorems 1.4.8(ii) and 1.4.10(i) of \cite{HuffF} and Remark \ref{RRk}.
\end{proof}

\begin{remark}\label{comp1.1} 
By Theorem 2 of Hu \etal \cite{Hu2024}, the code $\mathcal{D}_1$ (as defined in Remark \ref{comp1})  is a $2$-weight code over $\mathbb{F}_q$ with parameters $\big[q^{|\mathpzc{A}|+|\mathpzc{B}|}(q^m-q^{|\mathpzc{C}|}),~ m+|\mathpzc{A}|+|\mathpzc{B}|,~ (q-1)q^{|\mathpzc{A}|+|\mathpzc{B}|-1}(q^m-q^{|\mathpzc{C}|})\big].$ When $\mathpzc{A}\not\subseteq \mathpzc{B},$ we observe that the code $\Phi(\mathscr{C}_{\mathcal{S}_1})$ considered in Theorem \ref{ThN1} achieves the same Hamming distance as the code $\mathcal{D}_1,$ while having twice the length and a larger dimension.  More precisely, the dimension of $\Phi(\mathscr{C}_{\mathcal{S}_1})$ exceeds that of $\mathcal{D}_1$ by  $|\mathpzc{A}\cup \mathpzc{B}|-|\mathpzc{B}| =|\mathpzc{A}|-|\mathpzc{A}\cap \mathpzc{B}|> 0.$
\end{remark}

The following example illustrates Theorem \ref{ThN1}.
\begin{example}\label{Ex4.1}
Let $q=4$ and $m=3,$ and let $\mathpzc{A}=\{1\},$ $\mathpzc{B}=\{2,3\},$ and $\mathpzc{C}=\{3\}$ be subsets of $[m]=[3]=\{1,2,3\}.$ By carrying out computations in Magma \cite{Magma}, we see that the code $\Phi(\mathscr{C}_{\mathcal{S}_1})$ is a $4$-weight linear $[7680,7,2880]$-code over $\mathbb{F}_4$ with Hamming weight enumerator $W_{\Phi(\mathscr{C}_{\mathcal{S}_1})}(Z)=1+6Z^{2880}+16272Z^{5760}+90Z^{5952}+15Z^{6144}.$ This  is consistent with Theorem \ref{ThN1}. \end{example}

To study the Gray image $\Phi(\mathscr{C}_{\mathcal{S}_2}),$ we first make the following observation.

\begin{remark}\label{comp2}
When $\mathpzc{B}=[m],$ we note,  arguing as  in Remark \ref{comp1} and using Lemma \ref{Matrix}, that the spanning matrix $\mathscr{G}_{\mathcal{N}_2}$ of the code $\Phi(\mathscr{C}_{\mathcal{S}_2}),$ up to permutation equivalence, is of the form $\mathscr{G}_{\mathcal{N}_2}=\big[\widehat{\mathscr{G}}_{\mathcal{N}_2}\mid \widehat{\mathscr{G}}_{\mathcal{N}_2}\big],$  where $\widehat{\mathscr{G}}_{\mathcal{N}_2}\in M_{3m \times |\mathcal{S}_{2}|}(\mathbb{F}_q)$ is a matrix whose transposed columns form the set $\widehat{\mathcal{N}}_2 = \{(w_2,w_3,w_1)\in \mathbb{F}_q^{3m} : w_1 \in \Delta_{\mathpzc{A}}^c, ~w_2 \in \Delta_{\mathpzc{B}},~w_3 \in \Delta_{\mathpzc{C}} \text{ with }|\mathpzc{A}|<m \}.$ As a consequence, the code $\Phi(\mathscr{C}_{\mathcal{S}_2})$ is a double repetition of the code $\mathcal{D}_2$ spanned by the rows of the matrix $\widehat{\mathscr{G}}_{\mathcal{N}_2}.$
   
   Additionally, we observe that all columns of $\widehat{\mathscr{G}}_{\mathcal{N}_2}\in M_{3m \times |\mathcal{S}_{2}|}(\mathbb{F}_q)$ are distinct and that there are exactly $2m-|\mathpzc{B}| -|\mathpzc{C}|$  zero rows in $\widehat{\mathscr{G}}_{\mathcal{N}_2}.$ By applying a suitable row permutation to the matrix $\widehat{\mathscr{G}}_{\mathcal{N}_2}$ (if necessary) so that its last $2m-|\mathpzc{B}| -|\mathpzc{C}|$  rows are zero, the transposed columns of $\widehat{\mathscr{G}}_{\mathcal{N}_2}$ can be regarded as elements of $\mathbb{F}_q^{m+|\mathpzc{B}|+|\mathpzc{C}|} \setminus \Delta_{\mathpzc{L}_2},$ where $\Delta_{{\mathpzc{L}}_2} \subseteq \mathbb{F}_q^{m+|\mathpzc{B}|+|\mathpzc{C}|}$ is a simplicial complex with support 
\begin{equation*}
\mathpzc{L}_2 = \{1,2,\ldots,|\mathpzc{B}|+|\mathpzc{C}|\} \cup \{i+|\mathpzc{B}|+|\mathpzc{C}| : i \in \mathpzc{A}\}.
\end{equation*}
Thus, the code $\mathcal{D}_2$ is a linear code over $\mathbb{F}_q$ with defining set $\mathbb{F}_q^{m+|\mathpzc{B}|+|\mathpzc{C}|} \setminus \Delta_{\mathpzc{L}_2},$ and hence it belongs to the  well-known family of Solomon–Stiffler codes \cite{Solomon1965}.  Further, the  parameters and Hamming weight distribution of $\mathcal{D}_2,$ and hence of the code $\Phi(\mathscr{C}_{\mathcal{S}_2})$ in the case $\mathpzc{B}=[m],$ can be obtained by applying Theorem 2 of Hu \etal \cite{Hu2024}. 
\end{remark}

In the following theorem,  we determine the parameters and Hamming weight distribution  of the Gray image $\Phi(\mathscr{C}_{\mathcal{S}_2})$ under the assumption that $\mathpzc{B} \neq [m]$ (see Remark \ref{comp2}).
\begin{thm}\label{ThN2}
When $\mathpzc{B} \neq [m],$ the Gray image $\Phi(\mathscr{C}_{\mathcal{S}_2})$ of the code $\mathscr{C}_{\mathcal{S}_2}$ is a linear code over $\mathbb{F}_q$ with parameters
$$\big[2q^{|\mathpzc{B}|+|\mathpzc{C}|}(q^m-q^{|\mathpzc{A}|}), ~{2m+|\mathpzc{C}|},~(q-1)(q^m-\epsilon q^{|\mathpzc{A}|})q^{|\mathpzc{B}|+|\mathpzc{C}|-1}\big ]$$ and Hamming weight distribution given by  Table \ref{T2}, where $\epsilon=0$ if $\mathpzc{A}\subseteq \mathpzc{B},$ while $\epsilon=1$ if $\mathpzc{A}\not\subseteq \mathpzc{B}.$   Furthermore, the code $\Phi(\mathscr{C}_{\mathcal{S}_2})$ is
 \begin{itemize}
 \item  self-orthogonal if either $q=2,$ or $q=3.$ 
\item a $2$-weight code if $\mathpzc{A}\subseteq \mathpzc{B},$ $q=2,$ and $|\mathpzc{A}| = m-1.$ 
\item a $3$-weight code if $\mathpzc{A}\subseteq \mathpzc{B}$ with either
$q \neq 2,$ or $|\mathpzc{A}| \neq m-1.$
\item a $4$-weight code if
$\mathpzc{A} \not\subseteq \mathpzc{B} $ with either
$\mathpzc{A} \cup \mathpzc{B} = [m],$ or $q = 2$ and
$|\mathpzc{A}| = m-1.$
\item a $5$-weight code if
$\mathpzc{A} \not\subseteq \mathpzc{B}$ and
$\mathpzc{A} \cup \mathpzc{B} \neq [m]$ with either
$q \neq 2,$ or $|\mathpzc{A}| \neq m-1.$
\end{itemize}

\end{thm}
\begin{proof}
The result follows immediately by applying Theorem \ref{Th2}, Lemma \ref{Matrix}, and Theorems 1.4.8(ii) and 1.4.10(i) of \cite{HuffF}  and Remark \ref{RRk}.
\end{proof}

\begin{remark}\label{comp2.2}
By Theorem 2 of Hu \etal \cite{Hu2024}, the code $\mathcal{D}_2$ (as defined in Remark \ref{comp2}) is  a $2$-weight  code  over $\mathbb{F}_q$ with parameters $\big[q^{|\mathpzc{B}|+|\mathpzc{C}|}(q^m-q^{|\mathpzc{A}|}),~ m+|\mathpzc{B}|+|\mathpzc{C}|,~ (q-1)q^{|\mathpzc{B}|+|\mathpzc{C}|-1}(q^m-q^{|\mathpzc{A}|})\big].$ 

If $\mathpzc{A}\subseteq \mathpzc{B}\subsetneq [m]$ with either $q \neq 2,$ or $|\mathpzc{A}| \neq m-1,$ then the Gray image $\Phi(\mathscr{C}_{\mathcal{S}_2})$ studied in Theorem \ref{ThN2} is a $3$-weight code with parameters $\big[2q^{|\mathpzc{B}|+|\mathpzc{C}|}(q^m-q^{|\mathpzc{A}|}), ~{2m+|\mathpzc{C}|},~(q-1)q^{m+|\mathpzc{B}|+|\mathpzc{C}|-1}\big ],$ which are distinct from the parameters of the code $\mathcal{D}_2$ and its double repetition.

When $\mathpzc{A}\not\subseteq \mathpzc{B}\subsetneq [m],$ the code $\Phi(\mathscr{C}_{\mathcal{S}_2})$ considered in Theorem \ref{ThN2} achieves the same Hamming distance as the code $\mathcal{D}_2,$ while its length is doubled and its dimension is increased by $m-|\mathpzc{B}|>0.$ Moreover, in this case, the code $\Phi(\mathscr{C}_{\mathcal{S}_2})$ is a $4$-weight code if either
$\mathpzc{A} \cup \mathpzc{B} = [m],$ or $q = 2$ and
$|\mathpzc{A}| = m-1,$ while it is a $5$-weight code otherwise.

When $q=2,$ $\mathpzc{A}=\mathpzc{B},$ and $|\mathpzc{A}|=m-1,$ the code $\Phi(\mathscr{C}_{\mathcal{S}_2})$ is a binary linear $\left[2^{2m+|\mathpzc{C}|-1}, 2m+|\mathpzc{C}|,2^{2m+|\mathpzc{C}|-2}\right]$-code and has non-zero Hamming weights $2^{2m+|\mathpzc{C}|-2}$ and $2^{2m+|\mathpzc{C}|-1}$ with frequencies  $2^{2m+|\mathpzc{C}|}-2$ and $1,$ respectively. This,  by Proposition 3(4) of Mondal \etal \cite{Mondal2024}, implies that the binary code $\Phi(\mathscr{C}_{\mathcal{S}_2})$ is also a first order Reed-Muller code.
\end{remark}

The following example illustrates Theorem \ref{ThN2}.
\begin{example}\label{Ex4.2}
Let $q=4$ and $m=4,$ and let $\mathpzc{A}=\{1\},$ $\mathpzc{B}=\{2,3\},$ and $\mathpzc{C}=\{2\}$ be subsets of $[m]=[4]=\{1,2,3,4\}.$ By carrying out computations in Magma \cite{Magma}, we see that the code $\Phi(\mathscr{C}_{\mathcal{S}_2})$ is a $5$-weight linear $[32256, 9, 12096]$-code over $\mathbb{F}_4$ with Hamming weight enumerator $W_{\Phi(\mathscr{C}_{\mathcal{S}_2})}(Z)=1+24Z^{12096}+6Z^{12288}+260352Z^{24192}+1512Z^{24384}+249Z^{24576}.$ This agrees  with Theorem \ref{ThN2}. \end{example}

Next, to study the Gray image $\Phi(\mathscr{C}_{\mathcal{S}_3}),$ we first make the following observation.

\begin{remark}\label{comp3}
When $\mathpzc{A}\subseteq \mathpzc{B},$ we note, arguing as in Remark \ref{comp1} and using Lemma \ref{Matrix}, that the spanning matrix $\mathscr{G}_{\mathcal{N}_3}$ of the code $\Phi(\mathscr{C}_{\mathcal{S}_3}),$ up to permutation equivalence, is of the form $\mathscr{G}_{\mathcal{N}_3}=\big[\widehat{\mathscr{G}}_{\mathcal{N}_3}\mid \widehat{\mathscr{G}}_{\mathcal{N}_3}\big],$ where  $\widehat{\mathscr{G}}_{\mathcal{N}_3}\in M_{3m \times |\mathcal{S}_{3}|}(\mathbb{F}_q)$ is a matrix whose transposed columns form the set $\widehat{\mathcal{N}}_3 = \{(w_2,w_3,w_1)\in \mathbb{F}_q^{3m} : w_1 \in \Delta_{\mathpzc{A}}, ~w_2 \in \Delta_{\mathpzc{B}}^c,~w_3 \in \Delta_{\mathpzc{C}} \text{ with }|\mathpzc{B}|<m \}.$  As a consequence,  the code $\Phi(\mathscr{C}_{\mathcal{S}_3})$ is a double repetition of the code $\mathcal{D}_3$ spanned by the rows of the matrix  $\widehat{\mathscr{G}}_{\mathcal{N}_3}$ over $\mathbb{F}_q.$
   
Additionally, we observe that all columns of $\widehat{\mathscr{G}}_{\mathcal{N}_3}\in M_{3m \times |\mathcal{S}_{3}|}(\mathbb{F}_q)$ are distinct and that there are exactly $2m-|\mathpzc{A}| -|\mathpzc{C}|$  zero rows in $\widehat{\mathscr{G}}_{\mathcal{N}_3}.$ By applying a suitable row permutation to the matrix  $\widehat{\mathscr{G}}_{\mathcal{N}_3}$ (if necessary) so that its last $2m-|\mathpzc{A}| -|\mathpzc{C}|$ rows are zero, the transposed columns of $\widehat{\mathscr{G}}_{\mathcal{N}_3}$ can be viewed as elements of $\mathbb{F}_q^{m+|\mathpzc{A}|+|\mathpzc{C}|} \setminus \Delta_{\mathpzc{L}_3},$ where $\Delta_{{\mathpzc{L}}_3} \subseteq \mathbb{F}_q^{m+|\mathpzc{A}|+|\mathpzc{C}|}$ is a simplicial complex with support 
\begin{equation*}
    \mathpzc{L}_3 = \mathpzc{B} \cup \{m+1, m+2, \ldots,m+|\mathpzc{A}|+|\mathpzc{C}|\}.
\end{equation*}
Thus, the code $\mathcal{D}_3$ is a linear code over $\mathbb{F}_q$ with defining set $\mathbb{F}_q^{m+|\mathpzc{A}|+|\mathpzc{C}|} \setminus \Delta_{\mathpzc{L}_3},$ and hence it belongs to the well-known family of  Solomon–Stiffler codes \cite{Solomon1965}.  Further, the  parameters and Hamming weight distribution of $\mathcal{D}_3,$ and hence of the code  $\Phi(\mathscr{C}_{\mathcal{S}_3})$ in the case $\mathpzc{A}\subseteq \mathpzc{B},$ can be obtained by applying Theorem 2 of Hu \etal \cite{Hu2024}.  
\end{remark}

In the following theorem, we determine the parameters and Hamming weight distribution of the Gray image $\Phi(\mathscr{C}_{\mathcal{S}_3})$  under the assumption that $\mathpzc{A}\not\subseteq \mathpzc{B}$ (see Remark \ref{comp3}). 

\begin{thm}\label{ThN3}
When $\mathpzc{A} \not\subseteq \mathpzc{B},$ the Gray image $\Phi(\mathscr{C}_{\mathcal{S}_3})$ of the code $\mathscr{C}_{\mathcal{S}_3}$ is a linear code over $\mathbb{F}_q$ with parameters
$$\big[ 2q^{|\mathpzc{A}|+|\mathpzc{C}|}(q^m-q^{|\mathpzc{B}|}), ~{m+|\mathpzc{A}|+|\mathpzc{C}|},~2(q-1)q^{|\mathpzc{A}|+|\mathpzc{C}|-1}(q^m-q^{|\mathpzc{B}|})\big ]$$ and Hamming weight distribution given by  Table \ref{T3}. Furthermore, the code $\Phi(\mathscr{C}_{\mathcal{S}_3})$ is
   \begin{itemize}
\item a minimal, near-Griesmer, and distance-optimal code over $\mathbb{F}_q.$
\item  self-orthogonal  if either $q=2,$ or $q=3.$
\item  a $2$-weight code if $[m]=\mathpzc{A}\cup\mathpzc{B}.$ Otherwise, the code $\Phi(\mathscr{C}_{\mathcal{S}_3})$ is a $3$-weight code over $\mathbb{F}_q.$
 \end{itemize} 
\end{thm}
\begin{proof}
The result follows immediately by applying Theorem \ref{Th3}, Lemmas \ref{Lem1} and \ref{Matrix}, Theorems 1.4.8(ii) and 1.4.10(i) of \cite{HuffF}, and using the Griesmer bound \eqref{GB} and Remark \ref{RRk}.
\end{proof}

\begin{remark}\label{comp3.2}\begin{description}\item[(a)] When $\mathpzc{A} \not\subseteq \mathpzc{B},$ we have $\mathpzc{A}\cap \mathpzc{B}^c\neq\emptyset.$
In this case, let us choose $w_2\in \mathbb{F}_q^m$ such that $\supp(w_2)=\mathpzc{A}\cap \mathpzc{B}^c.$ One can easily see that $w_2\in \Delta_{\mathpzc{A}}\cap \Delta_{\mathpzc{B}}^c.$ Accordingly, by taking $w_1 = w_2,$ we see, for each $w_3 \in \Delta_\mathpzc{C},$ that
$$(w_2, w_3, w_2 )\in \mathcal{N}_3.$$
Here, we assert that the vector  $(w_2, w_3, w_2)$ appears exactly once in the multiset $\mathcal{N}_3.$ 
To prove this assertion, we suppose, on the contrary, that $(w_2, w_3, w_2 ) = (w_1'+w_2', w_3', w_1' )$ for some $w_1' \in \Delta_\mathpzc{A},$ $w_2' \in \Delta_\mathpzc{B}^c,$ and $w_3' \in \Delta_\mathpzc{C}.$ This holds if and only if $w_2' = \mathbf{0},$ $w_2 = w_1',$ and $w_3 = w_3'.$ This is a contradiction, as $\mathbf{0} \notin \Delta_\mathpzc{B}^c.$ Therefore, the vector $(w_2, w_3, w_2)$ appears exactly once in the multiset $\mathcal{N}_3,$ and hence its transpose  appears exactly once as a column of the matrix $\mathscr{G}_{\mathcal{N}_3}$.  From this, it follows that the code $\Phi(\mathscr{C}_{\mathcal{S}_3})$ is not a double repetition of the code $\mathcal{D}_3$ spanned  by the rows of  $\widehat{\mathscr{G}}_{\mathcal{N}_3}$ (as defined in Remark \ref{comp3}),  even though the code $\Phi(\mathscr{C}_{\mathcal{S}_3})$ and the double repetition code of $\mathcal{D}_3$  have the same parameters. In particular, when $[m]=\mathpzc{A} \cup \mathpzc{B},$ the code $\Phi(\mathscr{C}_{\mathcal{S}_3})$ considered in Theorem \ref{ThN3} is  not equivalent to the double repetition code of  $\mathcal{D}_3,$ even though the code $\Phi(\mathscr{C}_{\mathcal{S}_3})$ and the double repetition code of $\mathcal{D}_3$  have the same parameters and Hamming weight distribution. On the other hand, when $\mathpzc{A} \cup \mathpzc{B}\neq[m],$ the code $\Phi(\mathscr{C}_{\mathcal{S}_3})$ is a $3$-weight code and the double repetition code of $\mathcal{D}_3$ is a $2$-weight code. 
\item[(b)] 
When $\mathpzc{C}\subseteq\mathpzc{B},$  the parameters and  Hamming weight distribution of the code $\Phi(\mathscr{C}_{\mathcal{S}_3})$ over $\mathbb{F}_q$ coincide with that of the $3$-weight code studied in Theorem 5 of Chen \etal \cite{Chen2025}.  \item[(c)] Finally, when $\mathpzc{C}\not\subseteq\mathpzc{B},$  the Hamming weight distribution of the code $\Phi(\mathscr{C}_{\mathcal{S}_3})$  differs from that  of the $3$-weight code studied in Theorem 5 of Chen \etal \cite{Chen2025}. Consequently, in this case, the code $\Phi(\mathscr{C}_{\mathcal{S}_3})$ is  not equivalent to the code considered in Theorem 5 of Chen \etal \cite{Chen2025}. 
\end{description}\end{remark}

The following example illustrates Theorem \ref{ThN3}.
\begin{example}\label{Ex4.3}
Let $q=4$ and $m=4,$ and let $\mathpzc{A}=\{2\},$ $\mathpzc{B}=\{1,3\},$ and $\mathpzc{C}=\{1\}$ be subsets of $[m]=[4]=\{1,2,3,4\}.$ By carrying out computations in Magma \cite{Magma}, we see that the code $\Phi(\mathscr{C}_{\mathcal{S}_3})$ is a $3$-weight, minimal, and near-Griesmer distance-optimal linear $[7680,6,5760]$-code over $\mathbb{F}_4$ with Hamming weight enumerator $W_{\Phi(\mathscr{C}_{\mathcal{S}_3})}(Z)=1+4068Z^{5760}+24Z^{5952}+3Z^{6144}.$ This is in agreement  with Theorem \ref{ThN3}. 
\end{example}

In the following theorem, we determine the parameters and Hamming weight distribution of the Gray image $\Phi(\mathscr{C}_{\mathcal{S}_4})$ of the code $\mathscr{C}_{\mathcal{S}_4}.$ 
\begin{thm}\label{ThN4}
\begin{itemize}
\item[(a)] When $\mathpzc{A}\subseteq \mathpzc{B},$ the Gray image $\Phi(\mathscr{C}_{\mathcal{S}_4})$ of the code $\mathscr{C}_{\mathcal{S}_4}$ is a linear code over $\mathbb{F}_q$ with parameters 
$$\big[ 2(q^m-q^{|\mathpzc{C}|})(q^{|\mathpzc{A}|}-1)q^{|\mathpzc{B}|}, ~{m+|\mathpzc{A}|+|\mathpzc{B}|},~2(q-1)q^{|\mathpzc{B}|-1}\big((q^m-q^{|\mathpzc{C}|})(q^{|\mathpzc{A}|}-1)-q^{|\mathpzc{C}|}\big)\big]$$ and Hamming weight distribution given by  Table \ref{T4}. Furthermore, the code $\Phi(\mathscr{C}_{\mathcal{S}_4})$ is
\begin{itemize}
\item[$\bullet$] a $3$-weight code if $|\mathpzc{A}| + |\mathpzc{C}| = m.$
\item[$\bullet$] a $4$-weight code if $|\mathpzc{A}| + |\mathpzc{C}| \neq m.$
\end{itemize}

\item[(b)] When $\mathpzc{A}\not\subseteq \mathpzc{B},$ the Gray image $\Phi(\mathscr{C}_{\mathcal{S}_4})$ of the code $\mathscr{C}_{\mathcal{S}_4}$ is  a linear code over $\mathbb{F}_q$ with parameters $$\big[ 2(q^m-q^{|\mathpzc{C}|})(q^{|\mathpzc{A}|}-1)q^{|\mathpzc{B}|},~ {m+|\mathpzc{A}|+|\mathpzc{A}\cup \mathpzc{B}|},~(q-1)(q^{m}-q^{|\mathpzc{C}|})q^{|\mathpzc{A}|+|\mathpzc{B}|-1}\big ]$$ and Hamming weight distribution given by Table \ref{T4}. Furthermore, the code $\Phi(\mathscr{C}_{\mathcal{S}_4})$ is
\begin{itemize}
\item[$\bullet$] a $4$-weight code if $|\mathpzc{A}| + |\mathpzc{C}| = m + 1,$ $q=2,$ and $|\mathpzc{A}| = 2.$
\item[$\bullet$] a $5$-weight code if either $|\mathpzc{A}| + |\mathpzc{C}| = m,$ or
$|\mathpzc{A}| + |\mathpzc{C}| = m + 1,$ $q=2,$ and $|\mathpzc{A}| \geq 3.$
\item[$\bullet$] a $6$-weight code otherwise.
\end{itemize}
\end{itemize}
Moreover, the code $\Phi(\mathscr{C}_{\mathcal{S}_4})$ is 
 self-orthogonal if either $q=2,$ or $q=3.$
\end{thm}
\begin{proof}
 The desired result follows immediately by applying Theorem \ref{Th4}, Lemma \ref{Matrix}, and Theorems 1.4.8(ii) and 1.4.10(i) of \cite{HuffF} and Remark \ref{RRk}.
\end{proof}
\begin{remark}\label{Remark6.7}
\begin{itemize}
\item[(a)] Let $n\geq 3$ be an integer, and let $\mathpzc{E}_1$ and $\mathpzc{E}_2$ be non-empty subsets of $[n]$ satisfying $1\leq |\mathpzc{E}_1|, |\mathpzc{E}_2|<n$ and $q^n> q^{|\mathpzc{E}_1|}+q^{|\mathpzc{E}_2|}.$ By Theorem 3 of Hu \etal \cite{Hu2024}, there exists a linear code $\mathcal{D}_4$ over $\mathbb{F}_q$ with at most $5$ non-zero Hamming weights and parameters $[2(q-1)(q^n-q^{|\mathpzc{E}_1|}-q^{|\mathpzc{E}_2|}+q^{|\mathpzc{E}_1\cap\mathpzc{E}_2|}),~n,~2(q-1)(q^{n-1}-q^{|\mathpzc{E}_1|-1}-q^{|\mathpzc{E}_2|-1})].$ On the other hand, when $\mathpzc{A}\subseteq \mathpzc{B},$ we observe that the code $\Phi(\mathscr{C}_{\mathcal{S}_4})$ considered in Theorem \ref{ThN4}(a) has the same parameters as the code $\mathcal{D}_4,$ under the identification $n=m+|\mathpzc{A}|+|\mathpzc{B}|,$ $|\mathpzc{E}_1|=m+|\mathpzc{B}|,$ $|\mathpzc{E}_2|=|\mathpzc{A}|+|\mathpzc{B}|+|\mathpzc{C}|,$ and $|\mathpzc{E}_1\cap\mathpzc{E}_2|=|\mathpzc{B}|+|\mathpzc{C}|.$  However,  if $|\mathpzc{E}_1\cap\mathpzc{E}_2|>|\mathpzc{B}|+|\mathpzc{C}|,$ then we see that the code $\Phi(\mathscr{C}_{\mathcal{S}_4})$ attains the same dimension and Hamming distance as the code $\mathcal{D}_4,$ while having a smaller length. Moreover, the Hamming weight distributions of these two codes are distinct. 

\item[(b)] Using Theorem 3 of Hu \etal \cite{Hu2024}, the code $\mathcal{D}_5$  is a linear code over $\mathbb{F}_q$ with parameters $\big[(q^m-q^{|\mathpzc{C}|})(q^{|\mathpzc{A}|}-1)q^{|\mathpzc{B}|},~ m+|\mathpzc{A}|+|\mathpzc{B}|,~ (q-1) (q^{m+|\mathpzc{A}|+|\mathpzc{B}|-1} - q^{m+|\mathpzc{B}| - 1} - q^{|\mathpzc{A}|+|\mathpzc{B}|+|\mathpzc{C}| - 1}) \big]$ and has at most $5$ non-zero Hamming weights, under the identification $n=m+|\mathpzc{A}|+|\mathpzc{B}|,$ $|\mathpzc{E}_1|=m+|\mathpzc{B}|,$ $|\mathpzc{E}_2|=|\mathpzc{A}|+|\mathpzc{B}|+|\mathpzc{C}|,$ and $|\mathpzc{E}_1\cap\mathpzc{E}_2|=|\mathpzc{B}|+|\mathpzc{C}|.$  On the other hand, when $\mathpzc{A}\not\subseteq \mathpzc{B},$ the code $\Phi(\mathscr{C}_{\mathcal{S}_4}),$ studied in Theorem \ref{ThN4}(b), has twice the length of $\mathcal{D}_5,$ while its dimension is increased by $|\mathpzc{A}\cup\mathpzc{B}|-|\mathpzc{B}|=|\mathpzc{A}|-|\mathpzc{A}\cap\mathpzc{B}| > 0$ and its Hamming distance is increased by $(q-1)q^{m+|\mathpzc{B}|-1}.$ Moreover, the code $\Phi(\mathscr{C}_{\mathcal{S}_4})$ is
a $4$-weight code if $|\mathpzc{A}| + |\mathpzc{C}| = m + 1,$ $q=2,$ and $|\mathpzc{A}| = 2;$ a $5$-weight code if either $|\mathpzc{A}| + |\mathpzc{C}| = m,$ or
$|\mathpzc{A}| + |\mathpzc{C}| = m + 1,$ $q=2,$ and $|\mathpzc{A}| \geq 3;$
and a $6$-weight code otherwise.
\end{itemize}
\end{remark}
The following example illustrates Theorem \ref{ThN4}.
\begin{example}\label{Ex4.4}
Let $q=3$ and $m=3,$ and let $\mathpzc{A}=\{1,2\}$ and $\mathpzc{B}=\mathpzc{C}=\{1\}$ be subsets of $[m]=[3]=\{1,2,3\}.$ By carrying out computations in Magma \cite{Magma}, we see that the code $\Phi(\mathscr{C}_{\mathcal{S}_4})$ is a ternary $5$-weight and self-orthogonal  $[1152,7,432]$-code with Hamming weight enumerator $W_{\Phi(\mathscr{C}_{\mathcal{S}_4})}(Z)=1+4Z^{432}+176Z^{756}+1944Z^{768}+32Z^{810}+30Z^{864}.$ This is consistent with Theorem \ref{ThN4}. 
\end{example}

\section{Some constructions of projective few-weight codes over $\mathbb{F}_q$}\label{Sec5}
 In this section, we will provide two constructions of projective few-weight codes over $\mathbb{F}_q$ with the help of  the multisets $\mathcal{N}_2$ and $\mathcal{N}_4.$ To this end, we first prove the following useful lemma.
\begin{lemma}\label{Lem8}
Let $\mathcal{N}_2$ and $\mathcal{N}_4$ be the multisets as defined by \eqref{D2} and \eqref{D4}, respectively. The following hold.
\begin{itemize}
\item[(a)] If $\mathpzc{B} \subseteq \mathpzc{A},$ then the multiset $\mathcal{N}_2$ consists of distinct vectors. Furthermore, there exists a subset $\overline{\mathcal{N}}_2$ of $\mathcal{N}_2$ such that $|\overline{\mathcal{N}}_2|= \frac{|\mathcal{N}_2|}{q-1},$ each element in $\overline{\mathcal{N}}_2$ generates a distinct  one-dimensional subspace of $\mathbb{F}_q^{3m}$ over  $\mathbb{F}_q,$ and
$$\mathcal{N}_2 = \{\alpha x : x \in \overline{\mathcal{N}}_2 \text{ and } \alpha \in \mathbb{F}_q^\ast\}.$$ 
\item[(b)]  If $\mathpzc{A} \cap \mathpzc{B} = \emptyset,$ then the multiset $\mathcal{N}_4$ consists of distinct vectors.  Moreover, there exists a subset $\overline{\mathcal{N}}_4$ of $\mathcal{N}_4$ such that $|\overline{\mathcal{N}}_4|= \frac{|\mathcal{N}_4|}{q-1},$ each element in $\overline{\mathcal{N}}_4$ generates a distinct  one-dimensional subspace of $\mathbb{F}_q^{3m}$ over  $\mathbb{F}_q,$ and
$$\mathcal{N}_4 = \{\alpha x : x \in \overline{\mathcal{N}}_4 \text{ and } \alpha \in \mathbb{F}_q^\ast\}.$$
\end{itemize}
\end{lemma}
\begin{proof}
\begin{itemize}
\item[(a)] To prove the result, let us  assume that $\mathpzc{B} \subseteq \mathpzc{A}.$ Here, we have $\Delta_{\mathpzc{B}} \subseteq \Delta_{\mathpzc{A}},$ which holds if and only if $\Delta_{\mathpzc{A}}^c \cap \Delta_{\mathpzc{B}} = \emptyset.$ 

First of all, we assert that all vectors in the multiset $\mathcal{N}_2$ are distinct. To prove this assertion, we see that the multiset $\mathcal{N}_2$ can be written as $\mathcal{N}_2=\mathcal{T}_1 \cup \mathcal{T}_2,$ where 
\begin{eqnarray*}\label{SetT}\notag
  \mathcal{T}_1 &=&\big\{\big\{ (w_2,w_3,w_1)
  \in \mathbb{F}_q^{3m} : w_1 \in \Delta_{\mathpzc{A}}^c, w_2 \in \Delta_{\mathpzc{B}}, w_3 \in \Delta_{\mathpzc{C}} \text{ with }|\mathpzc{A}|<m \big\} \big\} \text{ and }\\
\mathcal{T}_2 &=&\big\{\big\{(w_1+w_2,w_3,w_1)
\in \mathbb{F}_q^{3m} : w_1 \in \Delta_{\mathpzc{A}}^c, w_2 \in \Delta_{\mathpzc{B}}, w_3 \in \Delta_{\mathpzc{C}} \text{ with }|\mathpzc{A}|<m \big\}\big\}.
\end{eqnarray*}
Note that both $\mathcal{T}_1$ and $\mathcal{T}_2$ consist of distinct vectors. So to prove this assertion, it suffices to show that $\mathcal{T}_1 \cap \mathcal{T}_2 = \emptyset.$ For this, we suppose, on the contrary, that $z\in \mathcal{T}_1 \cap \mathcal{T}_2.$ By \eqref{D2} and \eqref{D4}, we see that there exist   $w_1,v_1 \in \Delta_{\mathpzc{A}}^c,$ $w_2,v_2 \in \Delta_{\mathpzc{B}},$ and $w_3,v_3 \in \Delta_{\mathpzc{C}}$ such that 
$$z=(w_2,w_3,w_1)
 =(v_1 + v_2,v_3,v_1).$$ This gives $w_2 = v_1 + v_2,$ $w_3=v_3,$ and $w_1 = v_1.$ This implies that $v_1 = w_2 - v_2 \in \Delta_{\mathpzc{A}}^c \cap \Delta_{\mathpzc{B}},$ which is a contradiction to the fact that $\Delta_{\mathpzc{A}}^c \cap \Delta_{\mathpzc{B}} = \emptyset.$ This proves the assertion that all vectors in the multiset $\mathcal{N}_2$ are distinct. 

Further, if $\Lambda$ is a simplicial complex  of $\mathbb{F}_q^m$ with a single maximal element, then one can easily observe that $\alpha x \in \Lambda$ and $\alpha y \in \Lambda^c$ for all  $x \in \Lambda,$  $y \in \Lambda^c,$ and $\alpha \in \mathbb{F}_q^\ast.$ 
Using this fact, one can easily observe that  $\alpha z \in \mathcal{N}_2$ for all  $z \in \mathcal{N}_2$ and $\alpha \in \mathbb{F}_q^\ast.$ Accordingly, one can define an equivalence relation $\sim$ on $\mathcal{N}_2$ as follows:  For  $x, y \in \mathcal{N}_2,$ set $x\sim y$   if and only if there exists $\alpha \in \mathbb{F}_q^\ast$ such that $x = \alpha y.$ Now, 
let $\overline{\mathcal{N}}_2$ denote a complete set of representatives of all the distinct equivalence classes of $\mathcal{N}_2$ under $\sim.$ Here, one can easily see that $|\overline{\mathcal{N}}_2|= \frac{|\mathcal{N}_2|}{q-1},$ each vector in the set $\overline{\mathcal{N}}_2$ generates a distinct one-dimensional subspace of $\mathbb{F}_q^{3m}$  over $\mathbb{F}_q$ and that
$\mathcal{N}_2 = \{\alpha x : x \in \overline{\mathcal{N}}_2 \text{ and } \alpha \in \mathbb{F}_q^\ast\}.$ 
\item[(b)] Here, let us assume that $\mathpzc{A} \cap \mathpzc{B} = \emptyset.$ This holds if and only if $\Delta_{\mathpzc{A}}^\ast \cap \Delta_{\mathpzc{B}} = \emptyset.$ Using this fact and working as in part (a), we see that all vectors in the multiset $\mathcal{N}_4$ are distinct and that $\alpha z \in \mathcal{N}_4$ for all $\alpha \in \mathbb{F}_q^\ast$ and $z \in \mathcal{N}_4.$  Now, 
let $\overline{\mathcal{N}}_4$ denote a complete set of representatives of all the distinct equivalence classes of $\mathcal{N}_4$ under the equivalence relation $\sim$ as defined in part (a). Here, one can easily see that $|\overline{\mathcal{N}}_4|= \frac{|\mathcal{N}_4|}{q-1},$ each vector in the set $\overline{\mathcal{N}}_4$ generates a distinct one-dimensional subspace of $\mathbb{F}_q^{3m}$  over $\mathbb{F}_q$ and that
$\mathcal{N}_4 = \{\alpha x : x \in \overline{\mathcal{N}}_4 \text{ and } \alpha \in \mathbb{F}_q^\ast\},$ which proves part (b). 
\end{itemize}
\vspace{-4mm}\end{proof}

Henceforth, when  $\mathpzc{B} \subseteq \mathpzc{A}\subsetneq [m],$ we assume that $\overline{\mathcal{N}}_2$ is a subset of $\mathcal{N}_2$ as described in Lemma \ref{Lem8}(a).
Similarly, when $\mathpzc{A} \cap \mathpzc{B} = \emptyset,$ we assume that $\overline{\mathcal{N}}_4$ is a subset of $\mathcal{N}_4$ as described in Lemma \ref{Lem8}(b). Further, let  $G_{\overline{\mathcal{N}}_2} \in M_{3m \times |\overline{\mathcal{N}}_2|}(\mathbb{F}_q)$ and $G_{\overline{\mathcal{N}}_4} \in M_{3m \times |\overline{\mathcal{N}}_4|}(\mathbb{F}_q)$ be the matrices such that their transposed columns form the sets $\overline{\mathcal{N}}_2$ and $\overline{\mathcal{N}}_4,$ respectively. Now, let $\mathcal{C}_{\overline{\mathcal{N}}_2}$ and $\mathcal{C}_{\overline{\mathcal{N}}_4}$ be the linear codes over $\mathbb{F}_q$ spanned by the matrices $G_{\overline{\mathcal{N}}_2}$ and $G_{\overline{\mathcal{N}}_4},$ respectively. Note that the sets $\overline{\mathcal{N}}_2$ and $\overline{\mathcal{N}}_4$ are not uniquely fixed; hence, the codes $\mathcal{C}_{\overline{\mathcal{N}}_2}$ and $\mathcal{C}_{\overline{\mathcal{N}}_4}$ are  uniquely defined only up to monomial equivalence. In particular, monomially equivalent codes possess identical parameters and Hamming weight distributions. Now, as any two columns in each of the matrices  $G_{\overline{\mathcal{N}}_2}$ and $G_{\overline{\mathcal{N}}_4}$ are linearly independent over $\mathbb{F}_q,$ we see, by Lemma \ref{Lem2}(a), that both $\mathcal{C}_{\overline{\mathcal{N}}_2}$ and $\mathcal{C}_{\overline{\mathcal{N}}_4}$ are projective  codes over $\mathbb{F}_q.$

In the following theorem, we assume that $\mathpzc{B} \subseteq \mathpzc{A}\subsetneq [m]$ and  explicitly determine the parameters and  Hamming weight distribution of the projective code $\mathcal{C}_{\overline{\mathcal{N}}_2}$ over $\mathbb{F}_q.$ 
We show that the code $\mathcal{C}_{\overline{\mathcal{N}}_2}$ is a projective $2$-weight code over $\mathbb{F}_q$ if $q=2,$ $\mathpzc{B}=\mathpzc{A},$ and $|\mathpzc{A}|=m-1;$ a projective $3$-weight code if either $\mathpzc{B} = \mathpzc{A} \subsetneq [m]$ and $q \geq 3,$ or $\mathpzc{B} = \mathpzc{A},$ $|\mathpzc{A}| \leq m-2,$ and $q=2$; a projective $4$-weight code if $\mathpzc{B} \subsetneq \mathpzc{A} \subsetneq[m],$ $q=2,$ and $|\mathpzc{A}| = m-1;$ and a projective $5$-weight code otherwise. We also observe that $\mathcal{C}_{\overline{\mathcal{N}}_2}$ is a self-orthogonal code for $q = 2$ or $3.$
\begin{thm}\label{Th7} 
When $\mathpzc{B} \subseteq \mathpzc{A}\subsetneq[m],$ the code $\mathcal{C}_{\overline{\mathcal{N}}_2}$ is a projective  code over $\mathbb{F}_q$ with parameters  
$$\left[\frac{2q^{|\mathpzc{B}|+|\mathpzc{C}|}(q^m-q^{|\mathpzc{A}|})}{q-1},~ {2m+|\mathpzc{C}|},~ (q^m-\rho q^{|\mathpzc{A}|})q^{|\mathpzc{B}|+|\mathpzc{C}|-1}\right]$$ and
 Hamming weight distribution given by Table \ref{T6}, where $\rho=0$ if $\mathpzc{B}= \mathpzc{A},$ while $\rho=1$ if $\mathpzc{B}\subsetneq \mathpzc{A}.$  As a consequence, the code $\mathcal{C}_{\overline{\mathcal{N}}_2}$ is 
 \begin{itemize}
\item a projective $2$-weight code if $q=2,$ $\mathpzc{B}=\mathpzc{A},$ and $|\mathpzc{A}| = m-1.$
\item  a projective $3$-weight code if either $\mathpzc{B}=\mathpzc{A}\subsetneq[m]$ and $q\geq 3,$ or $\mathpzc{B}=\mathpzc{A},$ $|\mathpzc{A}| \leq m-2,$ and $q=2.$
\item  a projective $4$-weight code if  $\mathpzc{B} \subsetneq \mathpzc{A}\subsetneq[m],$ $q = 2,$ and $|\mathpzc{A}| = m-1.$
\item a projective $5$-weight code otherwise.
 \end{itemize}
Furthermore, the code $\mathcal{C}_{\overline{\mathcal{N}}_2}$ is  self-orthogonal if either $q = 2,$ or $q=3.$
\begin{table}[H]
\centering
    \begin{tabular}{ |c|c|} 
 \hline Hamming weight & Frequency \\ \hline
$0$ & $1$ \\
 \hline
$(q^m-q^{|\mathpzc{A}|})q^{|\mathpzc{B}|+|\mathpzc{C}|-1}$ & $2(q^{m-|\mathpzc{B}|}-q^{m-|\mathpzc{A} \cup \mathpzc{B}|})$ \\
 \hline
$2(q^m-q^{|\mathpzc{A}|})q^{|\mathpzc{B}|+|\mathpzc{C}|-1}$ & $q^{2m+|\mathpzc{C}|} - 2q^{m-|\mathpzc{A}|}(q^{m-|\mathpzc{B}|}-q^{m-|\mathpzc{A}\cup\mathpzc{B}|}) - q^{2m-|\mathpzc{A}\cup\mathpzc{B}|-|\mathpzc{A}|}$ \\
\hline
$(2q^m-q^{|\mathpzc{A}|})q^{|\mathpzc{B}|+|\mathpzc{C}|-1}$ & $2(q^{m-|\mathpzc{B}|}-q^{m-|\mathpzc{A} \cup \mathpzc{B}|})(q^{m-|\mathpzc{A}|}-1)$ \\
\hline
$2q^{m+|\mathpzc{B}|+|\mathpzc{C}|-1}$ & $(q^{m-|\mathpzc{A}|}-1)(q^{m-|\mathpzc{A} \cup \mathpzc{B}|}-1)+q^{m-|\mathpzc{A}|}-q^{m-|\mathpzc{A}\cup\mathpzc{B}|}$ \\
 \hline
 $q^{m+|\mathpzc{B}|+|\mathpzc{C}|-1}$ & $2(q^{m-|\mathpzc{A} \cup \mathpzc{B}|}-1)$ \\
\hline
\end{tabular}
\caption{The Hamming weight distribution of the code $\mathcal{C}_{\overline{\mathcal{N}}_2}$ over $\mathbb{F}_q$ when $\mathpzc{B} \subseteq \mathpzc{A}\subsetneq[m]$}\label{T6} 
\end{table}

\end{thm}
\begin{proof}
The desired result follows from  Theorem \ref{Th2} and the description of the set $\overline{\mathcal{N}}_2$  given in Lemma \ref{Lem8}(a). We also see, by Theorems 1.4.8(ii) and 1.4.10(i) of \cite{HuffF},  that the code $\mathcal{C}_{\overline{\mathcal{N}}_2}$ is  self-orthogonal for $q=2$ or $3.$
\end{proof}
 Shi and Sol{\'e} \cite[Sec. 6]{Shi2019} posed an open question concerning the construction of new projective $3$-weight codes  over $\mathbb{F}_q,$  with non-zero Hamming weights summing to  $\frac{3 (q-1)}{q}$ times the code length. In the following corollary, we address this question in the special case $q=4,$ by identifying an infinite family of quaternary projective $3$-weight codes whose non-zero Hamming weights sum to $\frac{9}{4}$ times the code length.
\begin{cor} \label{SHI} 
Let $q=4,$ $\mathpzc{B}=\mathpzc{A},$ and $|\mathpzc{A}|=m-1,$ and let us define $\theta = 2^{4m+2|\mathpzc{C}|-3}.$ The code $\mathcal{C}_{\overline{\mathcal{N}}_2}$ is a projective $3$-weight code over $\mathbb{F}_4$ with parameters $\left[\theta, {2m+|\mathpzc{C}|},\frac{\theta}{2}\right]$ and  has non-zero Hamming weights $\mathtt{w}_1=\frac{\theta}{2},$   $\mathtt{w}_2=\frac{3\theta}{4},$ and  $\mathtt{w}_3= \theta$ with frequencies $A_{\mathtt{w}_1}=6,$  $A_{\mathtt{w}_2}=4^{2m+|\mathpzc{C}|}-16,$ and $A_{\mathtt{w}_3}=9,$ respectively. Moreover, we have $\mathtt{w}_1+\mathtt{w}_2+\mathtt{w}_3= \frac{9\theta}{4}.$ \end{cor}
In the appendix, we will construct strongly walk-regular graphs using the codes constructed in the above corollary.  Now, in the following theorem, we assume that $\mathpzc{A} \cap \mathpzc{B}=\emptyset,$  $\mathpzc{C} \neq [m],$ and $|\mathpzc{A}|\geq 2,$ and  explicitly determine the parameters and  Hamming weight distribution of the projective code $\mathcal{C}_{\overline{\mathcal{N}}_4}$ over $\mathbb{F}_q.$ We show that the code $\mathcal{C}_{\overline{\mathcal{N}}_4}$ is a projective $4$-weight code over $\mathbb{F}_q$ if $|\mathpzc{A}| + |\mathpzc{C}| = m + 1,$ $q=2,$ and $|\mathpzc{A}| = 2;$ a projective $5$-weight code over $\mathbb{F}_q$ if either $|\mathpzc{A}| + |\mathpzc{C}| = m,$ or $|\mathpzc{A}| + |\mathpzc{C}| = m + 1,$ $q=2,$ and $|\mathpzc{A}| \geq 3;$ and a projective $6$-weight code over $\mathbb{F}_q$ otherwise. We also observe that $\mathcal{C}_{\overline{\mathcal{N}}_4}$ is a self-orthogonal code  for $q = 2$ or $3.$
\begin{thm}\label{Th8}
When $\mathpzc{A}\cap\mathpzc{B}=\emptyset,$ $\mathpzc{C} \neq [m],$ and $|\mathpzc{A}|\geq 2,$ the code $\mathcal{C}_{\overline{\mathcal{N}}_4}$ is a projective code over $\mathbb{F}_q$ with parameters 
$\left[\frac{2(q^m-q^{|\mathpzc{C}|})(q^{|\mathpzc{A}|}-1)q^{|\mathpzc{B}|}}{q-1},~  {m+2|\mathpzc{A}|+|\mathpzc{B}|},~ (q^{m}-q^{|\mathpzc{C}|})q^{|\mathpzc{A}|+|\mathpzc{B}|-1}\right]$ and
 Hamming weight distribution given by Table \ref{T7}. As a consequence, the code $\mathcal{C}_{\overline{\mathcal{N}}_4}$ is
\begin{itemize}
\item[$\bullet$] a projective $4$-weight code if $|\mathpzc{A}| + |\mathpzc{C}| = m + 1,$ $q=2,$ and $|\mathpzc{A}| = 2.$
\item[$\bullet$] a projective $5$-weight code if either $|\mathpzc{A}| + |\mathpzc{C}| = m,$ or
$|\mathpzc{A}| + |\mathpzc{C}| = m + 1,$ $q=2,$ and $|\mathpzc{A}| \geq 3.$
\item[$\bullet$] a projective $6$-weight code otherwise.
\end{itemize}
Furthermore, the code $\mathcal{C}_{\overline{\mathcal{N}}_4}$ is self-orthogonal if either $q = 2,$ or $q=3.$
\begin{table}[H]
\centering
    \begin{tabular}{ |c|c|} 
 \hline Hamming weight & Frequency \\ \hline
$0$ & $1$ \\
\hline
$(q^{m}-q^{|\mathpzc{C}|})q^{|\mathpzc{A}|+|\mathpzc{B}|-1}$ & $2(q^{|\mathpzc{A}|}-1)$\\
\hline
$2q^{|\mathpzc{B}|-1}\big((q^m-q^{|\mathpzc{C}|})(q^{|\mathpzc{A}|}-1)-q^{|\mathpzc{C}|}\big)$ & $(q^{|\mathpzc{A}|}-1)^2(q^{m-|\mathpzc{C}|}-1)$ \\
\hline
$2(q^m-q^{|\mathpzc{C}|})(q^{|\mathpzc{A}|}-1)q^{|\mathpzc{B}|-1}$ & $q^{m+2|\mathpzc{A}|+|\mathpzc{B}|}-q^{m+2|\mathpzc{A}|-|\mathpzc{C}|}$ \\
\hline
$2(q^m-q^{|\mathpzc{C}|})q^{|\mathpzc{A}|+|\mathpzc{B}|-1}$ & $(q^{|\mathpzc{A}|}-1)^2$\\
\hline
$q^{|\mathpzc{B}|-1}\big((2q^m-q^{|\mathpzc{C}|})(q^{|\mathpzc{A}|}-1)-q^{|\mathpzc{C}|}\big)$ & $2(q^{|\mathpzc{A}|}-1)(q^{m-|\mathpzc{C}|}-1)$ \\
\hline
$2q^{m+|\mathpzc{B}|-1}(q^{|\mathpzc{A}|}-1)$ & $q^{m-|\mathpzc{C}|}-1$ \\
\hline
\end{tabular}
\caption{The Hamming weight distribution of the code $\mathcal{C}_{\overline{\mathcal{N}}_4}$ over $\mathbb{F}_q$ when $\mathpzc{A}\cap\mathpzc{B}=\emptyset,$ $\mathpzc{C} \neq [m],$ and $|\mathpzc{A}|\geq 2$}\label{T7} 
\end{table}
\end{thm}
\begin{proof}
The  result follows from Theorem \ref{Th4} and the description of the set $\overline{\mathcal{N}}_4$ given in Lemma \ref{Lem8}(b). Moreover, for $q=2$ or $3,$ we see, by Theorems 1.4.8(ii) and 1.4.10(i) of \cite{HuffF}, that the code $\mathcal{C}_{\overline{\mathcal{N}}_4}$ is self-orthogonal.
\end{proof}

\begin{remark}\label{NewRemark}We observe that the parameters of the projective codes  $\mathcal{C}_{\overline{\mathcal{N}}_2}$ and $\mathcal{C}_{\overline{\mathcal{N}}_4}$ over $\mathbb{F}_q$ are distinct from those recorded in Table II of Hu \etal\ \cite[Sec. 5]{Hu2024} as well as those listed in  Wu \etal \cite{Wu2024S}.  \end{remark}

While the parameters and Hamming weight enumerators of the duals $\mathcal{C}_{\overline{\mathcal{N}}_2}^\perp$ and $\mathcal{C}_{\overline{\mathcal{N}}_4}^\perp$ over $\mathbb{F}_q$ can be derived using the MacWilliams identity together with Theorems \ref{Th7} and \ref{Th8}, we adopt an alternative approach in the next two theorems. Specifically, we obtain their parameters by examining the linear independence relations among the vectors of $\overline{\mathcal{N}}_2$ and $\overline{\mathcal{N}}_4,$ thereby gaining a clearer understanding of their underlying structure.

In the following theorem, we  determine the parameters of the  dual $\mathcal{C}_{\overline{\mathcal{N}}_2}^\perp$  over $\mathbb{F}_q.$  As a consequence, we obtain an infinite family of binary distance-optimal codes with Hamming distance $4.$ 
\begin{thm}\label{Th9}
    Suppose that $\mathpzc{B}\subseteq \mathpzc{A}\subsetneq[m].$ The following hold.
\begin{itemize}
\item[(a)] When either $q\geq 3,$ or $q=2$ and $|\mathpzc{A}|\leq m-2,$  the dual $\mathcal{C}_{\overline{\mathcal{N}}_2}^\perp$ is a linear code over $\mathbb{F}_q$ with parameters $$\bigg[\frac{2q^{|\mathpzc{B}|+|\mathpzc{C}|}(q^m-q^{|\mathpzc{A}|})}{q-1},~  \frac{2q^{|\mathpzc{B}|+|\mathpzc{C}|}(q^m-q^{|\mathpzc{A}|})}{q-1}-2m-|\mathpzc{C}|,~ 3\bigg ].$$
\item[(b)] When $q=2$ and $|\mathpzc{A}|=m-1,$ the dual $\mathcal{C}_{\overline{\mathcal{N}}_2}^\perp$ is a  binary distance-optimal  code  with parameters $$[2^{m+|\mathpzc{B}|+|\mathpzc{C}|},~ 2^{m+|\mathpzc{B}|+|\mathpzc{C}|}-2m-|\mathpzc{C}|,~ 4].$$
    \end{itemize}
\end{thm}
\begin{proof}
To prove the result, we note, by Theorem \ref{Th7}, that $\mathcal{C}_{\overline{\mathcal{N}}_2}$ is a linear code of length
$\frac{2q^{|\mathpzc{B}|+|\mathpzc{C}|}(q^m - q^{|\mathpzc{A}|})}{q - 1}$ and dimension
$2m + |\mathpzc{C}|$
over $\mathbb{F}_q.$ This, by Theorem 7.3 of \cite{Hill1986}, implies that the dual $\mathcal{C}_{\overline{\mathcal{N}}_2}^\perp$ is a linear code of length
$\frac{2q^{|\mathpzc{B}|+|\mathpzc{C}|}(q^m - q^{|\mathpzc{A}|})}{q - 1}$ and dimension
$\frac{2q^{|\mathpzc{B}|+|\mathpzc{C}|}(q^m - q^{|\mathpzc{B}|})}{q - 1} - 2m - |\mathpzc{C}|$
over $\mathbb{F}_q.$ 

Next, to determine the Hamming distance of the code $\mathcal{C}_{\overline{\mathcal{N}}_2}^\perp,$ we note, by Lemma \ref{Lem8}(a), that all vectors in the multiset $\mathcal{N}_2$ are distinct, $|\overline{\mathcal{N}}_2|=\frac{|\mathcal{N}_2|}{q-1},$ each vector in $\overline{\mathcal{N}}_2$ generates a distinct one-dimensional subspace of $\mathbb{F}_q^{3m}$ over $\mathbb{F}_q,$  and
that $\mathcal{N}_2 = \{\alpha x : x \in \overline{\mathcal{N}}_2 \text{ and } \alpha \in \mathbb{F}_q^\ast\}.$ This, by Lemma \ref{Lem2}(a), implies that $d_H(\mathcal{C}_{\overline{\mathcal{N}}_2}^\perp)\geq 3.$
Now, to show that $d_H(\mathcal{C}_{\overline{\mathcal{N}}_2}^\perp) = 3,$ we see, by Lemma \ref{Lem2}(b), that it suffices to produce three distinct vectors in the set $\overline{\mathcal{N}}_2$ that are linearly dependent over $\mathbb{F}_q.$ 
Towards this, we will produce three linearly dependent vectors in $\mathcal{N}_2$ such that any pair of vectors among them is linearly independent, so that their corresponding representatives in the set $\overline{\mathcal{N}}_2$  satisfy the desired condition.
To do this, let $v \in \Delta_{\mathpzc{B}} \setminus \{\mathbf{0}\}$ and $w \in \Delta_{\mathpzc{C}} \setminus \{\mathbf{0}\} $ be fixed. For  $i \in [m],$ let us define  $e_i=(0,0,\ldots,0,\underbrace{1}_{i\text{-th}},0,0,\ldots0)\in \mathbb{F}_q^m.$ We will now distinguish the following three cases: (i) $q \geq 3,$  (ii) $q=2$ and $|\mathpzc{A}|\leq m-2,$ and (iii) $q=2$ and $|\mathpzc{A}|= m-1.$

\begin{itemize}\item[(i)] Let $q \geq 3.$ Let $\kappa \in [m]\setminus\mathpzc{A}$ be fixed,
and let $\alpha\in \mathbb{F}_q^\ast$ be such that $1+\alpha\in \mathbb{F}_q^\ast.$ 
Now, let us define  $v_1,v_2\in \mathbb{F}_q^{3m}$ as follows:
\begin{equation*}
    v_1=(\mathbf{0},w,e_\kappa) \text{ and }
    v_2=(\mathbf{0},\mathbf{0},\alpha e_\kappa). 
\end{equation*}
Note that the vectors $v_1, v_2, v_1 + v_2 \in \mathcal{N}_2.$  We further observe that the vectors $v_1, v_2, v_1 + v_2 \in \mathcal{N}_2$ are linearly dependent over $\mathbb{F}_q,$  while any two of them are linearly independent over $\mathbb{F}_q.$ Accordingly, their corresponding representatives in the set $\overline{\mathcal{N}}_2$ satisfy the desired condition. From this, it follows that $d_H(\mathcal{C}_{\overline{\mathcal{N}}_2}^\perp) = 3.$

\item[(ii)] Let $q=2$ and $|\mathpzc{A}|\leq m-2.$ Here, one can easily observe that $\mathcal{N}_2=\overline{\mathcal{N}}_2.$ Further, as $|\mathpzc{A}|\leq m-2,$ we have $|[m]\setminus\mathpzc{A}|\geq 2.$ Let $i,j\in [m]\setminus\mathpzc{A}$ be such that $i\neq j.$ Now, let us define $x_1,x_2\in \mathbb{F}_2^{3m}$  as follows:
\begin{equation*}
    x_1=(\mathbf{0},w,e_i) \text{ and }
    x_2=(\mathbf{0},\mathbf{0},e_j).
\end{equation*}
Note that the vectors $x_1, x_2, x_1+x_2 \in \mathcal{N}_2.$  One can easily see that the vectors $x_1, x_2, x_1+x_2 \in \mathcal{N}_2$  are distinct and linearly dependent over $\mathbb{F}_2.$  From this, it follows that $d_H(\mathcal{C}_{\overline{\mathcal{N}}_2}^\perp)   = 3.$

\item[(iii)] Finally, let $q=2$ and $|\mathpzc{A}|=m-1.$ Here, we have $\overline{\mathcal{N}}_2=\mathcal{N}_2$ and $|[m]\setminus\mathpzc{A}|=1.$ Now, to show that $d_H(\mathcal{C}_{\overline{\mathcal{N}}_2}^\perp) = 4,$ it suffices, by Lemma \ref{Lem2}, to show that any three vectors of $\mathcal{N}_2$ are linearly independent over $\mathbb{F}_2$ and that there exist four vectors in $\mathcal{N}_2$ that are linearly dependent over $\mathbb{F}_2.$ Here, one can easily see that any three distinct vectors $z_1, z_2, z_3 \in \mathcal{N}_2$ are linearly independent over $\mathbb{F}_2$ if and only if $ z_3 \neq z_1 + z_2.$ 
Now, as $|[m]\setminus\mathpzc{A}|=1,$  let us assume that  $ [m]\setminus\mathpzc{A}=\{ \kappa \}.$ For all $a,b\in \Delta_{\mathpzc{A}}^c,$ we have $(a)_\kappa=(b)_\kappa=1,$ which gives $(a+b)_\kappa=0.$ This shows that  $a+b\notin \Delta_{\mathpzc{A}}^c.$ From this, it follows that  $z_3 = z_1 + z_2\notin  \mathcal{N}_2$ for all $z_1,z_2\in \mathcal{N}_2.$  Thus, any three vectors in the set $\mathcal{N}_2$ are linearly independent over $\mathbb{F}_2.$ 

Further, let us define $y_1,$ $y_2,$ $y_3 \in \mathbb{F}_2^{3m}$ as follows:   
 \begin{equation*}
y_1=(\mathbf{0},w,e_\kappa),\,
y_2=(\mathbf{0},\mathbf{0},e_\kappa), \text{ and }
y_3=(v,\mathbf{0},e_\kappa).
\end{equation*}
Note that $y_1, y_2, y_3, y_1+y_2+y_3 \in \mathcal{N}_2.$ Further, one can easily see that the vectors $y_1, y_2, y_3, y_1+y_2+y_3 \in \mathcal{N}_2$ are linearly dependent over $\mathbb{F}_2.$ This, by Lemma \ref{Lem2}, implies that  $d_H(\mathcal{C}_{\overline{\mathcal{N}}_2}^\perp) = 4.$ Moreover, we observe, using the Sphere-packing bound \eqref{SPB}, that the code $\mathcal{C}_{\overline{\mathcal{N}}_2}^\perp$ is distance-optimal.  \end{itemize}
This completes the proof of the theorem.
\end{proof}
In the following theorem, we  determine the parameters of the  dual $\mathcal{C}_{\overline{\mathcal{N}}_4}^\perp$ of the projective codes $\mathcal{C}_{\overline{\mathcal{N}}_4}$ over $\mathbb{F}_q.$  As a consequence, we obtain an infinite family binary distance-optimal codes. 
\begin{thm}\label{Th10}
When $q=2,$ $\mathpzc{A}\cap \mathpzc{B}=\emptyset,$ $|\mathpzc{C}|=m-1,$ and $|\mathpzc{A}|\geq 2,$ the dual $\mathcal{C}_{\overline{\mathcal{N}}_4}^\perp$ is a binary distance-optimal code  with parameters $$\big[2^{m+|\mathpzc{A}|+|\mathpzc{B}|}-2^{m+|\mathpzc{B}|},~  2^{m+|\mathpzc{A}|+|\mathpzc{B}|}-2^{m+|\mathpzc{B}|}-m-2|\mathpzc{A}|-|\mathpzc{B}|,~ 4\big].$$
\end{thm}
\begin{proof}
The desired result follows by working as in Theorem \ref{Th9} and applying Theorem \ref{Th8} and using the Sphere-packing bound \eqref{SPB}.
\end{proof}

\section{Some additional applications}\label{add}
In this section, we will explore two additional applications of the results derived in Sections \ref{GI} and \ref{Sec5}. More precisely, we will study minimal access structures of Massey's secret sharing schemes based on the duals of  minimal codes obtained in Theorem \ref{ThN3}. We will also construct infinite families of LRCs with locality either $2$ or $3$ through the projective codes constructed  in Theorems \ref{Th7} and \ref{Th8}.

\subsection{Minimal access structures of Massey's secret sharing schemes based on the duals of our minimal codes}\label{Sec7}

In this section, we will examine the minimal access structures of Massey's secret sharing schemes based on the duals of minimal codes constructed in Theorem \ref{ThN3}, via the Gray image $\Phi(\mathscr{C}_{\mathcal{S}_3}).$ We will also obtain the number of dictatorial participants in this scheme. We begin by outlining Massey's secret sharing scheme based on linear codes over $\mathbb{F}_q$ \cite{Massey1993}.

Let $\mathcal{C}$ be a linear code of length $n$ and dimension $k$ over $\mathbb{F}_q$ with a generator matrix  
$\mathcal{H} = \big[h_0^T ~~ h_1^T ~~ \cdots ~~ h_{n-1}^T \big],$ where $h_i\in \mathbb{F}_q^k \setminus \{\mathbf{0}\}$ for $0 \leq i \leq n-1.$  
 In Massey's secret sharing scheme based on the  code $\mathcal{C},$ the system consists of one dealer (the trusted party) and $n - 1$ participants, denoted by $P_1, P_2, \ldots, P_{n-1}.$  Here, the secret  is an element $s \in \mathbb{F}_q.$  To create shares, the dealer selects a random vector $y \in \mathbb{F}_q^k$ satisfying $s = y \cdot h_0^T$ and computes the codeword $v  = y \mathcal{H}= (s, v_1, \ldots, v_{n-1})\in \mathcal{C}.$ The dealer then distributes the share $v_i$ to the participant $P_i$ for $1 \leq i \leq n - 1.$

A subset of participants $\{P_{i_1}, P_{i_2}, \ldots, P_{i_t}\}$ can recover the secret $s$ by combining their shares  if and only if  the column $h_0^T$  is an $\mathbb{F}_q$-linear combination of the columns $ h_{i_1}^T, h_{i_2}^T, \ldots, h_{i_t}^T,$ where $1 \leq i_1 < i_2 < \cdots < i_t \leq n-1$ \cite[Lem. 2]{Renvall1996}. Such a subset is referred to as an access set. Further, an access set is called minimal if the participants in the set can collectively reconstruct the secret, but no proper subset of them can do so. The collection of all such minimal access sets constitutes the minimal access structure of the scheme. Furthermore, a participant who belongs to every minimal access set is referred to as a dictatorial participant.

Ding and Yuan \cite[Prop. 2]{Ding2003} studied minimal access structures of Massey’s secret sharing schemes based on the linear codes whose duals are minimal. For the sake of completeness, we state Proposition 2 of Ding and Yuan \cite{Ding2003} below.
\begin{lemma}\label{MLem}\cite[Prop. 2]{Ding2003}
Let $\mathcal{C}$ be a minimal code of length $n$ and dimension $k$ over $\mathbb{F}_q$ with a generator matrix
$\mathcal{H} = \big[ h_0^T ~~ h_1^T ~~ \cdots ~~ h_{n-1}^T \big],$ where $h_i\in \mathbb{F}_q^k \setminus \{\mathbf{0}\}$ for $0 \leq i \leq n-1.$  
In Massey's secret sharing scheme based on the dual  $\mathcal{C}^\perp$ with participants $P_1, P_2, \ldots, P_{n-1},$ there are precisely $q^{k - 1}$ minimal access sets.
Moreover, for $1 \leq i \leq n - 1,$ if the vector $h_i$ is a scalar multiple of $h_0,$ then the corresponding participant $P_i$ belongs to every minimal access set and is therefore dictatorial. Otherwise, the participant $P_i$ belongs to exactly
$(q - 1)q^{k - 2}$ out of the $q^{k - 1}$ minimal access sets.
\end{lemma}

 In the following theorem, we examine the minimal access structure of the Massey’s secret sharing scheme based on the dual  $\Phi(\mathscr{C}_{\mathcal{S}_3})^\perp$ over $\mathbb{F}_q$ 
 and determine the number of dictatorial participants.
\begin{thm}\label{TA3} 
Let  $\mathscr{G}_{\mathcal{S}_3}= \big[ x_0^T ~~ x_1^T ~~ \cdots ~~ x_{2|\mathcal{S}_3|-1}^T \big]$   be a spanning matrix of the code $\Phi(\mathscr{C}_{\mathcal{S}_3}),$ where the vectors $x_0 ,x_1 , \ldots ,x_{2|\mathcal{S}_3|-1}$ form the multiset $\mathcal{N}_3$ (as defined by \eqref{D3}).
In Massey’s secret sharing scheme based on the dual  $\Phi(\mathscr{C}_{\mathcal{S}_3})^\perp$ over $\mathbb{F}_q,$ there are precisely
$2q^{|\mathpzc{A}| + |\mathpzc{C}|}(q^m - q^{|\mathpzc{B}|}) - 1$ participants 
and 
$q^{m + |\mathpzc{A}| + |\mathpzc{C}| - 1}$ minimal access sets.  In addition, the following hold.
\begin{itemize}
\item[(a)] Suppose that $x_0 = (w_2,w_3,w_1)$
for some $w_1 \in \Delta_{\mathpzc{A}},$\ $w_2 \in \Delta_{\mathpzc{B}}^c,$ and $w_3 \in \Delta_{\mathpzc{C}}.$ If $w_2-w_1\in \Delta_{\mathpzc{B}}^c,$ then there are exactly 
$2q-3$ dictatorial participants; otherwise, there are exactly 
$q-2$ dictatorial participants. Each of the remaining participants belongs to exactly
$(q - 1)q^{m + |\mathpzc{A}| + |\mathpzc{C}| - 2}$
of the minimal access sets.
\item[(b)] Suppose that $x_0 = (w_2+w_1,w_3,w_1)$
for some $w_1 \in \Delta_{\mathpzc{A}}, $ $w_2 \in \Delta_{\mathpzc{B}}^c,$ and $w_3 \in \Delta_{\mathpzc{C}}.$ If $w_1+w_2\in \Delta_{\mathpzc{B}}^c,$ then there are exactly 
$2q-3$ dictatorial participants; otherwise, there are exactly 
$q-2$ dictatorial participants. Each of the remaining participants belongs to exactly
$(q - 1)q^{m + |\mathpzc{A}| + |\mathpzc{C}| - 2}$
of the minimal access sets. 
\end{itemize}
Moreover, any group of $2(q - 1)q^{|\mathpzc{A}| + |\mathpzc{C}| - 1}(q^m - q^{|\mathpzc{B}|}) - 2$ or fewer participants gathers no information about the secret.
\end{thm}
\begin{proof}  To prove the result, we see, by Theorem \ref{ThN3}, that the Gray image $\Phi(\mathscr{C}_{\mathcal{S}_3})$ is a minimal code of length $2(q^m - q^{|\mathpzc{B}|})q^{|\mathpzc{A}| + |\mathpzc{C}|}$ and dimension $m + |\mathpzc{A}| + |\mathpzc{C}|$ over $\mathbb{F}_q.$ In Massey's secret sharing scheme based on the dual  $\Phi(\mathscr{C}_{\mathcal{S}_3})^\perp,$ we see, by Lemma \ref{MLem},  that there are exactly $2(q^m - q^{|\mathpzc{B}|})q^{|\mathpzc{A}| + |\mathpzc{C}|} - 1$ participants and  $q^{m + |\mathpzc{A}| + |\mathpzc{C}|- 1}$ minimal access sets.  Further, using \eqref{D3}, it follows that there are $2m - |\mathpzc{A}| - |\mathpzc{C}|$ rows in the matrix $\mathscr{G}_{\mathcal{N}_3}$ that are zero. Let $\mathscr{H}\in M_{(m + |\mathpzc{A}| + |\mathpzc{C}|) \times 2|\mathcal{S}_3|}(\mathbb{F}_q)$ be the submatrix of $\mathscr{G}_{\mathcal{N}_3}$ obtained by deleting these $2m - |\mathpzc{A}| - |\mathpzc{C}|$ zero rows. Since the dimension of the code $\Phi(\mathscr{C}_{\mathcal{S}_3})$ is $m + |\mathpzc{A}| + |\mathpzc{C}|,$ we note that the matrix $\mathscr{H}$ is a generator matrix of the code $\Phi(\mathscr{C}_{\mathcal{S}_3}).$ Moreover,  the linear dependence relations among the columns of  $\mathscr{G}_{\mathcal{N}_3}$ and $\mathscr{H}$ are identical.

Next, to count the dictatorial participants, we observe, by \eqref{D3}, that the vector $x_0$ is of the form
$(w_2+\omega,w_3,w_1)$
for some $w_1 \in \Delta_{\mathpzc{A}},\ w_2 \in \Delta_{\mathpzc{B}}^c,\ w_3 \in \Delta_{\mathpzc{C}}, \text{ and } \omega \in \{\mathbf{0}, w_1\}.$ We will now distinguish the following two cases: $\mathbf{(i)}$ $\omega = \mathbf{0},$ and $\mathbf{(ii)}$ $\omega = w_1.$  
\begin{description}
\item[(i)] Suppose that $\omega = \mathbf{0}.$ Here, we see, by \eqref{D3}, that in the multiset $\mathcal{N}_3,$ the vector $x_0$ has multiplicity two  if $ w_2-w_1 \in \Delta_\mathpzc{B}^c,$ while it has multiplicity one  if $ w_2-w_1 \not\in \Delta_\mathpzc{B}^c.$ Furthermore, by \eqref{D3} again, we see that for each $\alpha \in \mathbb{F}_q^\ast,$ the vector $\alpha x_0$ also belongs to $\mathcal{N}_3,$ and that in the multiset $\mathcal{N}_3$, each such scalar multiple appears with multiplicity two  if $w_2 - w_1 \in \Delta_\mathpzc{B}^c$ and with multiplicity one, otherwise. Therefore, there are exactly $2(q-1)$ columns in $\mathscr{G}_{\mathcal{N}_3}$ that are scalar multiples of $x_0^T$ (including $x_0^T$ itself) if $w_2 - w_1 \in \Delta_\mathpzc{B}^c;$ otherwise, there are $q-1$ columns in $\mathscr{G}_{\mathcal{N}_3}$ that are scalar multiples of $x_0^T$ (including $x_0^T$ itself). Combining this with the observation that the linear dependence relations among the columns of both $\mathscr{G}_{\mathcal{N}_3}$ and $\mathscr{H}$ are identical, and by applying Lemma \ref{MLem}, we conclude that if $w_2-w_1\in \Delta_{\mathpzc{B}}^c,$ then there are exactly 
$2q-3$ dictatorial participants; otherwise, there are exactly 
$q-2$ dictatorial participants. Moreover, each of the remaining participants belongs to exactly
$(q - 1)q^{m + |\mathpzc{A}| + |\mathpzc{C}| - 2}$
of the minimal access sets.
\item[(ii)] When $\omega = w_1,$ working as in case (i), the desired result follows.
\end{description}

Furthermore, by Theorem \ref{ThN3}, we see that the code $\Phi(\mathscr{C}_{\mathcal{S}_3})$ has Hamming distance $2(q - 1)(q^m - q^{|\mathpzc{B}|})q^{|\mathpzc{A}| + |\mathpzc{C}| - 1}.$ From this and by Theorem 9 of Renvall and Ding \cite{Renvall1996}, it follows that any group of at most $2(q - 1)(q^m - q^{|\mathpzc{B}|})q^{|\mathpzc{A}| + |\mathpzc{C}| - 1} - 2$ participants gathers no information about the secret.
\end{proof}
\begin{remark}
In Massey's secret sharing schemes based on the dual $\Phi(\mathscr{C}_{\mathcal{S}_3})^\perp$ of the minimal code over $\mathbb{F}_q,$ we see, by Theorem \ref{TA3}, that  when $q\geq 3,$ there always exists at least one dictatorial participant. As remarked by Yuan and Ding \cite[p.\,212]{Yuan2005}, such schemes are particularly relevant in scenarios where it is necessary for certain participants to be involved in every decision-making process.
\end{remark}

\subsection{Construction of locally repairable codes with localities either $2$ or $3$}\label{Sec8}
In a distributed storage system, information is stored over a network of storage nodes, where failures are inevitable. To address this challenge, Gopalan \etal \cite{Gopalan2012} introduced  locally repairable codes, in which each coordinate of a codeword is stored on a separate node and can be recovered from a small subset of other coordinates using the code’s locality structure.

In this section, we will study the locality properties of the projective  codes $\mathcal{C}_{\overline{\mathcal{N}}_2}$ and $\mathcal{C}_{\overline{\mathcal{N}}_4}$ over $\mathbb{F}_q,$ constructed in Theorems \ref{Th7} and \ref{Th8}.  To do this, we first recall the definition of linear locally repairable codes  \cite[Def. 4.1]{Luo2022}. A linear code over $\mathbb{F}_q$ with a spanning matrix $G$ is said to have locality $r$ if $r$ is the least positive integer for which every column of $G$ can be expressed as an $\mathbb{F}_q$-linear combination of at most $r$ other columns of $G.$ A linear $[\mathrm{n},\mathrm{k},\mathrm{d}]$-code over $\mathbb{F}_q$ with locality $r$ is referred to as a $q$-ary $[\mathrm{n},\mathrm{k},\mathrm{d}]$ locally repairable code (LRC) with locality $r.$ We next state the well-known Cadambe-Mazumdar bound for linear LRCs.


\begin{lemma}\cite[Th. 1]{Cadambe2015}\label{Lemma5.5}
For a $q$-ary $[\mathrm{n},\mathrm{k},\mathrm{d}]$ LRC with  locality $r,$  we have
\begin{equation}\label{Eq5.10}
\mathrm{k} \leq \min\limits_{1 \leq i \leq \ceil{\frac{\mathrm{k}}{r}} -1} \Big\{ri +  \mathrm{k}^q_{opt}(\mathrm{n} - i(r+1), \mathrm{d}) \Big\},
\end{equation}
where $\mathrm{k}^q_{opt}(\mathrm{n}, \mathrm{d})$ denotes the largest possible dimension that can be achieved by a linear code of length $\mathrm{n}$ and Hamming distance $\mathrm{d}$ over $\mathbb{F}_q.$ 
\end{lemma}
The code $\mathcal{C}$ is said to be alphabet-optimal if it attains the bound \eqref{Eq5.10}.  To analyze the locality properties of the projective  codes $\mathcal{C}_{\overline{\mathcal{N}}_2}$ and $\mathcal{C}_{\overline{\mathcal{N}}_4}$ over $\mathbb{F}_q,$  we need the following lemma.

\begin{lemma}\label{Lem5.6}
Let $\mathpzc{P} $ be a  non-empty proper subset of $ [\mathrm{n}],$ and let  $\Delta_{\mathpzc{P}}$ be the simplicial complex of $\mathbb{F}_q^{\mathrm{n}}$ with support $\mathpzc{P}.$   If either $q \geq 3,$ or  $q = 2$ and $|\mathpzc{P}^c| \geq 2,$ then for each $v \in \Delta_{\mathpzc{P}}^c,$ there exist linearly
independent vectors $e, f \in \Delta_{\mathpzc{P}}^c$ such that $v = e + f.$
\end{lemma} 
\begin{proof}To prove the result, let $v \in \Delta_{\mathpzc{P}}^c$ be fixed. Here, we will consider the following two cases separately: (i) $q \geq 3,$ and (ii) $q = 2$ and $|\mathpzc{P}^c| \geq 2.$
\begin{itemize}\item[(i)] Let $q\geq 3.$ Here, we assume, without any loss of generality, that $1 \notin \mathpzc{P}$ and  $(v)_1 \neq 0.$ As $q \geq 3,$ there exists $\alpha \in \mathbb{F}_q^\ast$ such that $(v)_1 + \alpha \neq 0.$ Further, since $\mathpzc{P} \neq \emptyset,$ there exists  $\kappa \in \mathpzc{P}$ satisfying $2 \leq \kappa \leq \mathrm{n}.$ Now, let us define a vector $e\in \mathbb{F}_q^{\mathrm{n}}$ as $(e)_1 = (v)_1+\alpha,$ $(e)_i = (v)_i$ for all $i \in [\mathrm{n}] \setminus \{1,\kappa \},$ and 
\begin{equation*}
(e)_\kappa=\left\{\begin{array}{cl}  1 & \text{if } (v)_\kappa = 0;\\
0 & \text{otherwise.}
\end{array}\right.
\end{equation*}
One can easily see that both vectors $e$ and $v-e$ lie in $\Delta_{\mathpzc{P}}^c$ and satisfy the desired properties.

\item[(ii)]Next, let us assume that $q=2$ and $|\mathpzc{P}^c| \geq 2.$ Here, we assume, without any loss of generality, that  both $1,2 \in \mathpzc{P}^c.$ Let us define a vector $e\in \mathbb{F}_q^{\mathrm{n}}$ as  $(e)_i = (v)_i$ for all $i \in [\mathrm{n}] \setminus \{1,2 \},$ and 
\begin{equation*}
(e)_{ \{1,2 \} }=\left\{\begin{array}{cl}  (0,1) & \text{if } (v)_{ \{1,2\} } = (1,0);\\
(1,0) & \text{otherwise.}
\end{array}\right.
\end{equation*}
One can easily see that both vectors $e$ and $v-e$ belong to $\Delta_{\mathpzc{P}}^c$ and satisfy the desired properties.\end{itemize} 
\vspace{-5mm}\end{proof}

In the following theorem, we show that the projective code $\mathcal{C}_{\overline{\mathcal{N}}_2}$ is a $q$-ary LRC with locality $2$ when either $q \geq 3,$ or $q = 2$ and $|\mathpzc{A}^c| \geq 2,$ and  a binary LRC with locality $3$ when $q = 2$ and $|\mathpzc{A}^c| = 1.$

\begin{thm}\label{Thm5.12} Let us define $\rho=0$ if $\mathpzc{B}= \mathpzc{A},$ while $\rho=1$ if $\mathpzc{B}\subsetneq \mathpzc{A}.$
\begin{itemize}
\item[(a)] When either $q \geq 3,$ or  $q = 2$ and $|\mathpzc{A}^c| \geq 2,$ the projective code $\mathcal{C}_{\overline{\mathcal{N}}_2}$ is a $q$-ary $\left[\frac{2q^{|\mathpzc{B}|+|\mathpzc{C}|}(q^m-q^{|\mathpzc{A}|})}{q-1},~ 2m+|\mathpzc{C}|,~\right. \\ \left.(q^m-\rho q^{|\mathpzc{A}|})q^{|\mathpzc{B}|+|\mathpzc{C}|-1}\right]$ LRC with locality $2.$  
\item[(b)] When $q = 2$ and $|\mathpzc{A}^c| = 1,$ the projective  code $\mathcal{C}_{\overline{\mathcal{N}}_2}$ is a binary $\left[2^{m+|\mathpzc{B}|+|\mathpzc{C}|},~  2m+|\mathpzc{C}|,~ (2-\rho)2^{m+|\mathpzc{B}|+|\mathpzc{C}|-2}\right]$ LRC with locality $3.$ 
\end{itemize}
\end{thm}
\begin{proof}
To prove the result, we first recall that the projective code $\mathcal{C}_{\overline{\mathcal{N}}_2}$ is a linear code over $\mathbb{F}_q$ with a spanning matrix $G_{\overline{\mathcal{N}}_2}$ whose transposed columns form the set $\overline{\mathcal{N}}_2.$

\begin{itemize}
\item[(a)] Let us assume that either $q\geq 3,$ or $q=2$ and $|\mathpzc{A}^c| \geq 2.$  Now, to prove that the code $\mathcal{C}_{\overline{\mathcal{N}}_2}$ has locality $2,$ we will show that each column of $G_{\overline{\mathcal{N}}_2}$ can be expressed as a sum of its two other columns.  For this, we note, based on the description  of the set $\overline{\mathcal{N}}_2$ given in Lemma \ref{Lem8}(a), that an arbitrary vector of  $\overline{\mathcal{N}}_2$ is of the form $y_\kappa=(w_2+\omega,w_3,w_1) ,$ where $w_2 \in \Delta_\mathpzc{B},$ $w_3 \in \Delta_\mathpzc{C},$  $w_1 \in \Delta_\mathpzc{A}^c,$ and $\omega \in \{ \mathbf{0},w_1\}.$ Let us define $\xi \in \mathbb{F}_q$ as 
\begin{equation}\label{Eq6.4}
\xi =\left\{\begin{array}{cl}  0 & \text{if } \omega = \mathbf{0};\\
1 & \text{otherwise.}
\end{array}\right.
\end{equation}
We further see, by Lemma \ref{Lem5.6}, that there exist linearly independent vectors $u_1,v_1 \in \Delta_\mathpzc{A}^c$ such that $w_1 = u_1 + v_1.$  This implies that $y_\kappa = y_i +y_j,$ where  $y_i=(w_2 + \xi u_1,w_3,u_1)$ and $y_j=(\xi v_1,\mathbf{0},v_1).$ One can easily see that  $y_i$ and $y_j$  are linearly independent over $\mathbb{F}_q.$ From the description  of the set $\overline{\mathcal{N}}_2$ given in Lemma \ref{Lem8}(a), we note that both  $y_i$ and $y_j$ belong to $\overline{\mathcal{N}}_2$ up to scalar multiples. This shows that the code $\mathcal{C}_{\overline{\mathcal{N}}_2}$ has locality $2.$ From this and by applying Theorem \ref{Th7}, we get the desired result.

\item[(b)] Let $q=2$ and $|\mathpzc{A}^c| = 1.$ Here, we have $\overline{\mathcal{N}}_2 = \mathcal{N}_2.$ Now, to show that the code $\mathcal{C}_{\mathcal{N}_2}$ has locality $3,$ we will show that every column of $G_{\mathcal{N}_2}$ can be expressed as a sum of its three other columns. For this, we see, by \eqref{D2},  that an arbitrary vector of ${\mathcal{N}_2}$ is of the form $y_\kappa=(w_2+\omega,w_3,w_1),$ where $w_2 \in \Delta_\mathpzc{B},$ $w_3 \in \Delta_\mathpzc{C},$  $w_1 \in \Delta_\mathpzc{A}^c,$ and $\omega \in \{ \mathbf{0},w_1\}.$ We further observe that there exist  vectors $y_i = (w_2 + \xi u_1,u_3,u_1)$ and $y_j=(w_2 + \xi v_1,v_3,v_1)$
satisfying 
$ \big((u_3)_\mathpzc{C},(u_1)_\mathpzc{A} \big) \neq \big((w_3)_\mathpzc{C},(w_1)_\mathpzc{A}\big),$ $\big((v_3)_\mathpzc{C},(v_1)_\mathpzc{A} \big)\neq \big((w_3)_\mathpzc{C},(w_1)_\mathpzc{A}\big),$ and $ \big((v_3)_\mathpzc{C},(v_1)_\mathpzc{A} \big) \neq \big((u_3)_\mathpzc{C},(u_1)_\mathpzc{A} \big),$
 where $\xi$ is as defined by \eqref{Eq6.4}.
One can easily see that the  vectors $y_i, y_j,$ and $y_i+y_j+y_\kappa$ belong to $\mathcal{N}_2$ and that the vector $y_\kappa$ is equal to the sum of the vectors $y_i,y_j,$ and $y_i+y_j+y_\kappa.$  This shows that the code $\mathcal{C}_{\mathcal{N}_2}$ has locality $3.$ Now, by applying Theorem \ref{Th7}, we get the desired result.
\vspace{-2mm}\end{itemize}
\vspace{-2mm}\end{proof}
From the above theorem, we make the following observation.
\begin{remark}
In particular, when $q=2,$ $\mathpzc{A}=\mathpzc{B},$ and $|\mathpzc{A}|=m-1,$  one can easily see that $\mathcal{C}_{\overline{\mathcal{N}}_2}=\Phi(\mathscr{C}_{\mathcal{S}_2}).$ This, by Remark \ref{comp2.2}, implies that $\mathcal{C}_{\overline{\mathcal{N}}_2}$ is a first order Reed-Muller code. Furthermore, we see, by   Theorem \ref{Thm5.12}(b), that the 
 projective code $\mathcal{C}_{\overline{\mathcal{N}}_2}$ is a binary $\left[2^{2m+|\mathpzc{C}|-1}, 2m+|\mathpzc{C}|,2^{2m+|\mathpzc{C}|-2}\right]$ LRC with locality $3,$ which agrees with Lemma 14 of Huang \etal \cite{Huang2016}.  Moreover, we see, by Lemma \ref{Lemma5.5} and using the Griesmer bound \eqref{GB}, that the code $\mathcal{C}_{\overline{\mathcal{N}}_2}$ is an alphabet-optimal LRC.

 \end{remark}
 

In the following theorem, we show that the projective linear code $\mathcal{C}_{\overline{\mathcal{N}}_4}$ is a $q$-ary LRC with locality $2$ when either $q \geq 3,$ or  $q = 2$ and $|\mathpzc{C}^c| \geq 2,$ and a binary LRC with locality $3$ when $q = 2$ and $|\mathpzc{C}^c| = 1.$

\begin{thm}\label{Thm5.13}
\begin{itemize}
\item[(a)] When either $q \geq 3,$ or  $q = 2$ and $|\mathpzc{C}^c| \geq 2,$ the projective code $\mathcal{C}_{\overline{\mathcal{N}}_4}$ is a $q$-ary $\left[\frac{2(q^m-q^{|\mathpzc{C}|})(q^{|\mathpzc{A}|}-1)q^{|\mathpzc{B}|}}{q-1},~ \right. \\ \left.{m+2|\mathpzc{A}|+|\mathpzc{B}|},~(q^{m}-q^{|\mathpzc{C}|})q^{|\mathpzc{A}|+|\mathpzc{B}|-1}\right]$ LRC with locality $2.$
\item[(b)] When $q = 2$ and $|\mathpzc{C}^c| = 1,$ the projective  code $\mathcal{C}_{\overline{\mathcal{N}}_4}$ is a binary $\left[(2^{|\mathpzc{A}|}-1)2^{m+|\mathpzc{B}|},~  m+2|\mathpzc{A}|+|\mathpzc{B}|,~  2^{m+|\mathpzc{A}|+|\mathpzc{B}|-2}\right]$ LRC with locality $3.$ 
\end{itemize}
\end{thm}
\begin{proof}
Working as in Theorem \ref{Thm5.12} and by applying Lemmas   \ref{Lem8}(b) and \ref{Lem5.6} and Theorem \ref{Th8} and using equation \eqref{D4}, we get the desired result.
\end{proof}

\section{Conclusion and future work}\label{Conclusion}

In this paper, four infinite families of linear codes over the ring $\frac{\mathbb{F}_q[u]}{\langle u^2 \rangle}$ are constructed using defining sets formed from certain non-empty subsets of $\mathcal{R}^m$ associated with three simplicial complexes of $\mathbb{F}_q^m,$ each having a single maximal element, where $\mathcal{R} = \frac{\mathbb{F}_q[u]}{\langle u^2 \rangle} \times \mathbb{F}_q$ is a mixed-alphabet ring. The parameters and Lee weight distributions of these codes are explicitly determined. Through their Gray images, several infinite families of few-weight codes, binary and ternary-self-orthogonal codes as well as an infinite family of minimal, near-Griesmer, and  distance-optimal codes over $\mathbb{F}_q$  are obtained. 

Spanning matrices of a linear code over $\frac{\mathbb{F}_q[u]}{\langle u^2 \rangle}$ with defining set $\mathcal{D} \subseteq \mathcal{R}^m$ and of its Gray image are also determined.  Based on this result, two infinite families of projective few-weight codes over $\mathbb{F}_q$ with new parameters are constructed. By examining the duals of these projective codes, two infinite families of binary distance-optimal projective codes are obtained. Furthermore, an infinite family of quaternary projective 3-weight codes is constructed, in which the non-zero Hamming weights sum to $\frac{9}{4}$ times the code length. As an application of our newly constructed minimal codes over $\mathbb{F}_q,$ the minimal access structures of Massey’s secret sharing schemes based on their duals are investigated, along with the determination of the number of dictatorial participants. Finally, the locality properties of the constructed families of projective codes are investigated, and their localities are explicitly determined. As a consequence, four infinite families of $q$-ary locally repairable codes (LRCs) with locality $2,$ and two infinite families of binary LRCs with locality $3,$ are obtained.

In a subsequent work \cite{Jose2026}, we determine the parameters of linear codes over $\frac{\mathbb{F}_q[u]}{\langle u^2 \rangle}$ whose defining sets are subsets of $\mathcal{R}^m$ associated with certain special simplicial complexes of $\mathbb{F}_q^m$ having multiple maximal elements. Future work could extend these results to arbitrary simplicial complexes of $\mathbb{F}_q^m$ with multiple maximal elements and the corresponding defining sets in $\mathcal{R}^m$. Another direction is to study codes over the ring $\frac{\mathbb{F}_q[u]}{\langle u^e \rangle}$, $e\geq 3$, using defining sets constructed from mixed alphabets of chain rings and simplicial complexes, with the goal of identifying new classes of codes through their Gray images.

More generally, the mixed-alphabet defining-set construction of linear codes is not intrinsically restricted to $\frac{\mathbb{F}_q[u]}{\langle u^2\rangle}\times\mathbb{F}_q$. It can, in principle, be extended to a broader class of mixed-alphabet rings for which a compatible inner product and a suitable Gray map are available. For example, one may consider $R\times R'$, where $R$ and $R'$ are finite commutative chain rings and $R'$ can be embedded into $R$. In this setting, an appropriate inner product has been defined  (see Borges \etal \cite{Borges2018}), and Jitman's Gray map \cite{Jitman2010}, together with the homogeneous metric, provides an isometry from a homogeneous space over finite commutative chain rings to a Hamming space over finite fields. Thus, the defining-set construction  can, in principle, be formulated analogously. The main challenge, however, would be to determine explicit code parameters and weight distributions, as these depend substantially on the algebraic structures of the underlying rings.

\section{Acknowledgements}
The first author gratefully acknowledges the research support by the National Board for Higher Mathematics (NBHM), India,  under Grant no. 0203/13(46)/2021-R\&D-II/13176. 

\section{Appendix: Coset graphs and strongly walk-regular graphs}\label{Sec6}

In this appendix, we will first provide an elementary proof of the fact  that the coset graph of a linear code over $\mathbb{F}_q$ is connected. We will further construct strongly $\ell$-walk-regular graphs ($\ell$-SWRGs) for all odd integers $\ell \geq 3,$ using the quaternary projective $3$-weight codes constructed in  Corollary \ref{SHI}. To this end, we begin by recalling some basic definitions and results from graph theory.

A graph $\mathbb{G}$ is defined as an ordered pair $(V,E),$ where $V$ is the set of vertices and $E$ is the set of edges, with each edge being represented as an unordered pair of distinct elements from $V.$ Two vertices are said to be adjacent (or neighbours) if they are connected by an edge.  For a given ordering of the vertex set $V,$ $\ie$ $V=\{t_1,t_2, \ldots, t_{|V|} \},$ the adjacency matrix of the graph $\mathbb{G}$ is defined as a $|V| \times |V|$ matrix $\mathbb{G}_M$ whose $(i,j)$-th entry is $1$ if the vertices $t_i$ and $t_j$ are adjacent, and $0$ otherwise. Eigenvalues of the matrix $\mathbb{G}_M$ are known as the eigenvalues of the graph $\mathbb{G}.$ Furthermore, a spectrum of the graph $\mathbb{G}$ is defined as a multiset consisting of all its eigenvalues, each listed with its respective multiplicity. A walk in a graph $\mathbb{G}$ is defined as a finite sequence  $\{t_{i_1}, t_{i_2}, \ldots, t_{i_\ell}\}$ of vertices such that each  pair $(t_{i_j},t_{i_{j+1}})$ of consecutive vertices is an edge in the graph for  $1 \leq j < \ell.$ The walk $\{t_{i_1}, t_{i_2}, \ldots, t_{i_\ell}\}$ is called a walk of length $\ell-1$ between the vertices $t_{i_1}$ and $t_{i_\ell}.$ A graph is said to be connected if there exists a walk between each pair of distinct vertices.

Next, for a positive integer $r,$ a graph is said to be $r$-regular if every vertex has $r$ neighbours in the graph. 
A strongly regular graph (SRG) is a regular graph with the additional property that the number of common neighbors between any two distinct vertices depends only on whether the vertices are adjacent or not, or equivalently, the number of walks of length $2$ between any pair of distinct vertices depends solely on whether those two vertices are adjacent or not. The concept of strongly $\ell$-walk-regular graphs ($\ell$-SWRGs) is introduced by Dam and Omidi \cite{Dam2013} as a generalization of SRGs, where the condition on walks of length $2$ is extended to walks of length $\ell \geq 2.$  For an integer $\ell \geq 2,$ a graph $\mathbb{G}$ is said to be an $\ell$-SWRG with parameters $(\lambda_\ell,\mu_\ell,\nu_\ell)$ if the number of walks of length $\ell$ between any two vertices of $\mathbb{G}$ is (i) $\lambda_\ell$  if the vertices are adjacent, (ii) $\mu_\ell$ if the vertices are non-adjacent, and (iii) $\nu_\ell$ if the vertices are identical.

Let $\mathcal{C}$ be a linear code of length $\mathrm{n}$ over $\mathbb{F}_q.$ Elements of the quotient space $\mathbb{F}_q^{\mathrm{n}}/\mathcal{C}$ are called the cosets of $\mathcal{C}$ in $\mathbb{F}_q^{\mathrm{n}}.$  Clearly, there are precisely $\frac{q^\mathrm{n}}{|\mathcal{C}|}$ distinct cosets of $\mathcal{C}$ in $\mathbb{F}_q^{\mathrm{n}}.$ A coset leader of a coset of $\mathcal{C}$ is defined as a vector of the smallest Hamming weight in the coset.  Note that a coset leader of a coset of $\mathcal{C}$ need not be unique. The Hamming weight of a coset of $\mathcal{C}$ is defined as the Hamming weight of its coset leader. Now, the coset graph of the linear code $\mathcal{C},$ denoted by $\Gamma_\mathcal{C},$ is defined as a graph whose vertices are the cosets of $\mathcal{C}$ in $\mathbb{F}_q^{\mathrm{n}},$ where  any two cosets of $\mathcal{C}$ are adjacent if and only if they differ by a coset of $\mathcal{C}$ with Hamming weight $1.$  

Shi \etal \cite[pp. 4–5]{Shi2022a} observed that the coset graph $\Gamma_\mathcal{C}$ of $\mathcal{C}$ is isomorphic to its syndrome graph, and that the syndrome graph of $\mathcal{C}$ is a Cayley graph in which the columns of the parity-check matrix generate the full space. From this, they deduced that the coset graph $\Gamma_\mathcal{C}$ is connected. In the following lemma, we present an elementary proof of this result. Although we could not find this proof in the literature, we make no claim regarding its novelty.

\begin{lemma}\label{Lemma5.3}
For any linear code $\mathcal{C}$ of length $\mathrm{n}$ over $\mathbb{F}_q,$  the coset graph $\Gamma_{\mathcal{C}}$ is connected. \end{lemma}
\begin{proof}
To prove the result, let $ g_0+\mathcal{C}, g_1 + \mathcal{C}, \ldots, g_{N-1} + \mathcal{C}$ be all the distinct cosets of $\mathcal{C}$ in $\mathbb{F}_q^{\mathrm{n}},$ where $N=\frac{q^{\mathrm{n}}}{|\mathcal{C}|}$ and $g_0=\textbf{0},g_1, g_2,\ldots,g_{N-1} \in \mathbb{F}_q^\mathrm{n}.$  We also assume, for $0 \leq i \leq N-1,$ that $g_i$ is a coset leader of the coset $g_i + \mathcal{C},$  $\ie$  we have $wt_H(g_i+x) \geq wt_H(g_i)$ for all $x \in \mathcal{C}.$ Thus, the coset graph $\Gamma_{\mathcal{C}}$ of $\mathcal{C}$ is a graph with vertices $V=\{ g_0+\mathcal{C}, g_1 + \mathcal{C}, \ldots, g_{N-1} + \mathcal{C}\},$ where the vertices $g_i+\mathcal{C}$ and $g_j+\mathcal{C}$ are adjacent if and only if the Hamming weight of the coset $g_i-g_j+\mathcal{C}$ is $1.$

To prove that the graph $\Gamma_{\mathcal{C}}$ is connected, we need to show that there exists a walk between any two distinct vertices of $\Gamma_\mathcal{C}.$ Towards this, we first note,  for $1\leq i \leq N-1,$ that the vertices $g_0+\mathcal{C}$ and $g_i + \mathcal{C}$ are adjacent in $\Gamma_{\mathcal{C}}$ if and only if $wt_H(g_i) =1.$ 
Further, if $g_i + \mathcal{C}$ and $g_j + \mathcal{C}$ are two distinct vertices of $\Gamma_\mathcal{C}$ with $wt_H(g_i)=wt_H(g_j)=1,$ then the sequence $\{g_i+\mathcal{C},g_0+\mathcal{C},g_j+\mathcal{C}\}$ is a walk between  the vertices $g_i + \mathcal{C}$ and $g_j + \mathcal{C}.$

Next, let $g_i + \mathcal{C}$  be a vertex of $\Gamma_\mathcal{C}$ with $wt_H(g_i) =s\geq 2.$  Here, we  assert that there exists another vertex $g_i^{(1)} + \mathcal{C} \in V$ adjacent to $g_i + \mathcal{C}$ with $wt_H(g_i^{(1)}) = s - 1.$ To prove this assertion, let us suppose that  $\supp(g_i)=\{t_1,t_2,\ldots,t_s\},$ and let us define a vector $g_i^{(1)}\in \mathbb{F}_q^\mathrm{n}$ with the $\kappa$-th coordinate 
\begin{equation*}
(g_i^{(1)})_\kappa=\left\{\begin{array}{cl}  (g_i)_\kappa & \text{if } \kappa \in \{t_2,t_3,\ldots,t_s\};\\
0 & \text{otherwise,}
\end{array}\right.
\end{equation*}
 for $1\leq \kappa \leq \mathrm{n}.$ 
   Note that $wt_H(g_i^{(1)}) = s-1.$ We will first show that $g_i^{(1)} \notin \mathcal{C}.$ For, if $g_i^{(1)} \in \mathcal{C},$ then we have $g_i - g_i^{(1)} \in g_i + \mathcal{C}$ and $wt_{H}(g_i - g_i^{(1)}) = 1,$ which contradicts the fact that $g_i$ is a coset leader of  $g_i + \mathcal{C}$ with $wt_H(g_i)=s \geq 2.$ Thus, we have $g_i^{(1)} \notin \mathcal{C},$ and hence $g_i^{(1)} + \mathcal{C}$ is a non-zero coset of $\mathcal{C}.$ We next claim that $g_i^{(1)}$ is a coset leader of  $g_i^{(1)} + \mathcal{C}.$  For this, it is suffices to show that $wt_H(g_i^{(1)}+x) \geq wt_H(g_i^{(1)})=s-1$ for all non-zero $x \in \mathcal{C}.$ To do this, let $x=(x_1,x_2,\ldots,x_\mathrm{n}) \in \mathcal{C}$ be fixed.  For $1 \leq \kappa \leq n,$ the 
   $\kappa$-th coordinates of  $g_i+x$ and $g_i^{(1)}+x$ are given by 
   \begin{equation*}
    (g_i+x)_\kappa=\left\{\begin{array}{cl}  (g_i)_\kappa +x_\kappa& \text{if } \kappa \in \{t_1,t_2,t_3,\ldots,t_s\};\\
x_\kappa & \text{otherwise,}
\end{array}\right. \text{ ~and ~}(g_i^{(1)}+x)_\kappa=\left\{\begin{array}{cl}  (g_i)_\kappa +x_\kappa& \text{if } \kappa \in \{t_2,t_3,\ldots,t_s\};\\
x_\kappa & \text{otherwise,}
\end{array}\right.
\end{equation*}
respectively.  As $g_i$ is a coset leader of $g_i+\mathcal{C}$ with $wt_H(g_i)=s,$ we have $wt_H(g_i+x)\geq wt_H(g_i)=s.$ 

Now, if $x_{t_1}=0,$ then we have  $wt_H(g_i^{(1)}+x)=wt_H(g_i+x)-1\geq s-1.$ On the other hand, if $x_{t_1} \neq 0$ and $(g_{i}+x)_{t_1} =0,$ then we have $wt_H(g_i^{(1)}+x)=wt_H(g_i+x)+1 \geq s+1.$  Finally, when $x_{t_1} \neq 0$ and $(g_{i}+x)_{t_1} \neq 0,$ we have $wt_H(g_i^{(1)}+x)=wt_H(g_i+x) \geq s.$ This shows that $wt_H(g_i^{(1)}+x) \geq s-1=wt_H(g_i^{(1)}),$ and hence  $g_i^{(1)}$ is a coset leader of the coset $g_i^{(1)} + \mathcal{C}.$   We will further show that $g_i-g_i^{(1)} \notin \mathcal{C}.$ For, if $g_i-g_i^{(1)} \in \mathcal{C},$ then we have $g_i^{(1)}=g_i - (g_i-g_i^{(1)}) \in g_i + \mathcal{C},$ which contradicts the fact that  $g_i$ is a coset leader of  $g_i + \mathcal{C}$ with $wt_H(g_i)=s=wt_H(g_i^{(1)})+1.$  Further, using the fact  that $wt_H(g_i - g_i^{(1)})=1,$ we see that the coset $g_i - g_i^{(1)} + \mathcal{C}$ has Hamming weight $1.$ Thus, the vertices $g_i + \mathcal{C}$ and $g_i^{(1)} + \mathcal{C}$ are adjacent in the coset graph $\Gamma_\mathcal{C}$ with $wt_H(g_i)=s \geq 2$ and  $wt_H(g_i^{(1)})=s-1,$ which proves the assertion.

For $s \geq 3$ and  $2 \leq a \leq s-1,$  let us define a word $g_i^{(a)}\in \mathbb{F}_q^{\mathrm{n}}$  with the $\kappa$-th coordinate 
\vspace{-2mm}\begin{equation*}
    (g_i^{(a)})_\kappa=\left\{\begin{array}{cl}  (g_i)_\kappa & \text{if } \kappa \in \{t_{a+1},t_{a+2},\ldots,t_s\};\\
0 & \text{otherwise,}
\end{array}\right.
\vspace{-2mm}\end{equation*} for $1 \leq \kappa \leq \mathrm{n}.$ By repeatedly applying the above assertion, we see, for $2 \leq a \leq s-1,$ that the cosets $g_i^{(a-1)}+\mathcal{C}$ and $g_i^{(a)}+\mathcal{C}$ are adjacent in $\Gamma_{\mathcal{C}},$ where ${wt}_H(g_i^{(a-1)})=s-a+1$ and ${wt}_H(g_i^{(a)})=s-a.$  As ${wt}_H(g_i^{(s-1)})=1,$ the vertices $g_i^{(s-1)}+\mathcal{C}$ and $g_0+\mathcal{C}$ are adjacent $\Gamma_{\mathcal{C}}.$ This shows that the sequence $\{g_i+\mathcal{C}, g_i^{(1)}+\mathcal{C}, g_i^{(2)}+\mathcal{C}, \ldots, g_i^{(s-1)}+\mathcal{C}, g_0+\mathcal{C}\}$ is a  walk between the vertices $g_i+\mathcal{C}$ and $g_0+\mathcal{C}$ in $\Gamma_{\mathcal{C}}$ when ${wt}_H(g_i) =s\geq 2.$ 

Thus, if ${wt}_H(g_i)=s\geq 2$ and ${wt}_H(g_j)=1,$ then the sequence $\{g_i+\mathcal{C},g_i^{(1)}+\mathcal{C}, g_i^{(2)}+\mathcal{C}, \ldots, g_i^{(s-1)}+\mathcal{C} , g_0+\mathcal{C}, g_j+\mathcal{C}\}$ is a walk between the vertices $g_i+\mathcal{C}$ and $g_j+\mathcal{C}$ in $\Gamma_{\mathcal{C}}.$ On the other hand, if  ${wt}_H(g_i)=s\geq 2$ and ${wt}_H(g_j)=r \geq 2,$ then the sequence $\{g_i+\mathcal{C},g_i^{(1)}+\mathcal{C}, g_i^{(2)}+\mathcal{C}, \ldots, g_i^{(s-1)}+\mathcal{C} , g_0+\mathcal{C},g_j^{(r-1)}+\mathcal{C},g_j^{(r-2)}+\mathcal{C}, \ldots, g_j^{(1)}+\mathcal{C}, g_j+\mathcal{C}\}$ is a walk between the vertices $g_i+\mathcal{C}$ and $g_j+\mathcal{C}$ in $\Gamma_{\mathcal{C}}.$ This shows that the coset graph $\Gamma_{\mathcal{C}}$ is connected.
\end{proof}

 Now, the following well-known result provides a construction of $3$-SWRGs from projective $3$-weight codes of length $\mathrm{n}$ over $\mathbb{F}_q$  with non-zero Hamming weights   $\mathrm{w}_1,$ $\mathrm{w}_2,$ and $\mathrm{w}_3$ satisfying $\mathrm{w}_1 + \mathrm{w}_2 + \mathrm{w}_3 = \frac{3 \mathrm{n} (q-1)}{q}.$

\begin{thm}\cite{Shi2019,Brouwer1989}\label{Thm5.7}
Let $\mathcal{C}$ be a projective $3$-weight code of length $\mathrm{n}$ over $\mathbb{F}_q$ with Hamming weights $0=\mathrm{w}_0 < \mathrm{w}_1 < \mathrm{w}_2 < \mathrm{w}_3$ satisfying  $\mathrm{w}_1 + \mathrm{w}_2 + \mathrm{w}_3 = \frac{3\mathrm{n}(q-1)}{q}.$ The coset graph $\Gamma_{\mathcal{C}^\perp}$ is a $3$-SWRG and an $\big(\mathrm{n}(q-1)\big)$-regular graph. Moreover, the coset graph $\Gamma_{\mathcal{C}^\perp}$ has eigenvalues   $\mathrm{n}(q-1) ,$ $ \mathrm{n}(q-1) - q\mathrm{w}_1,$ $ \mathrm{n}(q-1) - q\mathrm{w}_2,$ and $\mathrm{n}(q-1) - q\mathrm{w}_3$ with multiplicities $A_{\mathrm{w}_0}=1,$ $A_{\mathrm{w}_1},$ $A_{\mathrm{w}_2},$ and $A_{\mathrm{w}_3},$ respectively.
\end{thm}
\begin{proof}
The result follows from Theorems 2 and  5 of Shi and Sol{\'e} \cite{Shi2019}, Theorem 11.1.11 of \cite{Brouwer1989} and Theorem 2.5 of Kiermaier \etal \cite{Kiermaier2023}. 
\end{proof}

In the following theorem, we apply Corollary \ref{SHI}, Lemma \ref{Lemma5.3}, and Theorem \ref{Thm5.7} to construct an infinite family of $\ell$-SWRGs for all odd integers $\ell \geq 3.$ 
\begin{thm}\label{Thm5.8} 
    Let $q=4,$ $\mathpzc{B}=\mathpzc{A},$ and $|\mathpzc{A}|=m-1,$ and let us define $\theta = 2^{4m+2|\mathpzc{C}|-3}.$ The following hold.
\begin{itemize}
\item[(a)] The coset graph $\Gamma_{\mathcal{C}_{\overline{\mathcal{N}}_2}^\perp}$ is a $\left(3\theta\right)$-regular graph and has eigenvalues  $3\theta,$ $\theta,$ $0,$ and $-\theta$ with respective multiplicities $1,$ $6,$ $4^{2m+|\mathpzc{C}|} - 16,$ and $9.$
\item[(b)] For all odd integers $\ell \geq 3,$ the coset graph $\Gamma_{\mathcal{C}_{\overline{\mathcal{N}}_2}^\perp}$ is an $\ell$-SWRG with parameters $(\lambda_\ell,\mu_\ell,\nu_\ell),$ where  $\lambda_\ell = \frac{(3^\ell-3)\theta^\ell}{4^{2m+|\mathpzc{C}|}} + \theta^{\ell-1},$   $\mu_\ell =\frac{(3^\ell-3)\theta^\ell}{4^{2m+|\mathpzc{C}|}},$  and $\nu_\ell =\frac{(3^\ell-3)\theta^\ell}{4^{2m+|\mathpzc{C}|}}.$   
    \end{itemize}
\end{thm}
\begin{proof}
To prove the result, we see, by Corollary \ref{SHI}, that when $q=4,$  $\mathpzc{B}=\mathpzc{A},$ and $|\mathpzc{A}|=m-1,$  the  code $\mathcal{C}_{\overline{\mathcal{N}}_2}$ is a projective $3$-weight code over $\mathbb{F}_4$ with parameters  $\big[\theta, {2m+|\mathpzc{C}|},\frac{\theta}{2}\big ]$ and has non-zero Hamming weights $\mathrm{w}_1 = \frac{\theta}{2},$ $\mathrm{w}_2 = \frac{3\theta}{4},$ and $\mathrm{w}_3 = \theta$ with frequencies $A_{\mathrm{w}_1} = 6,$ $A_{\mathrm{w}_2} = 4^{2m+|\mathpzc{C}|} - 16,$ and $A_{\mathrm{w}_3} = 9,$ respectively. 
Note that  $\mathrm{w}_1 + \mathrm{w}_2 + \mathrm{w}_3 = \frac{9\theta}{4}.$ By Lemma  \ref{Lemma5.3}, we see that $\Gamma_{\mathcal{C}_{\overline{\mathcal{N}}_2}^\perp}$ is a connected graph. Now, by Theorem \ref{Thm5.7} and applying Propositions 3.1 and 4.2 of Dam and Omidi \cite{Dam2013}, we get the desired result.\end{proof}
The parameters of the $\ell$-SWRGs constructed in the above theorem coincide with those obtained in Corollary 1 of Mondal and Lee\cite{Mondal2024}, upon setting $k = 2m$ in their result.

\end{document}